\documentclass[reqno, 11pt]{amsart}
\usepackage[utf8]{inputenc}
\usepackage{geometry}
\usepackage{amsfonts}
\usepackage{hyperref}
\usepackage[backend=biber, style=numeric-comp]{biblatex}
\hypersetup{
pdftitle={TBA},
pdfsubject={},
pdfauthor={Shijia Jin},
pdfkeywords={}
}

\usepackage{amsmath}
\usepackage{xcolor}
\usepackage{amssymb}
\usepackage{ulem}
\usepackage{amsthm}
\usepackage{pdflscape}
\usepackage{pgfplots}
\usepackage{mathrsfs}

\calclayout

\usepackage{comment}

\newtheorem{theorem}{Theorem}[section]

\newtheorem{proposition}[theorem]{Proposition}
\newtheorem{lemma}[theorem]{Lemma}

\newtheorem{assumption}[theorem]{Assumption}
\newtheorem{remark}[theorem]{Remark}

\theoremstyle{definition}
\newtheorem{definition}[theorem]{Definition}
\newtheorem{example}[theorem]{Example}

\usepackage{relsize}

\usepackage{mlmodern}
\usepackage[T1]{fontenc}

\newcommand{\w}{\mathfrak{w}}
\newcommand{\bdX}{\boldsymbol{X}}
\newcommand{\bdx}{\boldsymbol{x}}
\newcommand{\bdY}{\boldsymbol{Y}}
\newcommand{\bdy}{\boldsymbol{y}}

\newcommand{\m}{m}
\newcommand{\e}{e}

\newcommand{\kong}{\vspace{0.08cm}}

\title{Optimal Execution and Macroscopic Market Making}

\author{Ivan Guo}

\thanks{\hspace{0.12cm} Ivan Guo, Email: ivan.guo@monash.edu, Address: Centre for Quantitative Finance and Investment Strategies, School of Mathematics, Monash University, Wellington Rd, Clayton VIC 3800, Australia}

\author{Shijia Jin} 

\thanks{\hspace{0.12cm} Corresponding author: Shijia Jin, Email: shijia.jin@monash.edu, Address: School of Mathematics, Monash University, Wellington Rd, Clayton VIC 3800, Australia}

\begin{document}
\relscale{1.047}

\begin{abstract}
We propose a stochastic game modelling the strategic interaction between market makers and traders. From the trader's perspective, the conventional exogenous permanent price impact is replaced by the endogenous quoting strategies of the market makers. Conversely, from the market maker's perspective, order flows are no longer assumed to be exogenous, but are driven endogenously by the strategic traders. Characterizing the Nash equilibria via forward-backward stochastic differential equations (FBSDEs), we establish a local well-posedness result for the general game. For the specific `Almgren-Chriss-Avellaneda-Stoikov' model, the decoupling approach guarantees the global well-posedness of the FBSDEs by reducing it to a backward stochastic Riccati equation with $M_+$-matrix coefficients. Finally, by introducing small diffusion terms into the inventory processes as an approximation to the general game, we establish its global well-posedness. Simulations reveal a negative correlation between quotes and strategic orders, in contrast to the positive correlation observed between quotes and noise orders.\\


\end{abstract}

\maketitle
\setcounter{tocdepth}{1}
\tableofcontents

\section{Introduction}

\noindent From the liquidity-taking side, this paper formulates and solves a novel trade execution problem: how to design a trading schedule to liquidate a large position efficiently, when the trade impact is determined by market makers. We focus on a macroscopic time horizon (minutes to hours), so trading schedules can be represented by trading rates. From the liquidity-provision side, the paper also addresses a new market making problem: how to dynamically quote limit orders in the presence of competing market makers and strategic order flows. In summary, we propose a unified game-theoretic framework that links optimal execution and market making through strategic interactions. Because of the time scale, we follow the macroscopic framework of \cite{guo2023macroscopic} and \cite{guo2024macroscopicmarketmakinggames}, which adapts the seminal market-making model of \cite{avellaneda2008high} to a macroscopic setting. Both works study market making problems with noise traders only, excluding strategic traders. Consequently, our model emphasizes quote-driven markets and order-driven markets where the bid-ask spread is large relative to the tick size. See \cite{guo2023macroscopic} for further details on the macroscopic model.

Introduced by \cite{almgren2001optimal} and \cite{bertsimas1998optimal}, the optimal execution problem arises from the trade-off between trading slowly to reduce market friction costs and trading quickly to avoid adverse price fluctuations. The framework rests on two pillars: (1) a permanent price impact that embeds trading activity into price dynamics, and (2) a temporary price impact that captures instantaneous frictional costs, such as the cost of `walking the book' when submitting large orders over short intervals. Subsequent work has explored many variants of these impact functions. Nonlinear impact models are studied in \cite{almgren2003optimal}. Linear impact functions in game settings have attracted recent attention; see \cite{casgrain2018mean} and \cite{evangelista2020finite} for probabilistic approaches, and \cite{cardaliaguet2018mean} and \cite{huang2019mean} for PDE methods. Empirical evidence also points to transient components of permanent impact, which are examined in \cite{neuman2023trading} and \cite{schied2017state}.

For the optimal execution problem, our main contribution is the introduction of an \textit{endogenous permanent impact}. Permanent impact captures how trading activity affects prices, and prices in turn depend on how market makers requote their limit orders partially. Motivated by this feedback loop, we model a game between liquidity takers and liquidity providers, in which permanent impact is not imposed exogenously but emerges from market makers' quoting strategies. Specifically, if $g$ is some permanent impact functional and $\delta$ is a quoting strategy, we replace the classical term
\begin{equation*}
    P_t + g(v_{[0, t]}) \text{ \; by \; } P_t + \delta_t
\end{equation*}
in market-price modelling. Here, the process $P$ stands for a reference price and $v$ represents a trading schedule.

From the market-making perspective, this paper introduces a more comprehensive game framework among market makers by incorporating \textit{endogenous order flows}. Building on the literature initiated by \cite{ho1981optimal} and \cite{avellaneda2008high}, market making has been modelled as an optimization problem addressing two main issues: diminishing profit margins as transaction frequency increases, and rising inventory risk as positions deviate from zero. While \cite{campi2020optimal} models more realistic order flows via a hidden Markov chain and \cite{jusselin2021optimal} captures clustering and long-memory effects with general Hawkes processes, both treat order flow as exogenous. Recent game-theoretic studies (such as \cite{luo2021dynamic}, \cite{cont2022dynamics}, and \cite{guo2024macroscopicmarketmakinggames}) analyse how competition among market makers affects the quoting strategies, whereas the order flow continues to be modelled exogenously. Our contribution is to endogenize the order flow so that it interacts strategically with market makers within the game.

Our proposed game framework, illustrated in Figure \ref{figure}, optimizes the actions of all market makers and of strategic traders. Noise traders with exogenous objectives are included in the model but are not optimized; they form part of the stochastic environment. Treating these flows as exogenous allows us to derive order-flow-driven execution strategies such as the volume-weighted average price. See \cite{cartea2016incorporating} for related work on optimal execution with order flows.

\begin{figure}
    \centering
    \includegraphics[scale = 0.2]{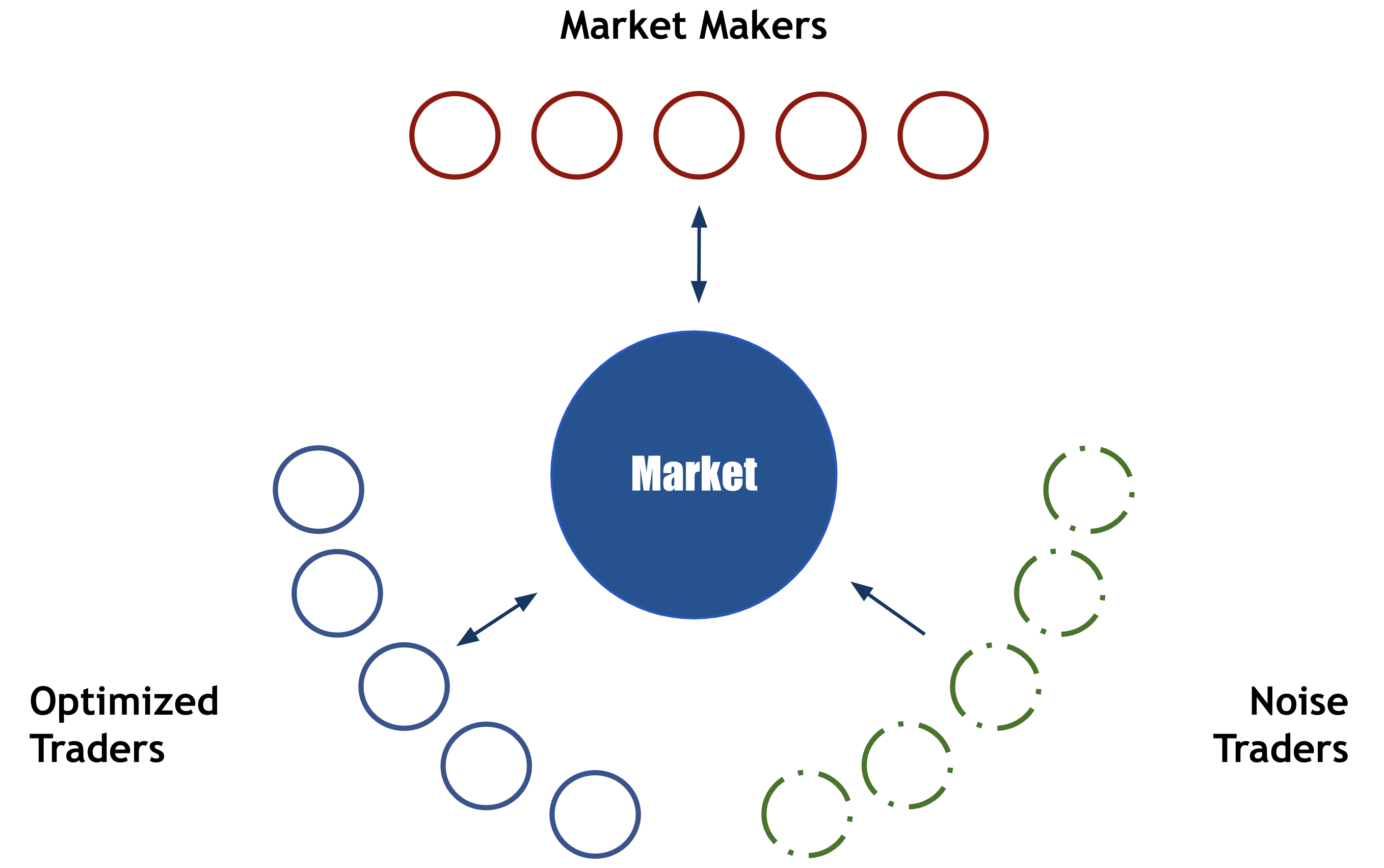}
    \caption{Three major components in the game}
    \label{figure}
\end{figure}

\kong

\textit{Outline:} In Section \ref{paper 3 section 2}, we begin with the optimal execution game, which sets up the framework for strategic traders. Moreover, while the existing literature predominantly employs individual temporary price impact, with the two-player analysis in \cite{voss2022two} as a notable exception, we adopt a linear aggregated temporary impact to bridge this gap. Our work can be viewed as a generalization of \cite{voss2022two} to an $N$-player environment (without tracking terms or constraints), and as an extension of execution games (e.g., \cite{evangelista2020finite}, \cite{drapeau2019fbsde}) to a setting with joint temporary impact and a more general formulation. Numerical examples show that such joint impact discourages the crowding of orders.

Continuing the optimal execution games, Section \ref{paper 3 section 3} incorporates the role of market makers. From the traders' perspective, the permanent price impact component is replaced by market makers’ quoting strategies: market prices are determined by those quotes. Simultaneously, we extend the previous linear specification and introduce a class of aggregate temporary impact functions. On the other hand, following the competition modelling of \cite{guo2024macroscopicmarketmakinggames}, each market maker competes against the best quotes of her rivals. Unlike previous market-making games, order flows here are endogenous, consisting of both noise traders and strategic traders. We apply the stochastic maximum principle to characterize Nash equilibria as a system of forward-backward stochastic differential equations (FBSDEs) and then establish a local well-posedness result.

In Section \ref{paper 3 section 4} we analyse the `Almgren-Chriss-Avellaneda-Stoikov' model under three specifications: (1) one strategic trader and two market makers, (2) a linear temporary impact for the trader as in \cite{almgren2001optimal}, and (3) exponential intensities for market-making competition as in \cite{avellaneda2008high}. Under the conditions above, the associated FBSDE system admits an explicit representation. Using the multidimensional decoupling method in \cite{guo2024macroscopicmarketmakinggames}, its global well-posedness reduces to the well-posedness of a backward stochastic Riccati equation (BSRE) with $M_+$-matrix coefficients. We establish a new well-posedness result in the deterministic case. Computing the equilibrium numerically, we observe the classical aggressive-then-taper pattern for the strategic trader. By contrast, market makers behave very differently: prices decline over the trading horizon, driven by the perfect-information assumption and the running inventory penalty.

Extending results beyond the `Almgren-Chriss-Avellaneda-Stoikov' model requires studying a broader class of BSREs that lie outside the scope of the existing literature. To address this difficulty, we introduce an approximation framework for the original game. Specifically, we add uncontrolled diffusion terms to the inventory dynamics of both traders and market makers, and we model noise trading rates by a Markov SDE. Furthermore, we also adopt the market-making competition style of \cite{luo2021dynamic} and \cite{cont2022dynamics}. The resulting FBSDE system is then a non-degenerate Markovian system, whose well-posedness can be derived by the result in \cite{nam2022coupled}. Simulations show that order flows play a role similar to that studied in \cite{cartea2016incorporating} for traders and in \cite{guo2023macroscopic} for market makers. Unlike the positive correlation we find between quotes and noise ask orders, we observe (and explain) a negative correlation between quotes and strategic orders.

The article is organized as discussed above. We start with optimal execution games in Section \ref{paper 3 section 2}. Section \ref{paper 3 section 3} adds the role of market makers, and the `Almgren-Chriss-Avellaneda-Stoikov' model is studied in Section \ref{paper 3 section 4}. Section \ref{paper 3 section 5} investigates the approximation game. All proofs are provided in the appendix.

\kong

\textit{Notation:} Throughout the present work, we fix $T > 0$ to represent the finite trading horizon. We denote by $(\Omega, \mathcal{F}, \mathbb{F} = (\mathcal{F}_t)_{0 \leq t \leq T}, \mathbb{P})$ a complete filtered probability space, with $\mathcal{F}_T = \mathcal{F}$. A $d$-dimensional Brownian motion $W$ is defined on such space, for a fixed positive integer $d$, and the filtration $\mathbb{F}$ is generated by $W$ and augmented. Let $\mathcal{G}$ represent an arbitrary $\sigma$-algebra contained in $\mathcal{F}$ and consider the following spaces:
\begin{align*}
    L^p(\Omega, \mathcal{G}) &:= \big\{X: X \text{ is } \mathcal{G} \text{-measurable and } \mathbb{E}|X|^p < \infty \big\};\\
    \mathbb{H}^p &:= \bigg\{ X: X \text{ is } \mathbb{F} \text{-progressively measurable and } \mathbb{E}\bigg[ \bigg( \int_0^T |X_t|^2 \, dt \bigg)^{p/2} \bigg] < \infty \, \bigg\};\\
    \mathbb{S}^p &:= \Big\{ X \in \mathbb{H}^p: \mathbb{E} \Big[ \sup_{0 \leq t \leq T} |X_t|^p \Big] < \infty \Big\};\\
    \mathbb{M} &:= \big\{ M : M_t \in L^2(\Omega, \mathcal{F}_t) \text{ for a.e. } t \in [0, T]\\ 
    &\hspace{4.15cm} \text{ and } \{M_t, \mathcal{F}_t\}_{0 \leq t \leq T}\text{ is a continuous martingale} \big\}.
\end{align*}

We use superscripts for enumerating purposes. For example, superscripts in states $Q^1, Q^2$ denote objects which are associated with player $1$ and player $2$ respectively. In particular, $Q^2$ is not to be confused with quadratic powers, which will be explicitly denoted with brackets like $(Q)^2$, or as $(Q^2)^2$ if necessary. For any vector-valued function $F$, the superscript in $F^i$ denotes the $i$-th entry. However, an exception to this superscript usage occurs when handling matrices, a context that will be evident. For matrices, the superscript `$*$' denotes the transpose. The symbol $C$ represents a positive constant that may vary across different lines. 

\vspace{0.2cm}

\section{Optimal Execution Game}
\label{paper 3 section 2}
\noindent This section is devoted to the $N$-player game of optimal execution, where the interaction appears in both linear permanent impact and linear temporary impact. While this section introduces the setup of traders who later participate in the more general game, it also fills a noteworthy gap in the literature by examining strategic interactions arising from aggregated temporary impact. Much of the optimal-execution game literature emphasises interactions only through permanent impact, with the notable exception of the two-player games studied in \cite{voss2022two}. Our work can be viewed as a generalisation of \cite{voss2022two} to an $N$-player setting (without tracking terms or constraints), and as an extension of execution games (e.g., \cite{evangelista2020finite}, \cite{drapeau2019fbsde}) to a setting with joint temporary impact and a more general formulation.

Let us consider a stochastic differential game comprising $N \in \mathbb{N}$ traders, trading a single financial asset. Let $\{P_t, \mathcal{F}_t\}_{t \in [0, T]}$---the reference price of the asset---be a square-integrable martingale. We index the players by $(i, \e)$ and the additional superscript $\e$ underlines the execution purpose. To trade $q_0^{i, \e} \in \mathbb{R}$ units of the asset, the inventory process $Q^{i, \e}$ of agent $(i, \e)$ reads
\begin{equation*}
    Q_t^{i,\e} = q_0^{i,\e} + \int_0^t v_s^{i,\e} \, ds,
\end{equation*}
where the trading rate $v^{i,\e} \in \mathbb{H}^2$ represents the control. The market price $S$ of the asset follows
\begin{equation}
    S_t = P_t + \int_0^t \bigg( \frac{\alpha_u}{N} \sum_{i=1}^{N} \, v_u^{i,\e} \bigg) \, du + \varpi_t \text{\; with \;}  \varpi_t := \int_0^t \alpha_u \, (a_u - b_u)\,du,
    \label{market price}
\end{equation}
where a bounded and positive process $\alpha \in \mathbb{S}^2$ represents the coefficient of the linear permanent price impact. Here, bounded and positive processes $a, b \in \mathbb{S}^2$ indicate respectively the ask and bid market order rates from the noise traders. Given the market price, the transaction price $\hat{S}^{i,\e}$ for agent $(i, \e)$ is given as
\begin{equation*}
    \hat{S}^{i, \e}_t = S_t + \frac{\beta_t}{N} \, \sum_{j = 1}^N v_t^{j, \e},
\end{equation*}
with a bounded and positive process $\beta \in \mathbb{S}^2$ specifying the coefficient of the linear temporary price impact.

\kong

\begin{remark}
    (1) The integral terms in \eqref{market price} capture the permanent price impact exerted by both strategic agents and noise traders. We defer the rationale for modelling the bid and ask market order rates ($a$ and $b$) as distinct processes, rather than relying solely on the order imbalance $a - b$, to the next section. From another perspective, the drift caused by the noise traders can also be interpreted as signals; see \cite{lehalle2019incorporating} for example.
    
    \indent (2) While market makers are not explicitly modelled as strategic agents in this game, their presence is implicit; notably, the evolution of the best quote is governed by the specified permanent price impact.

    (3) The subsequent discussion can also be applied to the case
    \begin{equation}
        \hat{S}^{i, \e}_t = S_t + \frac{\beta_t}{N} \, \Big( v_t^{i, \e} + \pi \sum_{j \neq i} v_t^{j, \e} \Big)
        \label{sect1 impact example}
    \end{equation}
    for some $\pi \in [0,2)$, at the cost of heavier notations. We will cover such cases in the next section.
\end{remark}

\kong

The agent $(i, \e)$ intends to maximize the objective functional
\begin{equation}
\begin{aligned}
    &J^{i, \e}(v^{i, \e}; v^{-i, \e})\\
    &:= \mathbb{E}\bigg[ - \int_0^T \hat{S}_t^{i, \e} \, v_t^{i, \e} \, dt + S_T \, Q_T^{i, \e} - \int_0^T \phi_t^{i, \e} \big( Q_t^{i, \e} \big)^2 \, dt - A^{i, \e} \big( Q_T^{i, \e} \big)^2 \bigg]\\
    & = \mathbb{E} \bigg[ - \int_0^T \Big(v_t^{i, \e} \, \frac{\beta_t}{N} \, \sum_{j = 1}^N v_t^{j, \e} \Big) \, dt - \int_0^T S_t \, dQ_t^{i, \e} + S_T \, Q_T^{i, \e} - \int_0^T \phi_t^{i, \e} \big( Q_t^{i, \e} \big)^2 \, dt - A^{i, \e} \big( Q_T^{i, \e} \big)^2 \bigg]\\
    &= P_0 \, q_0^{i, \e} + \mathbb{E} \bigg[ \int_0^T \Big( Q_t^{i, \e} \,  \frac{\alpha_t}{N} \, \sum_{j = 1}^{N} v_{t}^{j, \e} \Big) \, dt - \int_0^T \Big( v_t^{i, \e} \, \frac{\beta_t}{N} \, \sum_{j = 1}^N v_t^{j, \e} \Big) \,dt \\
    & \hspace{3cm} - \int_0^T v_t^{i, \e} \, \varpi_t \, dt + Q_T^{i, \e} \, \varpi_T - \int_0^T \phi_t^{i, \e} \big( Q_t^{i, \e} \big)^2 \, dt - A^{i, \e} \big( Q_T^{i, \e} \big)^2 \bigg].
    \label{obj fun sec 1}
\end{aligned}
\end{equation}
The above expression is obtained using integration by parts and the martingale property of $P$, as is standard in the optimal execution literature. Here, a positive bounded process $\phi^{i, \e} \in \mathbb{S}^2$ represents the running inventory penalty. The process $\phi^{i, \e}$ can be viewed as controlling the agent’s risk appetite, since large values of $\phi^{i, \e}$ impose a strong disincentive for agent $i$ to take on market exposure. This penalty can also be interpreted as the agent accounting for model uncertainty, as analysed in \cite{cartea2017algorithmic}. A positive bounded random variable $A^{i, \e} \in L^2(\Omega, \mathcal{F}_T)$ stands for the terminal liquidation penalty. We look for the Nash equilibrium in the following sense.

\kong

\begin{definition}
An admissible strategy profile $(\hat{v}^{k, \e})_{k = 1}^N \in \big(\mathbb{H}^2\big)^N$ is called
a Nash equilibrium if, for all $i$ and any admissible strategies $v \in \mathbb{H}^2$, it holds that
\begin{equation*}
    J^{i, \e}(v; \hat{v}^{-i, \e}) \leq J^{i, \e}(\hat{v}^{i, \e}; \hat{v}^{-i, \e}).
\end{equation*}
\end{definition}

\kong

The Hamiltonian of agent $i$ reads
\begin{equation*}
\begin{aligned}
    H^{i, \e}(t, q^{i, \e}, y^{i, \e}, v^{i, \e}; v^{-i, \e}) = v^{i, \e} \, y^{i, \e} &+ q^{i, \e} \, \frac{\alpha_t}{N} \, \sum_{j = 1}^N v^{j, \e} - v^{i, \e} \, \frac{\beta_t}{N} \sum_{j = 1}^N v^{j, \e} \\
    & - v^{i,\e} \, \varpi_t - \phi_t^{i, \e} \big( q^{i, \e} \big)^2.
\end{aligned}
\end{equation*}
The Hessian matrix of the function $(v^{i, \e}, q^{i, \e}) \mapsto H^{i, \e}(t, q^{i, \e}, y^{i, \e}, v^{i, \e}; v^{-i, \e})$ is then given by
\begin{equation*}
    \begin{bmatrix}
        -2 \, \cfrac{\beta_t}{N} & \cfrac{\alpha_t}{N} \\
        \cfrac{\alpha_t}{N} & -2 \, \phi_t^{i, \e}
    \end{bmatrix},
\end{equation*}
which is negative semi-definite if its determinant is non-negative. This motivates the following assumption that rules out abnormally large permanent impact coefficient.

\kong

\begin{assumption}
\label{paper 3 hamilt convex assumpt}
    The temporary impact coefficient $\beta$ is bounded away from zero in the sense that $\beta_t \geq C$ for some $C > 0$. Additionally, it holds for all $i, t$ that $ 4 \, N \beta_t \, \phi_t^{i, \e} \geq (\alpha_t)^2$. 
\end{assumption}

\kong

Given the concavity of $H$ with respect to $(q^{i, \e}, v^{i, \e})$, we are able to apply the stochastic maximum principle (see \cite{carmona2016lectures}, \cite{carmona2018probabilistic}) and hence turn to the Isaacs' condition: for any $i \in \{1, \dots, N\}$ and all $t \in [0, T]$, to maximize each Hamiltonian, the first-order condition yields
\begin{equation}
    v^{i, \e} = \frac{N}{2 \, \beta_t} \, y^{i, \e} + \frac{\alpha_t}{2 \, \beta_t} \, q^{i, \e} - \frac{N}{2 \, \beta_t} \, \varpi_t - \frac{1}{2} \sum_{k \neq i} v^{k, \e}.
    \label{paper 3 issac}
\end{equation}
Exploiting the linearity and the symmetry among agents, we solve \eqref{paper 3 issac} directly to obtain:
\begin{equation}
\begin{aligned}
    v^{i, \e} &= \frac{2N}{N + 1} \, \bigg( \frac{N}{2 \, \beta_t} \, y^{i, \e} + \frac{\alpha_t}{2 \, \beta_t} \, q^{i, \e} - \frac{N}{2 \, \beta_t} \, \varpi_t \bigg) - \frac{2}{N + 1} \, \sum_{k \neq i} \bigg(\frac{N}{2 \, \beta_t} \, y^{k, \e} + \frac{\alpha_t}{2 \, \beta_t} \, q^{k, \e} - \frac{N}{2 \, \beta_t} \, \varpi_t \bigg)\\
    & = \frac{N}{(N + 1) \, \beta_t} \, \Big(N \, y^{i, \e} - \sum_{k \neq i} y^{k, \e} \Big) + \frac{\alpha_t}{(N + 1) \, \beta_t} \, \Big(N \, q^{i,\e} - \sum_{k \neq i} q^{k, \e} \Big) - \frac{N}{(N + 1) \, \beta_t} \, \varpi_t.  
    \label{paper 3 solve issac}
\end{aligned}
\end{equation}
Introduce these $N$-by-$N$ matrices:
\begin{gather*}
    L := \begin{bmatrix}
        2A^{1,\e} & \cdots & 0\\
        \vdots & \ddots & \vdots\\
        0  & \cdots & 2A^{N, \e}
    \end{bmatrix},\\ 
    B:= \begin{bmatrix}
        N & -1 & \cdots & -1\\
        -1 & N & \cdots & -1\\
        \vdots & \vdots & \ddots & \vdots\\
        -1 & -1 & \cdots & N
    \end{bmatrix}, \hspace{0.5cm} 
    D_t : = \frac{N}{(N + 1) \, \beta_t} \,B, \hspace{0.5cm}
    E_t : = \frac{\alpha_t}{(N + 1) \, \beta_t} \, B,\\
    O :=\begin{bmatrix}
        1 & \cdots & 1\\
        \vdots & \ddots & \vdots\\
        1 & \cdots & 1
    \end{bmatrix}, \hspace{0.4cm}
    F_t : = \frac{- \alpha_t}{(N + 1) \, \beta_t} \, O, \hspace{0.4cm}
    G_t : = 
    \begin{bmatrix}
        2\phi_t^{1,\e} & \cdots & 0\\
        \vdots & \ddots & \vdots\\
        0  & \cdots & 2\phi_t^{N,\e}
    \end{bmatrix}
    - \frac{(\alpha_t)^2}{N(N+1) \, \beta_t} \,
    O.
\end{gather*}

\kong

\noindent Let us also write $\boldsymbol{Q}^\e_t := (Q_t^{1, \e}, \dots, Q_t^{N, \e})$, $\boldsymbol{Y}_t^\e := (Y_t^{1, \e}, \dots, Y_t^{N, \e})$, and $\w_t := (\varpi_t, \dots, \varpi_t)$. The Nash equilibrium has an FBSDE characterization as follows.

\kong

\begin{proposition}
\label{FBSDE char in Sec2}
    A strategy profile $(v^{i,\e})_{i = 1}^N \in (\mathbb{H}^2)^N$ forms a Nash equilibrium if and only if it admits the representation \eqref{paper 3 solve issac}, where $\boldsymbol{Q}^\e$ and $\boldsymbol{Y}^\e$, together with martingales $\boldsymbol{M}^\e$, solve the following FBSDE:
    \begin{equation}
        \left\{
        \begin{aligned}
        \;& d\boldsymbol{Q}^\e_t = D_t \, \boldsymbol{Y}^\e_t \, dt + E_t \, \boldsymbol{Q}^\e_t \, dt - \frac{N}{(N + 1) \, \beta_t} \, \w_t \, dt,\\
        & d\boldsymbol{Y}^\e_t = F_t \, \boldsymbol{Y}^\e_t \, dt + G_t \, \boldsymbol{Q}^\e_t \, dt + \frac{N \, \alpha_t}{(N + 1) \, \beta_t} \, \w_t \, dt + d\boldsymbol{M}_t^\e,\\
        & \boldsymbol{Q}^\e_0 = \boldsymbol{q}^\e_0, \quad \boldsymbol{Y}^\e_T = - L \, \boldsymbol{Q}^\e_T + \w_T.
        \end{aligned}
        \right.
    \label{paper 3 execut game FBSDE}
    \end{equation}
\end{proposition}

\kong

The equilibrium admits an explicit form in two simple cases: (1) there is no permanent impact; (2) agents are homogeneous in penalty parameters. In these cases, the adjoint process $\boldsymbol{Y}^\e$ can be represented by a linear function of $\boldsymbol{Q}^\e$, where the first-order coefficient matrix is determined by a BSRE.

\kong

\begin{theorem}
\label{paper 3 stochastic riccati approach}
    If either of the following holds:
    \kong
    \begin{itemize}
        \item[(1)] the permanent impact $\alpha_t = 0$;\\
        \vspace{-0.2cm}
        
        \item[(2)] penalty parameters are homogeneous in the sense that $\phi^{i,\e} = \phi$ and $A^{i,\e} = A$ for any $i$. In addition, it holds for all $t$ that $ 2(N+1) \, \beta_t \, \phi_t \geq (\alpha_t)^2$,
    \end{itemize}
    \kong
    then FBSDE \eqref{paper 3 execut game FBSDE} has a unique solution $(\boldsymbol{Q}^\e, \boldsymbol{Y}^\e, \boldsymbol{M}^\e) \in (\mathbb{S}^2 \times \mathbb{S}^2 \times \mathbb{M})^N$, which can be represented by
    \begin{equation*}
        \boldsymbol{Y}_t^\e = R_t \, \boldsymbol{Q}_t^\e + H_t
    \end{equation*}
    for some symmetric negative semi-definite $R_t$.
\end{theorem}

\kong

\begin{remark}
Throughout this article, when we say that an FBSDE system has a unique solution, as in Theorem \ref{paper 3 stochastic riccati approach}, we mean that every forward and backward process belong to $\mathbb{S}^2$, and every martingale term lies in $\mathbb{M}$.
\end{remark}

\kong

By imposing stricter conditions on the coefficients, the well-posedness of the equation \eqref{paper 3 execut game FBSDE} in the general setting can be derived by the method of continuation in \cite{peng1999fully}.

\begin{theorem}
\label{method of continuation}
If there exists $C > 0$ such that both of the following conditions hold:
\kong

\begin{itemize}
    \item  it holds for all $i$ that $A^{i,\e} > C$;\\
    \vspace{-0.2cm}

    \item  it holds for all $i, t$ that 
    \begin{equation}
        \phi_t^{i, \e} - \frac{(N + 1) \, (\alpha_t)^2}{8 N \, \beta_t} \geq C.
        \label{tech assumpt for cont}
    \end{equation}
\end{itemize}
    Then, system \eqref{paper 3 execut game FBSDE} has a unique solution. 
\end{theorem}

\kong

\noindent The only non-trivial condition is \eqref{tech assumpt for cont}. However, in practice the permanent impact coefficient $\alpha$ is much smaller than the temporary impact coefficient $\beta$. Hence, the condition essentially requires $\phi$ to be bounded below by a small constant.

\kong

\begin{figure}[h]
    \centering
    \includegraphics[width=0.8\textwidth]{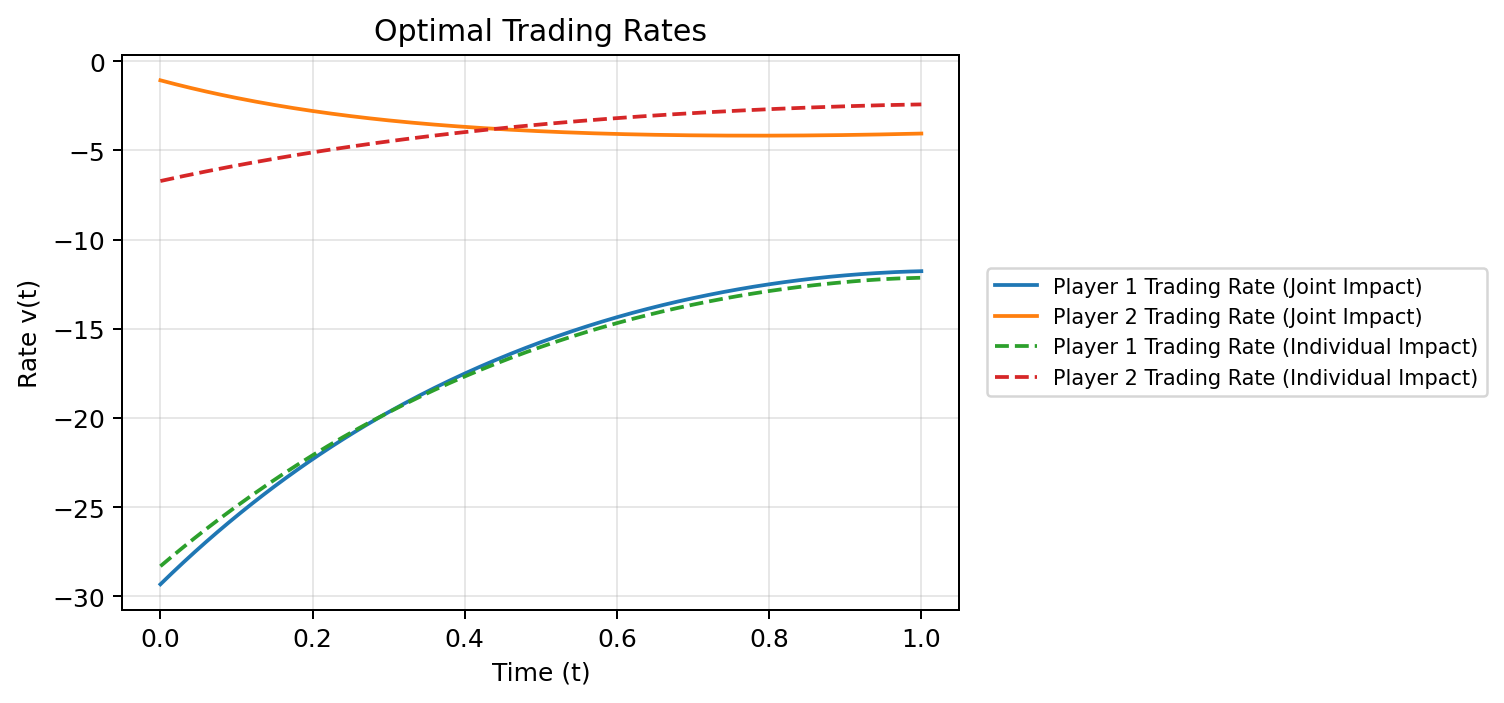}
    \caption{Equilibrium trading rates in the joint and individual impact cases}
    \label{sec 1 equili rate}
\end{figure}

\begin{example}
    In this example, we numerically study how joint linear temporary impact can notably alter the equilibrium trading strategies, even when players share homogeneous risk appetites and trade in the same direction. For parameters, we set $T = 1$, $\alpha_t = 0.01, \beta_t = 0.1$, and $a_t = b_t = 0$ for all $t$. We consider two players with initial inventories $q_0^{1, \e} = 18$ and $q_0^{2, \e} = 4$. They share the same running and terminal penalty parameters: $\phi_t = 0.1$ for all $t$ and $A = 1$. Because all coefficients are deterministic, the equation \eqref{paper 3 execut game FBSDE} reduces to a deterministic system, which we solve via the boundary-value problem solver in SciPy. In addition to the case with joint temporary impact, we also consider the individual temporary-impact case (i.e., setting $\pi = 0$ in \eqref{sect1 impact example}), as is common in the execution-game literature.

    As shown by the dashed lines in Figure \ref{sec 1 equili rate}, both players exhibit the classical equilibrium trading pattern (see \cite{evangelista2020finite} and \cite{drapeau2019fbsde}). They begin with an intensive liquidation phase due to the running penalties, and then gradually reduce their trading intensity to avoid incurring high individual temporary-impact costs. In the joint-impact case the pattern differs. As shown by the solid lines in Figure \ref{sec 1 equili rate}, player 2 begins with low liquidation rates and then gradually increases her trading intensity. This arises because player 1, who has the larger liquidation target, is comparatively more impatient and follows the classical aggressive-then-taper strategy. Player 2, being more patient, optimally avoids the initial crowded phase---where joint impact costs are high---and accelerates later when market activity subsides. This behaviour is an equilibrium response: player 1's slightly higher initial trading intensity (compared with the individual-impact case) makes delayed acceleration profitable for player 2. More broadly, the introduction of joint temporary impact smooths market activity. Figure \ref{sec 1 total rate} shows that the total trading rates of the two players are more stable than under the individual-impact case.
\end{example}

\begin{figure}[h]
    \centering
    \includegraphics[width=0.6\textwidth]{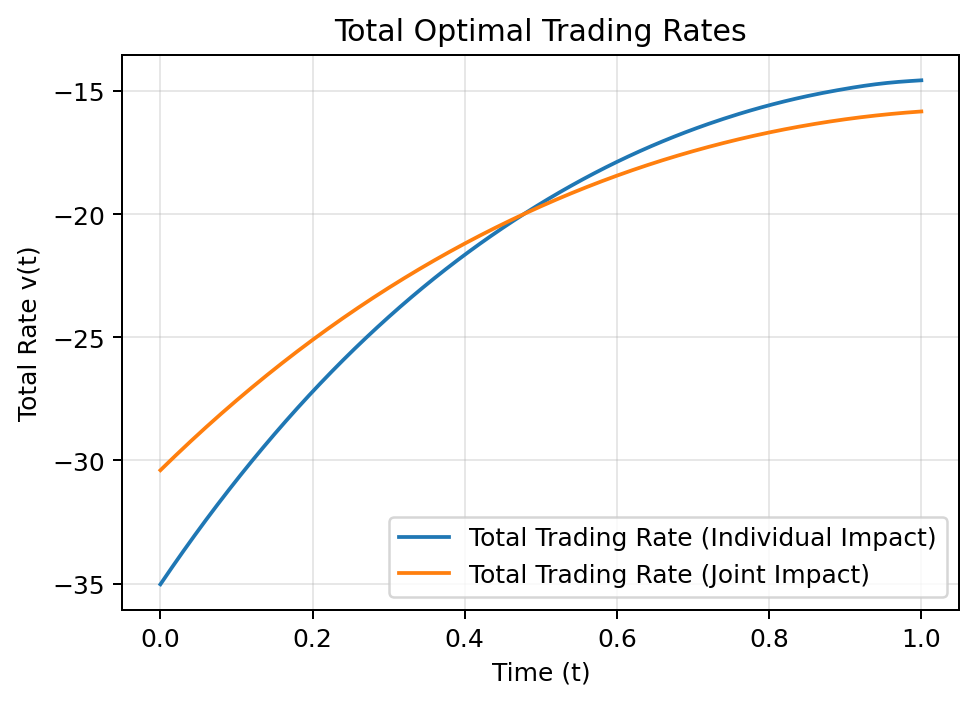}
    \caption{Equilibrium total rates in the joint and individual impact cases}
    \label{sec 1 total rate}
    
\end{figure}

\kong

\section{Connection with Market Making}
\label{paper 3 section 3}
\noindent Here, we study the connection between optimal execution problems and macroscopic market making models (see \cite{guo2023macroscopic} and \cite{guo2024macroscopicmarketmakinggames}). We propose a stochastic game that encompasses both liquidity takers and liquidity providers. This framework departs from the conventional approach of relying on exogenous components---such as permanent price impact functions and pre-determined order flows---commonly used in classical market making and optimal execution models. Instead, we incorporate genuine strategic behaviours.

Let us still consider $N$ traders indexed by $(i, \e)$. The agent $(i, \e)$ intends to trade $q_0^{i, \e} \in \mathbb{R} \backslash \{ 0 \}$ amount of assets. Now, only buy market orders are allowed when $q_0^{i, \e} < 0$; similarly only sell market orders are permitted for the case $q_0^{i, \e} > 0$.  Consequently, the admissible control space of agent $(i,\e)$ is defined as
\begin{equation*}
    \mathbb{A}^{i, \e} := \Big\{ v \in \mathbb{H}^2 : \, 0 \leq v_t \leq \tilde{\xi} \text{\, if \,} q_0^{i,\e} < 0 \text{\, and \,}  -\tilde{\xi} \leq v_t \leq 0 \text{\, if \,} q_0^{i,\e} > 0 \Big\}
\end{equation*}
for some constants $\tilde{\xi} > 0$. 

\kong

\begin{remark}
    (1) The upper bound $\tilde{\xi}$ in $\mathbb{A}^{i, \e}$ is adopted to ensure the Lipschitz property of the resulting FBSDE system. In the next section, we demonstrate how this term can be removed in specific cases.

    (2) For convenience, we impose the constraint that liquidators may only sell and acquirers may only buy. To relax this restriction, one can introduce a two-dimensional control for each trader---one component for buy orders and another for sell orders---allowing agents to submit both types of orders and to switch sides.
\end{remark}

\kong

\noindent Defining $\boldsymbol{v}^\e := (v^{1, \e}, \dots, v^{N, \e})$, the inventory and cash then follow
\begin{align*}
    Q_t^{i,\e} &= q_0^{i, \e} + \int_0^t v_u^{i, \e} \, du,\\
    X_t^{i, \e} &= - \int_0^t S_u^{i,\e} \, v_u^{i, \e} \, du - \int_0^t \lambda^i \big( \boldsymbol{v}^\e_u \big) \, du.
\end{align*}
Here, process $S^{i, \e}$ represents the market price in view of agent $(i, \e)$, to be specified later, and functions $(\lambda^i)_{i = 1}^N \in \Upsilon$ stand for the cost induced by the aggregate temporary price impact. We introduce the class $\Upsilon$ as follows:

\kong

\begin{definition}
\label{paper 3 temp price impact}
Functions $(f^i)_{i = 1}^N$ belong to the class of temporary price impact $\Upsilon$ if the following conditions hold for all $i$:

\kong

\begin{itemize}
    \item the function $f^i : \mathbb{R}^{N} \to \mathbb{R}$ is twice continuously differentiable;\\
    \vspace{-0.2cm}

    \item  the function $f^i$ is strongly convex in the $i$-th entry, i.e., there exists $C > 0$ such that $\partial^2 f^i / \partial (u^i)^2 \geq C$;\\
    \vspace{-0.2cm}

    \item the Hessian matrix of $f^i$ is diagonally dominant and further satisfies
    \begin{equation*}
        \inf_{\boldsymbol{u} \in \mathbb{R}^{N}} \bigg\{ \frac{\partial^2 f^i}{\partial (u^i)^2}( \boldsymbol{u}) - \sum_{j \neq i} \Big| \frac{\partial^2 f^i}{\partial u^i \partial u^j}( \boldsymbol{u}) \Big| \bigg\} > 0.
    \end{equation*}
\end{itemize}
\end{definition}

\kong

\begin{remark}
    For agent $(i, \e)$, the function $f^i$ returns her temporary impact cost induced by the trading activities of herself and the others. The following type of functions is in $\Upsilon$:
    \begin{equation*}
        f^i(\boldsymbol{v}) = v^i \, \big( v^i + g^i(v^{-i}) \big),
    \end{equation*}
    where function $g^i: \mathbb{R}^{N - 1} \to \mathbb{R}$ satisfies
    \begin{equation*}
        \sup_{\boldsymbol{u} \in \mathbb{R}^{N - 1}} \bigg\{ \sum_{j = 1}^{N - 1} \Big| \frac{\partial g^i}{\partial u^j}( \boldsymbol{u} ) \Big| \bigg\} < 2.
    \end{equation*}
    A typical example would be $g(\boldsymbol{u}) = \frac{1}{N-1}\sum_{j = 1}^{N - 1} u^j$, which is equivalent to \eqref{sect1 impact example} when $\pi = (N - 1)^{-1}$. Intuitively, this means that other players have a weaker effect on the temporary impact than the player under consideration. While in Section \ref{paper 3 section 2}, these effects are identical across all players. 
    
    Motivated by the power-law temporary impact studied in \cite{said2017market}, we consider the individual cost function of the form
    \begin{equation*}
        f^i(\boldsymbol{v})
        =
        \big((v^i)^2 + \varepsilon^2 \big)^{d/2}
        -\varepsilon^d + \kappa \, (v^i)^2,
    \end{equation*}
    where $d \in [1, 2]$ and $\kappa > 0$ is a regularization coefficient. To account for other players, we can add the term $\eta \, v^i \frac{1}{N - 1} \sum_{j \neq i} v^j$ for some positive $\eta < 2\kappa$.
    
    The final condition in Definition \ref{paper 3 temp price impact} is adopted for technical reasons. The intuition is that the joint influence of all the others on $f^i$ is dominated by the one of the $i$-th entry.
\end{remark}

\kong

Simultaneously, there are $\tilde{N}$ market makers indexed by $(i, \m)$ providing the liquidity to the asset. Here, the additional superscript $\m$ emphasizes the market making purpose. Given any quoting strategy $\boldsymbol{\delta}^i = (\delta^{i,a}, \delta^{i,b}) \in \mathbb{A}^\m \times \mathbb{A}^\m$ with 
\begin{equation*}
    \mathbb{A}^\m := \big\{ \delta \in \mathbb{H}^2 : \, |\delta_t| \leq \xi \, \text{ for all } t\in[0,T] \big\}
\end{equation*}    
for some constant $\xi > 0$, the market maker $(i,\m)$ offers price $P_t + \delta_t^{i,a}$ to sell the asset and $P_t - \delta_t^{i,b}$ to buy at time $t$. Due to the competition with other market makers, player $(i, m)$ captures only a fraction of both the ask and bid order flows. Following the competition model in \cite{guo2024macroscopicmarketmakinggames}, this fraction depends on the distances to the best quotes offered by others and also on a decreasing function $\Lambda \in \boldsymbol{\Lambda}$. We now define the class $\boldsymbol{\Lambda}$.

\begin{definition}[\cite{gueant2017optimal}]
\label{paper 3 general_inten}
A function $\Lambda:\mathbb{R}\to\mathbb{R_+}$ belongs to the class of intensity functions $\boldsymbol{\Lambda}$ if:

\vspace{0.1cm}

\begin{itemize}
    \item the function $\Lambda$ is twice continuously differentiable;\\
    \vspace{-0.2cm}
    
    \item the function $\Lambda$ is strictly decreasing and hence $\Lambda'(x)<0$ for any $x\in\mathbb{R}$;\\
    \vspace{-0.2cm}
    
    \item $\lim_{x \to \infty}\Lambda(x) = 0 \,$ and $\, -\infty < \inf_{x \in \mathbb{R}} \frac{\Lambda(x) \, \Lambda''(x)}{( \Lambda'(x) )^2} \leq \sup_{x \in \mathbb{R}} \frac{\Lambda(x) \, \Lambda''(x)}{( \Lambda'(x) )^2} \leq 1$.
\end{itemize}
\label{paper 3 inten_assu}
\end{definition}

\noindent Take the ask side as an example. The fraction obtained by agent $(i, \m)$ at time $t$ is given by $\Lambda(\delta_t^{i,a} - \bar{\delta}_t^{i,a})$, where $\bar{\delta}_t^{i,a} := \min_{k \neq i} \delta_t^{k,a}$ and its bid part is similarly defined. The difference $\delta_t^{i,a} - \bar{\delta}_t^{i,a}$ measures the distance between her quote and the best quote among the others, and the function $\Lambda$ maps this distance to the fraction of order flow she captures. The optimal control problem for a single market maker (as studied in \cite{guo2023macroscopic}, \cite{avellaneda2008high}, \cite{gueant2017optimal}) can be regarded as a special case, in which the best offers from the other participants are assumed to be the reference price throughout. Therefore, the inventory and cash of agent $(i, \m)$ read
\begin{align*}
    X_t^{i, \m} &= \int_0^t \big( P_u + \delta_u^{i,a} \big) \, \tilde{a}_u \, \Lambda \big(\delta_u^{i,a} - \bar{\delta}_u^{i,a} \big) \, du - \int_0^t \big( P_u - \delta_u^{i,b} \big) \, \tilde{b}_u \, \Lambda \big(\delta_u^{i,b} - \bar{\delta}_u^{i,b} \big) \, du,\\
    Q_t^{i, \m} &= q_0^{i, \m} - \int_0^t \tilde{a}_u \, \Lambda \big(\delta_u^{i,a} - \bar{\delta}_u^{i,a} \big) \, du + \int_0^t \tilde{b}_u \, \Lambda \big(\delta_u^{i,b} - \bar{\delta}_u^{i,b} \big) \, du.
\end{align*}
Here, total order flow processes $\tilde{a}, \tilde{b}$ are given by
\begin{equation}
\begin{aligned}
        \tilde{a}_t &= a_t + \sum_{i = 1}^N v_t^{i,\e} \, \mathbb{I}(q_0^{i,e} < 0),\\
        \tilde{b}_t &= b_t - \sum_{i = 1}^N v_t^{i,\e} \, \mathbb{I}(q_0^{i,e} > 0),
        \label{total order flow model}
\end{aligned}
\end{equation}
where $a, b \in \mathbb{S}^2$ are bounded positive processes describing the buying/selling rate from the noise traders as in the previous section.

\kong

\begin{remark}
    Since order flows are modelled continuously, the discrete Avellaneda-Stoikov framework of \cite{avellaneda2008high} is not suitable here. We therefore adopt its macroscopic counterpart developed in \cite{guo2023macroscopic}. For the competition among market makers, we follow the approach of \cite{guo2024macroscopicmarketmakinggames}, in which each market maker competes with the one offering the best price. Alternative formulations (in \cite{cont2022dynamics} and \cite{luo2021dynamic}) are discussed in a later section.\end{remark}
\kong

Concerning the objective functional, on the market making side, the agent $(i, \m)$ aims at maximizing the objective functional
\begin{equation*}
\begin{aligned}
    J&^{i, \m} ( \boldsymbol{\delta}^i ; \boldsymbol{\delta}^{-i}, \boldsymbol{v}^\e )\\
    :&= \mathbb{E} \bigg[ X_T^{i, \m} + P_T \, Q_T^{i, \m} - \int_0^T \phi_t^{i, \m} \, \big( Q_t^{i, \m} \big)^2 \, dt - A^{i, \m} \, \big( Q_T^{i, \m} \big)^2 \bigg]\\
     &= \mathbb{E} \bigg[\int_0^T \Big( \delta_t^{i,a} \, \tilde{a}_t \, \Lambda \big( \delta_t^{i,a} - \bar{\delta}_t^{i,a} \big) + \delta_t^{i,b} \, \tilde{b}_t \, \Lambda \big( \delta_t^{i,b} - \bar{\delta}_t^{i,b} \big) \Big) \, dt - \int_0^T \phi_t^{i, \m} \, \big( Q_t^{i, \m} \big)^2 \, dt - A^{i, \m} \, \big( Q_T^{i, \m} \big)^2 \bigg].
\end{aligned}
\end{equation*}
The functional $J^{i, \m}$ depends on $\boldsymbol{v}^{\e}$ through the aggregated order flows $(\tilde{a}, \tilde{b})$. The process $\phi^{i, \m}$ represents the running inventory penalty and the random variable $A^{i, \m}$ denotes the terminal penalty. For market makers, the running penalty captures the volatility risk associated with holding the asset, while the terminal penalty can be interpreted as the overnight risk. Here, the market maker values her end-of-day position at the reference price rather than the market price. This prevents illusory gains or losses generated by the temporary order-imbalance impact.

To connect the quoting strategies with the optimal-execution component, we make the following assumption.

\kong

\begin{assumption}
Throughout, the best bid and ask prices are offered by the $\tilde{N}$ market makers considered. Specifically, the best ask and bid prices are given by $P_t + \min_j\delta^{j, a}_t$ and $P_t - \min_j\delta^{j, b}_t$ respectively at time $t$.
\end{assumption}

\kong

\noindent Consequently, the market price perceived by trader $(i, \e)$ reads
\begin{equation}
    S_t^{i,\e} = P_t + \mathbb{I}(q_0^{i,e} < 0) \, \min_{j} \delta_t^{j, a} - \mathbb{I}(q_0^{i,e} > 0) \, \min_{j} \delta_t^{j, b}.
    \label{paper 3 different market price}
\end{equation}
The agent $(i,\e)$ aims at maximizing the objective functional analogous to \eqref{obj fun sec 1}:
\begin{equation*}
\begin{aligned}
    J^{i, \e} &( v^{i, \e} ; v^{-i, \e}, (\boldsymbol{\delta}^{j})_{j})\\
    :& = \mathbb{E}\bigg[ X_T^{i, \e} + P_T \, Q^{i, \e}_T - \int_0^T \phi_t^{i, \e} \, \big( Q_t^{i,\e} \big)^2 \, dt - A^{i, \e} \, \big( Q_T^{i, \e} \big)^2 \bigg]\\
    &= P_0 \, q_0^{i, \e} - \mathbb{E} \bigg[ \int_0^T \Big( \mathbb{I}(q_0^{i, \e} < 0) \min_{j} \delta^{j, a}_t - \mathbb{I}(q_0^{i,\e} > 0)  \min_{j} \delta^{j, b}_t \Big) \, v_t^{i,\e} \, dt \bigg]\\
    & \hspace{1.5 cm} - \mathbb{E} \bigg[ \int_0^T \lambda^i \big( \boldsymbol{v}_u^{\e} \big) \, du + \int_0^T \phi_t^{i, \e} \, \big( Q_t^{i, \e} \big)^2 \, dt + A^{i, \e} \, \big( Q_T^{i, \e} \big)^2 \bigg].
\end{aligned}
\end{equation*}

We seek the Nash equilibrium in the following sense:

\kong

\begin{definition}
An admissible strategy profile $\big( \hat{\boldsymbol{v}}^\e, (\hat{\boldsymbol{\delta}}^j)_{j} \big) \in \big( \Pi_{i = 1}^N \mathbb{A}^{i, \e} \big) \times (\mathbb{A}^\m \times \mathbb{A}^\m)^{\tilde{N}}$ is called
a Nash equilibrium if: (1) for all $k$ and any admissible strategies $\boldsymbol{\delta} \in \mathbb{A}^\m \times \mathbb{A}^\m$, it holds that
\begin{equation*}
    J^{k, \m} \big(\boldsymbol{\delta}; \hat{\boldsymbol{\delta}}^{-k}, \hat{\boldsymbol{v}}^\e \big) \leq J^{k, \m} \big( \hat{\boldsymbol{\delta}}^k; \hat{\boldsymbol{\delta}}^{-k}, \hat{\boldsymbol{v}}^\e \big);
\end{equation*}

\noindent (2) for all $k$ and any admissible strategies $v \in \mathbb{A}^{k, \e}$, it holds that
\begin{equation*}
    J^{k, \e} \big(v; \hat{v}^{-k, \e}, (\hat{\boldsymbol{\delta}}^{j})_{j} \big) \leq J^{k, \e} \big( \hat{v}^{k, \e}; \hat{v}^{-k, \e}, (\hat{\boldsymbol{\delta}}^{j})_{j} \big).
\end{equation*}
\end{definition}

\kong

\begin{remark}
In contrast to classical optimal-execution formulations, we substitute the permanent price impact term with the best price tactically offered by market makers. Unlike standard market-making models, the order flow in our framework is endogenous and depends on the strategic actions of traders rather than being exogenously specified. If a market maker does not post the best price, her orders can still be consumed with a smaller rate, causing the temporary price impact to executors and noise traders.
\end{remark}

\kong

\noindent The stochastic maximum principle yields an FBSDE characterization of the Nash equilibrium.

\kong

\begin{theorem}
\label{paper 3 big maximum principle}
    An admissible strategy profile $\big( \boldsymbol{v}^\e, (\boldsymbol{\delta}^j)_{j} \big)$ forms a Nash equilibrium if and only if
    \begin{equation}
        v^{j, \e}_t = \varphi^j(\boldsymbol{Y}_t^\e, \boldsymbol{Y}_t^\m), \quad \delta^{i, a}_t = \psi^{i, a}(\boldsymbol{Y}_t^\m), \text{\; and \;} \delta^{i, b}_t = \psi^{i, b}(\boldsymbol{Y}_t^\m)
        \label{general game feedback}
    \end{equation}
    for all $i, j$, where functions $\psi^a, \psi^b: \mathbb{R}^{\tilde{N}} \to \mathbb{R}^{\tilde{N}}$ and $\varphi: \mathbb{R}^{N} \times \mathbb{R}^{\tilde{N}} \to \mathbb{R}^N$ are some Lipschitz functions. The notation $\varphi^{j}$ denotes the $j$-th component of $\varphi$. Here, the adjoint processes \(\boldsymbol{Y}^\m\) and \(\boldsymbol{Y}^\e\), together with the inventory processes, solve the FBSDE system:
    \begin{equation}
    \left\{
    \begin{aligned}
     dQ_t^{i, \m} &= - \hat{a}_t \, \Lambda \big(\psi^{i,a}(\boldsymbol{Y}_t^\m) - \bar{\psi}^{i,a}(\boldsymbol{Y}_t^\m) \big ) \, dt + \hat{b}_t \, \Lambda \big(\psi^{i,b}(\boldsymbol{Y}_t^\m) - \bar{\psi}^{i,b}(\boldsymbol{Y}_t^\m) \big ) \, dt,\\
    \, dY_t^{i, \m} &= 2\phi_t^{i, \m} \, Q_t^{i, \m} \, dt + dM_t^{i, \m},\\
    Q_0^{i, \m} &= q_0^{i, \m}, \quad Y_T^{i, \m} = -2A^{i,\m} \, Q_T^{i, \m};
    \end{aligned}
    \right.
    \label{paper 3 big FBSDE market maker}
    \end{equation}
    \begin{equation}
        \left\{
    \begin{aligned}
     dQ_t^{j, \e} &= \varphi^j(\boldsymbol{Y}_t^\e, \boldsymbol{Y}_t^\m) \, dt,\\
    \, dY_t^{j, \e} &= 2\phi_t^{j, \e} \, Q_t^{j, \e} \, dt + dM_t^{j, \e},\\
    Q_0^{j, \e} &= q_0^{j, \e}, \quad Y_T^{j, \e} = -2A^{j, \e} \, Q_T^{j, \e},
    \end{aligned}
    \right.
    \label{paper 3 big FBSDE optimal execution}
\end{equation}
for any $i, j$, where $\bar{\psi}^{i, a}(\boldsymbol{y}) = \min_{k \neq i} \psi^{k, a}(\boldsymbol{y})$ and $\hat{a}_t, \hat{b}_t$ are defined as
\begin{equation*}
\begin{aligned}
    \hat{a}_t &= a_t + \sum_{i = 1}^N \varphi^i(\boldsymbol{Y}_t^\e, \boldsymbol{Y}_t^\m) \, \mathbb{I}(q_0^{i,e} < 0),\\
    \hat{b}_t &= b_t - \sum_{i = 1}^N \varphi^i(\boldsymbol{Y}_t^\e, \boldsymbol{Y}_t^\m) \, \mathbb{I}(q_0^{i,e} > 0).
\end{aligned}
\end{equation*}
\end{theorem}

\noindent Despite the complexity, the truncation coefficients $\xi, \tilde{\xi}$ guarantee the local well-posedness of \eqref{paper 3 big FBSDE market maker}-\eqref{paper 3 big FBSDE optimal execution}.

\kong

\begin{proposition}
    \label{local wellpose}
    If $T$ is sufficiently small, then system \eqref{paper 3 big FBSDE market maker}-\eqref{paper 3 big FBSDE optimal execution} has a unique solution.
\end{proposition}

\vspace{0.2cm}

\section{Almgren-Chriss-Avellaneda-Stoikov Model}
\label{paper 3 section 4}
\noindent Built upon the previous models, a particular example is presented in this section. In regard to the liquidity consuming part, we consider a single trader and set the linear temporary impact $\lambda^1(u) = \beta_t \, |u|^2$ as in the second section, for some bounded and positive $\beta \in \mathbb{S}^2$ that is also bounded away from $0$. Then, the cash account of agent $(1, \e)$ becomes
\begin{equation*}
    X_t^{1, \e} = - \int_0^t \big(S_u^{1, \e} + \beta_u \, v_u^{1, \e} \big) \, v_u^{1, \e} \, du.
\end{equation*}
Because the transaction price comprises the market price and a linear temporary impact term---similar to \cite{almgren2001optimal}, we consider this trader to be of the \textit{`Almgren-Chriss type.'}

Concerning the liquidity providing side, there are two market makers competing for the order flow under the exponential intensity $\Lambda(\delta) = \exp(-\gamma \, \delta)$ as in \cite{avellaneda2008high}, for some $\gamma > 0$. The inventory process of agent $(i, \m)$ reads
\begin{equation*}
    Q_t^{i, \m} = q_0^{i, \m} - \int_0^t \tilde{a}_u \, \exp\Big(-\gamma \big( \delta_u^{i,a} - \bar{\delta}_u^{i,a} \big) \Big) \, du + \int_0^t \tilde{b}_u \, \exp \Big(-\gamma \big(\delta_u^{i,b} - \bar{\delta}_u^{i,b} \big) \Big) \, du.
\end{equation*}
The work of \cite{guo2023macroscopic} can be viewed as a continuous analogue of \cite{avellaneda2008high}. Moreover, the market maker's control problem can be regarded as a special case when the best price offered by others is taken to be the reference price throughout. Therefore, we regard this market maker as being of the \textit{`Avellaneda-Stoikov type.'}

The individual quadratic temporary cost and the exponential intensity inspire the name \textit{`Almgren-Chriss-Avellaneda-Stoikov model.'} Besides, we consider the homogeneous penalty parameters for tractability:
\begin{equation*}
    \phi^{1,\e} = \phi^{i,\m} = \phi \text{ \; and \; } A^{1,\e} = A^{i,\m} = A
\end{equation*}
for all $i$. For convenience, let us set $q_0^{1,\e} < 0$ and $q_0^{1,\m} \geq q_0^{2,\m}$. An advantage of this model is that Lipschitz mappings $\psi^a, \psi^b$, and $\varphi$ in the implicit function theorem can be explicitly derived. We introduce the FBSDE characterization of the equilibrium as follows. The truncation coefficients $\xi$ and $\tilde{\xi}$ will be removed due to the boundedness of the solution.

\kong

\begin{proposition}
\label{paper 3 first character}
    For sufficiently large truncation levels $\xi$ and $\tilde{\xi}$, an admissible strategy profile $(v^{1, \e}, (\boldsymbol{\delta}^j)_{j = 1}^2)$ forms a Nash equilibrium if and only if
    \begin{equation}
        \delta_t^{i, a} = \frac{1}{\gamma} + Y_t^{i, \m}, \quad \delta_t^{i, b} = \frac{1}{\gamma} - Y_t^{i, \m}, \text{\; and \;} v_t^{1, \e} = \frac{1}{2 \, \beta_t} \, \bigg( Y_t^{1, \e} - Y_t^{1, \m} -\frac{1}{\gamma} \bigg) \vee 0,
    \label{paper 3 big game feedback}
\end{equation}
    where adjoint processes solve the FBSDE system
    \begin{gather}
    \left\{
    \begin{aligned}
     dQ_t^{1, \m} &= - \hat{a}_t \, \exp\big(-\gamma \, (Y_t^{1, \m} -  Y_t^{2, \m}) \big) \, dt + b_t \, \exp\big(-\gamma \, (Y_t^{2, \m} - Y_t^{1, \m}) \big) \, dt,\\
    \, dY_t^{1, \m} &= 2 \phi_t \, Q_t^{1, \m} \, dt + dM_t^{1, \m},\\
    Q_0^{1, m} &= q_0^{1, \m}, \quad Y_T^{1, \m} = -2A \, Q_T^{1, \m};
    \end{aligned}
    \right.
    \label{paper 3 big FBSDE 1}\\
    \left\{
    \begin{aligned}
     dQ_t^{2, \m} &= - \hat{a}_t \, \exp\big(-\gamma \, (Y_t^{2, \m} -  Y_t^{1, \m}) \big) \, dt + b_t \, \exp\big(-\gamma \, (Y_t^{1, \m} - Y_t^{2, \m}) \big) \, dt,\\
    \, dY_t^{2, \m} &= 2 \phi_t \, Q_t^{2, \m} \, dt + dM_t^{2, \m},\\
    Q_0^{2, m} &= q_0^{2, \m}, \quad Y_T^{2, \m} = -2A \, Q_T^{2, \m};
    \end{aligned}
    \right.\\
    \left\{
    \begin{aligned}
     dQ_t^{1, \e} &= \frac{1}{2 \, \beta_t} \, \bigg(Y_t^{1, \e} - Y_t^{1, \m} - \frac{1}{\gamma} \bigg) \vee 0 \,dt,\\
    \, dY_t^{1, \e} &= 2\phi_t \, Q_t^{1, \e} \, dt + dM_t^{1, \e},\\
    Q_0^{1, \e} &= q_0^{1, \e}, \quad Y_T^{1, \e} = -2A \, Q_T^{1, \e}.
    \end{aligned}
    \right.
    \label{paper 3 big FBSDE 2}
\end{gather}
Here, the total ask order rate $\hat{a}$ is defined as 
\begin{equation*}
    \hat{a}_t = a_t + \frac{1}{2 \, \beta_t} \, \bigg(Y_t^{1, \e} - Y_t^{1, \m} -\frac{1}{\gamma} \bigg) \vee 0.
\end{equation*}
\end{proposition}

\kong

\noindent It is worth noting that only the adjoint process of market maker $(1, \m)$ appears in the system \eqref{paper 3 big FBSDE 2} of the strategic trader. This is because market maker $(1, \m)$ provides the best ask quote among the two market makers.\\

To solve system \eqref{paper 3 big FBSDE 1}--\eqref{paper 3 big FBSDE 2}, note that the adjoint processes \(Y^{1, \e},Y^{1, \m},Y^{2, \m}\) enter only through their differences and that all players share identical risk parameters. We are motivated to introduce the following notations:
\begin{equation*}
\begin{aligned}
    \begin{bmatrix}
        \mathcal{Q}^{1,\m} \\
        \mathcal{Q}^{2, \m} \\
        \mathcal{Q}^{1, \e} \\
    \end{bmatrix}
    :=
    \begin{bmatrix}
        1 & 0 & 0\\
        1 & -1 & 0\\
        -1 & 0 & 1\\
    \end{bmatrix}
    \begin{bmatrix}
        Q^{1,\m} \\
        Q^{2, \m} \\
        Q^{1, \e} \\
    \end{bmatrix} \text{\; and \;}
    \begin{bmatrix}
        \mathcal{Y}^{1,\m} \\
        \mathcal{Y}^{2, \m} \\
        \mathcal{Y}^{1, \e} \\
    \end{bmatrix}
    :=
     \begin{bmatrix}
        1 & 0 & 0\\
        1 & -1 & 0\\
        -1 & 0 & 1\\
    \end{bmatrix}
    \begin{bmatrix}
        Y^{1,\m} \\
        Y^{2, \m} \\
        Y^{1, \e} \\
    \end{bmatrix}.
\end{aligned}
\end{equation*}
By direct calculations, we can see $(\mathcal{Q}^{1, \e}, \mathcal{Y}^{1, \e}, \mathcal{M}^{1, \e})$ and $(\mathcal{Q}^{i, \m}, \mathcal{Y}^{i, \m}, \mathcal{M}^{i, \m})_{i = 1}^2$ satisfy another FBSDE system:
\begin{equation}
\begin{aligned}
    d\mathcal{Q}_t^{1, \m} &= - \bigg( a_t +\frac{1}{2 \, \beta_t} \, \Big( \mathcal{Y}_t^{1, \e} - \frac{1}{\gamma} \Big) \vee 0 \bigg) \, \exp(-\gamma \, \mathcal{Y}_t^{2, \m}) \, dt + b_t \, \exp(\gamma \, \mathcal{Y}_t^{2, \m}) \, dt,\\
     d\mathcal{Q}_t^{2, \m} &= \bigg( a_t + b_t + \frac{1}{2 \, \beta_t} \, \Big( \mathcal{Y}_t^{1, \e} - \frac{1}{\gamma} \Big) \vee 0 \bigg) \, \Big( \exp(\gamma \, \mathcal{Y}_t^{2, \m}) - \exp(-\gamma \, \mathcal{Y}_t^{2, \m}) \Big) \, dt,\\
    d\mathcal{Q}_t^{1, \e} &= \bigg( \frac{1}{2 \, \beta_t} \, \Big( \mathcal{Y}_t^{1, \e} - \frac{1}{\gamma} \Big) \vee 0 \bigg) \, \Big( 1 + \exp( - \gamma \, \mathcal{Y}_t^{2, \m}) \Big) \, dt \\
    &\hspace{4.5cm} + a_t \, \exp( - \gamma \, \mathcal{Y}_t^{2,\m}) \, dt - b_t \, \exp( \gamma \, \mathcal{Y}_t^{2,\m}) \, dt,
\label{paper 3 equiv big FBSDE}
\end{aligned}
\end{equation}
with $\mathcal{Q}_0^{1, \m} = q_0^{1, \m}, \mathcal{Q}_0^{2, \m} = q_0^{1, \m} - q_0^{2, \m} \geq 0$, and $\mathcal{Q}_0^{1, \e} = q_0^{1, \e} - q_0^{1, \m}$. Again, we omit the backward equations because they are more straightforward. On account of the non-singularity of the transformation matrix, systems \eqref{paper 3 big FBSDE 1}-\eqref{paper 3 big FBSDE 2} and \eqref{paper 3 equiv big FBSDE} are equivalent from the perspective of well-posedness. Further, since we already know that $(Y^{1, \m}, Y^{2, \m}, Y^{1, \e})$ are bounded, the same is true for $(\mathcal{Y}^{1, \m}, \mathcal{Y}^{2, \m}, \mathcal{Y}^{1, \e})$. We thus regard \eqref{paper 3 equiv big FBSDE} as a Lipschitz FBSDE system. Let us introduce the following matrix definition.

\kong

\begin{definition}
    The class of \textit{$Z$-matrices} are those matrices whose off-diagonal entries are less than or equal to zero. An \textit{$M_+$-matrix} is a $Z$-matrix whose row (or column) sums are non-negative.
\end{definition}

The definition of a $Z$-matrix is standard. An $M_+$-matrix is also an $M$-matrix (see Definition 2.5.2 in \cite{horn1994topics}), which can be regarded as a generalization of positive definite matrices with non-positive off-diagonals. Our study of $M_+$-matrices also extends the discussion of $M_0$-matrices in \cite{guo2024macroscopicmarketmakinggames}. 

The following theorem transfers the well-posedness of the FBSDE system \eqref{paper 3 equiv big FBSDE} to the one of BSRE. In one-dimensional cases, such an equation is known as the \textit{characteristic BSDE} introduced by \cite{ma2015well}. The proof is based on the multidimensional decoupling approach for FBSDEs. Although there are several general treatments (e.g., \cite{fromm2013existence} and \cite{ankirchner2020optimal}), we refer the reader to the short review in Section 5 of \cite{guo2024macroscopicmarketmakinggames}, on which our proof relies heavily. Briefly, the decoupling approach states that a Lipschitz relation between the backward and forward variables, with a non-exploding Lipschitz coefficient, ensures global well-posedness of the associated FBSDE.

\kong

\begin{proposition}
\label{Riccati with M matrix}
    The system \eqref{paper 3 big FBSDE 1}-\eqref{paper 3 big FBSDE 2} has a unique solution if the BSRE 
    \begin{equation}
        dR_t = \big( 2\phi_t \, I - R_t \, G_t \, R_t \big) \, dt + dM_t, \quad R_T = - 2A \, I,
        \label{paper 3 M+ BSRE}
    \end{equation}
    has a unique bounded solution $R$, where $I \in \mathbb{R}^{2 \times 2}$ denotes the identity matrix. Here, the adapted matrix-valued process $G$ is continuous, bounded, and is an $M_+$-matrix for all $t$.
\end{proposition}

\kong

In the deterministic setting---when order flows and penalty parameters are non-random---the following theorem presents the corresponding global well-posedness result.

\kong

\begin{theorem}
\label{global well-posed M+}
If the order flows $(a, b)$ and penalty parameters $(\phi, A)$ are deterministic, then BSRE \eqref{paper 3 M+ BSRE} has a unique bounded solution, such that $R_t$ has non-positive entries and is diagonally dominant for all $t$.
\end{theorem}

\begin{figure}
    \centering
    \includegraphics[width=.98\linewidth]{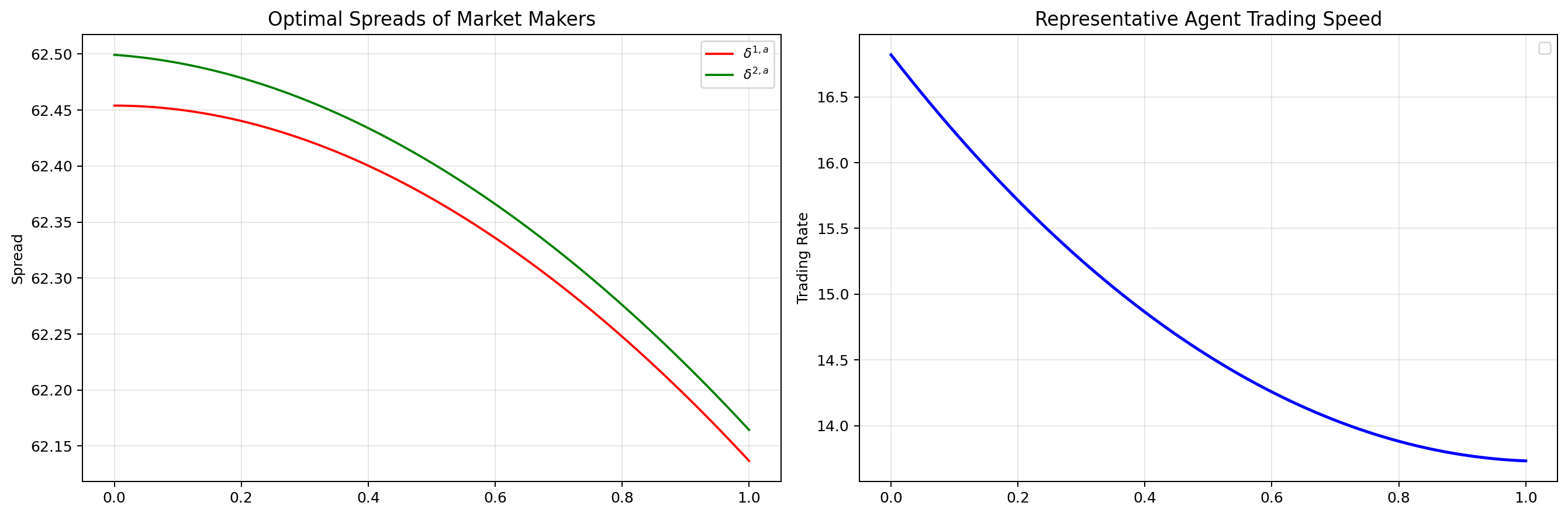}
    \caption{Equilibrium strategies of the trader and market makers}
    \label{fig equil strat}
\end{figure}

\begin{example}
\label{sec 4 example}
    We numerically study the equilibrium strategies of one trader and two market makers. For parameters, we set $T = 1$, $\beta_t = 0.1$, and $a_t = b_t = 0$ for all $t$. We consider a single trader with initial inventory $q_0^{1, \e} = -31$. On the other hand, there are two market makers providing the liquidity for the trader. Their initial inventories are $q_0^{1, \m} = 0$ and $q_0^{2, \m} = 1$. As the coefficient in function $\Lambda(\delta) = \exp(-\gamma \, \delta)$, we set $\gamma = 1$. Finally, all three agents share the same running and terminal penalty parameters: $\phi_t = 0.02$ for all $t$ and $A = 2$. Because equations \eqref{paper 3 big FBSDE 1}-\eqref{paper 3 big FBSDE 2} are deterministic, we still solve it using a boundary-value problem solver.
    
    Figure \ref{fig equil strat} shows the equilibrium strategy profile. It exhibits an ordering property and the distance between the two quotes shrinks over time, as discussed in \cite{guo2024macroscopicmarketmakinggames}. However, the price is monotone decreasing even though the order imbalance created by the trader increases. To understand this, first note that the game is of perfect information and that all market orders are submitted by the strategic trader. Since $ \delta^{1,a} = \gamma^{-1} + Y_t^{1,\m} $, we can interpret $Y^{1, \m}$ as a risk factor. In the deterministic setting, we have
    \begin{equation}
        Y_t^{1,\m} = -2A\, Q_T^{1,\m} - \int_t^T 2\phi_s\, Q_s^{1,\m} \, ds.
        \label{adjoint form}
    \end{equation}
    Therefore, the risk factor decreases over time as the remaining running penalty diminishes, leading to a decline in the price. Moreover, the concavity of the price path is driven by the monotone increase in the market maker's inventory. Although the price declines throughout the trading window, the terminal best ask price is power-law increasing with concavity in respect to both the trader's absolute initial inventory and the actual volume traded, as illustrated in Figure \ref{fig price impact}.

\begin{figure}
    \centering
    \includegraphics[width=0.64\linewidth]{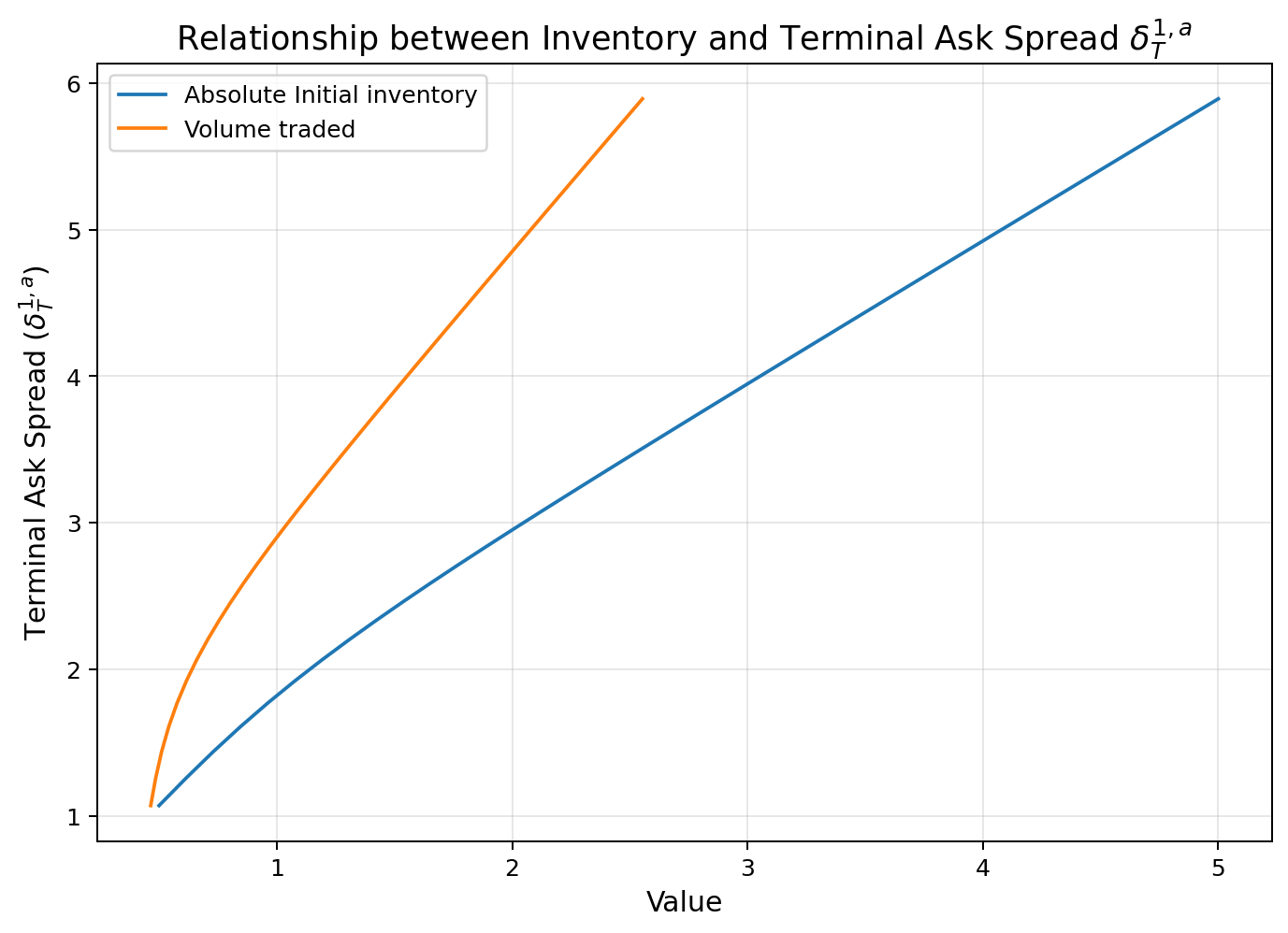}
    \caption{Relationship between volume and terminal ask spread $\delta_T^{1,a}$}
    \label{fig price impact}
\end{figure}

    From an economic standpoint, we may interpret it as follows: the market maker initially posts a high price, knowing that the trader’s urgency level (see the next paragraph) is greatest at the outset. As this urgency declines, she must quote increasingly competitive prices to induce additional trading.
    
    The trader’s strategy follows the classical convex shape familiar from the optimal-execution literature. Since $Y^{1, \e}$ has a similar form to \eqref{adjoint form}, we regard it as an urgency level. Let us compute that 
    \begin{equation*}
        dv_t^{1, \e} = \frac{1}{2 \, \beta_t} \, \big( dY_t^{1, \e} - dY_t^{1, \m} \big) = \frac{\phi_t}{\beta_t} \, \big(Q_t^{1, \e} - Q_t^{1, \m} \big) \, dt.
    \end{equation*}
    The convexity and monotonic decline of the trading rate follow from the facts that $Q^{1, \e}$ is increasing, $Q^{1, \m}$ is decreasing, and $Q_t^{1, \e} \leq Q_t^{1, \m}$ as observed. Given that the terminal penalty is the same and the volume is approximately transferred from the trader to the market maker, it is not profitable for the market maker to hold a more negative position than the trader. In other words, the trader's urgency level decreases faster than the market maker's risk factor.

    Finally, we examine the trader's objective functional values (negative trading cost minus inventory penalties) as the market makers' initial inventories vary. According to Figure \ref{fig obj value}, the contour lines approximately follow the family of lines $q_0^{1, \m} + q_0^{2, \m} = C$, $C \in \mathbb{R}$. As these values increase toward the upper-right, the outcome becomes more favourable for the trader when the market holds greater long positions, since the market makers are then more inclined to sell. Moreover, for fixed $q_0^{1, \m}$ the objective value increases with $q_0^{2, \m}$, reflecting intensified market-making competition and thus better prices, consistent with \cite{guo2024macroscopicmarketmakinggames}.

\end{example}

\begin{figure}
    \centering
    \includegraphics[width=0.6\linewidth]{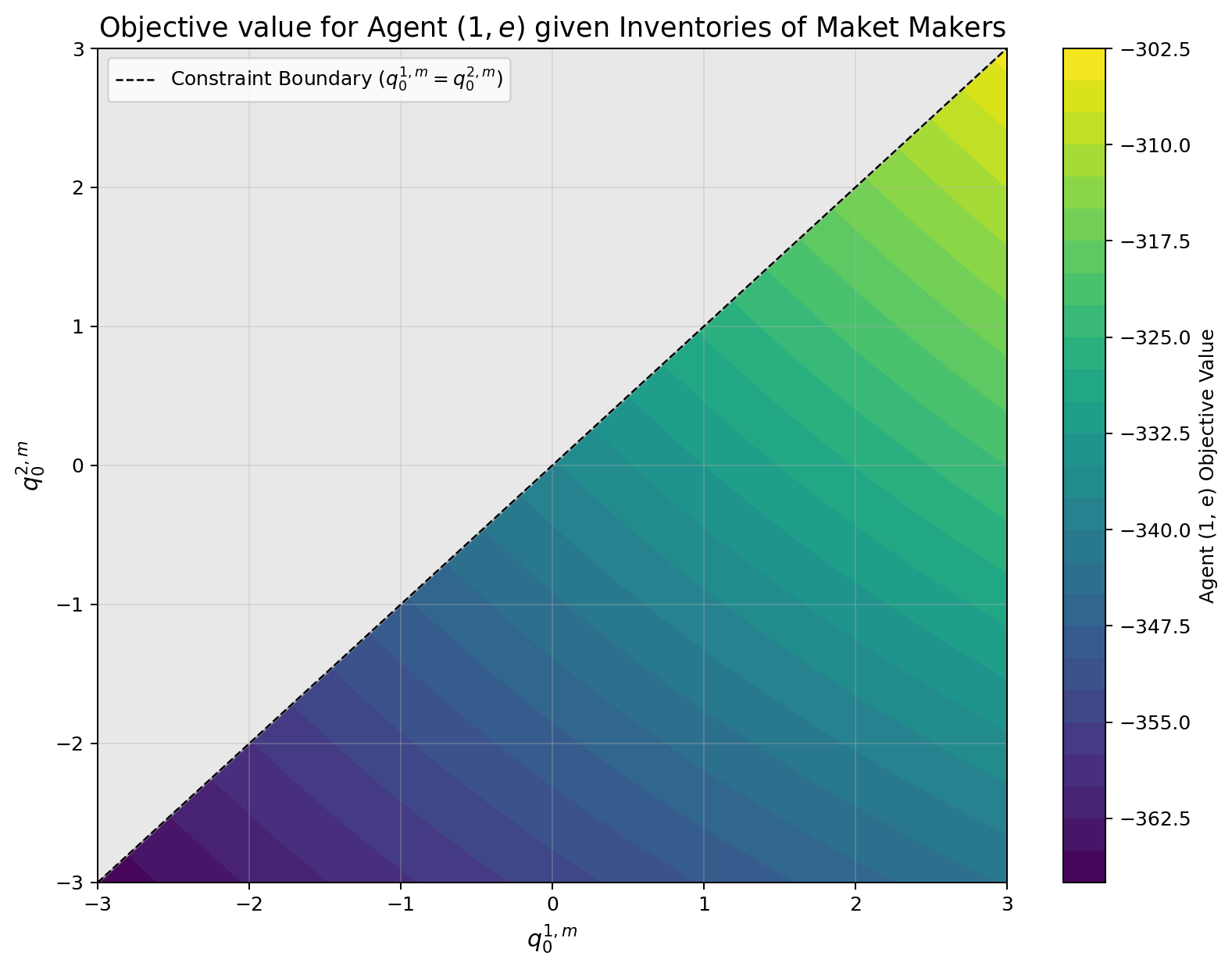}
    \caption{Trader’s objective functional values under various initial conditions}
    \label{fig obj value}
\end{figure}

\vspace{0.2cm}

\section{Approximation Game with Diffusive Inventories}
\label{paper 3 section 5}
\noindent The challenge in the general stochastic game between executors and market makers lies in the multidimensional and highly nonlinear system \eqref{paper 3 big FBSDE market maker}-\eqref{paper 3 big FBSDE optimal execution}. Even in the Almgren-Chriss-Avellaneda-Stoikov model, when considering multiple traders and market makers, the BSRE characterizing the well-posedness of the FBSDE system is beyond the scope of current literature. With this motivation, we introduce an approximate game to circumvent this issue. Specifically, we add uncontrolled and independent diffusion terms to the inventory processes. Economically, the presence of a non-zero Brownian component in inventories is supported by \cite{carmona2023optimal}. Mathematically, introducing such diffusion terms renders the FBSDE system non-degenerate, a setting that is comparatively better understood.

On the other hand, the market-making competition discussed in the preceding sections adheres to the style in \cite{guo2024macroscopicmarketmakinggames}. Competition among market makers within stochastic differential games has been rarely studied, with two notable recent exceptions \cite{cont2022dynamics} and \cite{luo2021dynamic}. Their approach will be considered in this section. Let us start by the following definition. For any $i$, function $\zeta^i$ returns the portion of the order flow captured by market maker $(i, \m)$.

\kong

\begin{definition}
 A set of functions $(\zeta^i)_{i = 1}^{\tilde{N}}$ belongs to the class $\beth$ if the following holds for all $i$:   
\begin{itemize}
    \item function $\zeta^i: \mathbb{R}^{\tilde{N}} \to \mathbb{R}_+$ is twice continuously differentiable;\\
    \vspace{-0.2cm}
    
    \item for all $(\delta^1, \dots, \delta^{\tilde{N}}) \in \mathbb{R}^{\tilde{N}}$, it holds that
    \begin{equation*}
        \frac{\partial \zeta^i}{\partial \delta^i} < 0, \quad \frac{\partial \zeta^i}{\partial \delta^j} \geq 0, \quad \text{and \;} \zeta^i \, \frac{\partial^2 \zeta^i}{\partial (\delta^i)^2} \leq 2 \, \Big( \frac{\partial \zeta^i}{\partial \delta^i} \Big)^2
    \end{equation*}
    for any $j \neq i$;\\
    \vspace{-0.2cm}
    
    \item for all $(\delta^1, \dots, \delta^{\tilde{N}}) \in \mathbb{R}^{\tilde{N}}$, it holds that
    \begin{equation}
        \Big(\frac{\partial \zeta^i}{\partial \delta^{i}} \Big)^{-2} \, \bigg[ 2 \, \Big( \frac{\partial \zeta^i}{\partial \delta^{i}} \Big)^2 - \zeta^i \, \frac{\partial^2 \zeta^i}{\partial (\delta^{i})^2} - \sum_{k \neq i} \Big| \frac{\partial \zeta^i}{\partial \delta^{i}} \, \frac{\partial \zeta^i}{\partial \delta^{k}} - \zeta^i \, \frac{\partial^2 \zeta^i}{\partial \delta^{i} \partial \delta^{k}} \Big|  \bigg] \geq C
        \label{paper 3 market maker growth cond}
    \end{equation}
    for some $C > 0$.
\end{itemize}
\label{paper 3 inten_assu extend}
\end{definition}

\kong

\begin{remark}
    Compared with the discussion of Section \ref{paper 3 section 3}, the Definition \eqref{paper 3 inten_assu extend}---borrowed and slightly revised from \cite{luo2021dynamic} and \cite{cont2022dynamics}---essentially introduces a different description of market makers' competition. Indeed, in the previous section the agent competes with the one who provides the best price throughout the time, while the above definition specifies that the agent competes with all other agents in a pre-determined manner.
    
    The intuition follows similarly: the portion decreases with respect to the gap of the agent $i$, but increases with respect to the gaps of her competitors. Inequality \eqref{paper 3 market maker growth cond} is a technical condition for the implicit function theorem. 
    
    The first example would be
    \begin{equation}
        \zeta^i(\delta^1, \dots, \delta^{\tilde{N}}) = \frac{e^{-\varsigma \, \delta^i + o^i}}{\sum_{k = 1}^{\tilde{N}} e^{-\varsigma \, \delta^k + o^k}}
        \label{exponential comp}
    \end{equation}
    for constants $\varsigma>0$ and $(o^i)_{i=1}^{\tilde{N}} \in \mathbb{R}^{\tilde{N}}$, which further satisfies the market clearance condition since $\sum_{j = 1}^{\tilde{N}} \zeta^j = 1$. More flexible examples can be
    \begin{equation*}
        \zeta^i(\delta^1, \dots, \delta^{\tilde{N}}) = \Lambda^i(\delta^i) \, \mathcal{G}^i(\delta^{-i}).
    \end{equation*}
    Here, function $\Lambda^i: \mathbb{R} \to \mathbb{R}_+$ is a twice continuously differentiable decreasing function, such that 
    \begin{equation*}
        \frac{\Lambda^i \, \frac{\partial^2 \Lambda^i}{\partial \delta^2}}{\big( \frac{\partial \Lambda^i}{\partial \delta} \big)^2} \leq 2 - \epsilon
    \end{equation*}
    for some $\epsilon > 0$, and function $\mathcal{G}^i: \mathbb{R}^{\tilde{N}- 1} \to \mathbb{R}_+$ is twice continuously differentiable and increasing in all variables, satisfying $\inf_{\boldsymbol{\delta}}\mathcal{G}^i(\boldsymbol{\delta})>0$.

    Although we focus on the competition specified by $\beth$ here, the following discussion still applies to the competition style used in the previous sections.
    
\end{remark}

\kong

\noindent Consider $\tilde{N}$ market makers indexed by $(i, \m)$. Given any admissible strategy $\boldsymbol{\delta}^i \in \mathbb{A}^\m \times \mathbb{A}^\m$ as before, let us write $\boldsymbol{\delta}^a:=(\delta^{1, a}, \dots, \delta^{\tilde{N}, a})$ and $\boldsymbol{\delta}^b$ is similarly defined. The inventory and cash of market maker $(i, \m)$ read
\begin{align*}
    X_t^{i, \m} &= \int_0^t \big( P_u + \delta_u^{i,a} \big) \, \tilde{a}_u \, \zeta^i \big( \boldsymbol{\delta}^a_u \big) \, du - \int_0^t \big( P_u - \delta_u^{i,b} \big) \, \tilde{b}_u \, \zeta^i \big( \boldsymbol{\delta}^b_u \big) \, du,\\
    Q_t^{i, \m} &= q_0^{i, \m} - \int_0^t \tilde{a}_u \, \zeta^i \big( \boldsymbol{\delta}^a_u \big) \, du + \int_0^t \tilde{b}_u \, \zeta^i \big( \boldsymbol{\delta}^b_u \big) \, du  + \epsilon \, W_t^{i, \m},
    \nonumber   
\end{align*}
where $\epsilon > 0$ is some constant and $\vec{W}^\m := (W^{1, \m}, \dots, W^{\tilde{N}, \m})$ consists of $\tilde{N}$ independent Brownian motions. Processes $\tilde{a}, \tilde{b}$ will be specified later on.

\kong

\begin{remark}
    We introduce uncontrolled and independent diffusion terms into the inventory dynamics. Financially, the presence of a non-zero Brownian component in inventories is documented in \cite{carmona2023optimal}. We include an independent Brownian perturbation to capture exogenous inventory shocks from uncontrollable external order flow and other exogenous factors. Mathematically, adding these diffusion terms renders the FBSDE non-degenerate and connects our analysis to a broader literature, notably the seminal work \cite{delarue2002existence}.
\end{remark}

\kong

\noindent The player $(i, \m)$ aims at maximizing the analogous objective functional
\begin{equation*}
\begin{aligned}
    &J^{i, \m} ( \boldsymbol{\delta}^i ; \boldsymbol{\delta}^{-i}, \boldsymbol{v}^\e )\\
    :&=\mathbb{E} \bigg[ X_T^{i, \m} + P_T \, Q_T^{i, \m} - \int_0^T \phi_t^{i, \m} \, \big( Q_t^{i, \m} \big)^2 \, dt - A^{i, \m} \, \big( Q_T^{i, \m} \big) ^2 \bigg]\\
     &= P_0 \,q_0^{i, \m} + \epsilon \, \mathbb{E}\big[ P_T \, W_T^{i, \m} \big]\\
     & \qquad + \mathbb{E} \bigg[ \int_0^T \Big( \delta_t^{i, a} \, \tilde{a}_t \, \zeta^i \big( \boldsymbol{\delta}^a_t \big) + \delta_t^{i, b} \, \tilde{b}_t \, \zeta^i \big( \boldsymbol{\delta}^b_t \big) \Big) \, dt - \int_0^T \phi_t^{i, \m} \, \big( Q_t^{i, \m} \big)^2 \, dt - A^{i, \m} \, \big( Q_T^{i, \m} \big)^2 \bigg].
\end{aligned}
\end{equation*}

The market also includes $N$ traders indexed by $(i, \e)$. Intending to trade $q_0^{i, \e}$ amount of assets, the inventory and cash of agent $(i, \e)$ now follow 
\begin{align*}
    Q_t^{i,\e} &= q_0^{i, \e} + \int_0^t v_u^{i, \e} \, du + \epsilon \, W_t^{i,\e},\\
    X_t^{i, \e} &= -\int_0^t \Big( S_u^{i,\e} \, v_u^{i,e} + \lambda^i(\boldsymbol{v}_u^{\e}) \Big) \, du,
\end{align*}
where $v^{i, \e}\in\mathbb{A}^{i,\e}$ represents the trading strategy, $S^{i,\e}$ indicates the market price in view of agent $(i,\e)$---see \eqref{paper 3 different market price}, and $(\lambda^i)_{i=1}^N \in \Upsilon$ denotes the joint temporary impact cost. Here, process $\vec{W}^\e := (W^{1, \e}, \dots, W^{N, \e})$ consists of $N$ independent Brownian motions that are also independent of $\Vec{W}^\m$. The agent $(i, \e)$ aims at maximizing the functional
\begin{equation*}
\begin{aligned}
    &J^{i, \e}(v^i ; v^{-i}, (\boldsymbol{\delta}^{j})_{j})\\
    :& = \mathbb{E}\bigg[ X_T^{i, \e} + P_T \, Q^{i, \e}_T - \int_0^T \phi_t^{i, \e} \, \big( Q_t^{i, \e} \big)^2 \, dt - A^{i, \e} \, \big( Q_T^{i, \e} \big)^2 \bigg]\\
    & = P_0 \, q_0^{i, \e} + \epsilon \, \mathbb{E} \big[ P_T \, W_T^{i, \e} \big]\\
    &\hspace{1.5cm} - \mathbb{E} \bigg[ \int_0^T \Big( \mathbb{I}(q_0^{i,\e} < 0) \, \min_{j} \delta^{j, a}_t - \mathbb{I}(q_0^{i, \e} > 0) \, \min_{j} \delta^{j, b}_t \Big) \, v_t^{i, \e} \, dt \bigg]\\
    & \hspace{2.4cm} - \mathbb{E} \bigg[  \int_0^T \lambda^i(\boldsymbol{v}_t^{\e}) \, dt + \int_0^T \phi_t^{i, \e} \, \big( Q_t^{i,\e} \big)^2 \, dt + A^{i, \e} \, \big(Q_T^{i, \e} \big)^2 \bigg].
\end{aligned}
\end{equation*}

Finally, total market order rates $\tilde{a}$ and $\tilde{b}$ are defined as
\begin{equation*}
\begin{aligned}
        \tilde{a}_t &= \kappa^{a}(L_t) + \sum_{i = 1}^N v_t^{i,\e} \, \mathbb{I}(q_0^{i,e} < 0),\\
        \tilde{b}_t &= \kappa^{b}(L_t) - \sum_{i = 1}^N v_t^{i,\e} \, \mathbb{I}(q_0^{i,e} > 0),
\end{aligned}
\end{equation*}
where $L$ is the (unique) solution of the stochastic differential equation
\begin{equation}
    dL_t = \Gamma(t, L_t) \, dt + \Sigma(t, L_t) \, d\vec{W}_t^0, \qquad L_0 = l_0 \in \mathbb{R}^{\hat{N}},
    \label{paper 3 noise sde}
\end{equation}
for continuous functions $\Gamma : [0, T] \times \mathbb{R}^{\hat{N}} \to \mathbb{R}^{\hat{N}}$ and $\Sigma : [0, T] \times \mathbb{R}^{\hat{N}} \to \mathbb{R}^{{\hat{N}} \times {\hat{N}}}$ subject to the assumptions below. Here, process $\vec{W}^0$ is a ${\hat{N}}$-dimensional Brownian motion independent of both $\Vec{W}^\e$ and $\Vec{W}^\m$. Properties of functions $\kappa^a, \kappa^b: \mathbb{R}^{\hat{N}} \to \mathbb{R}_+$ will be specified in the following assumption. Moreover, the original Brownian motion $W$ is now defined as $W := (\Vec{W}^0, \Vec{W}^\e, \Vec{W}^\m)$. Compared with the earlier order flow modelling, processes $a_t$ and $b_t$ are replaced by $\kappa^a(L_t)$ and $\kappa^b(L_t)$ respectively. The process $L$ serves as the driving process that determines the order flows. Mathematically, this modification ensures a Markovian structure.

\kong

\begin{assumption}
(1) Functions $\Gamma$ and $\Sigma$ also satisfy the following conditions:
\kong
\begin{itemize}
\item for $G\in\{ \Gamma, \Sigma \}$, there exists a constant $C > 0$ such that
\begin{equation*}
    |G(t,x)-G(t,y)| \leq C \, |x - y|;
\end{equation*}
for all $t\in[0,T]$ and any $x, y \in \mathbb{R}^{\hat{N}}$;\\
\vspace{-0.2cm}

\item function $\Sigma$ is uniformly elliptic in the sense that, for any $t \in [0, T]$ and $x\in \mathbb{R}^{\hat{N}}$, it holds that $(\Sigma \Sigma^*)(t, x) \geq C \, I_{\hat{N}}$ in the sense of symmetric matrices, where $C > 0$ is some constant and $I_{\hat{N}}$ is the ${\hat{N}}$-dimensional identity matrix;\\
\vspace{-0.2cm}

\item there exists $C > 0$ such that $|\Gamma(t, x)| \leq C$ for all $t$ and $x$.
\end{itemize}

\kong

(2) If $G \in \{ \kappa^a, \kappa^b \}$, then $G$ is a positive-valued function for which there exists a constant $C > 0$ such that
\begin{equation*}
    |G(x) - G(y)| \leq C \, |x-y| \text{ \; and \; } |G(x)| \leq C
\end{equation*}
for all $x, y \in \mathbb{R}^{\hat{N}}$.\\
\vspace{-0.2cm}

(3) In this section we assume all penalty parameters---namely the $\phi$s and $A$s---are deterministic. One may generalize by taking the penalty parameters to be functions of $L$,  at the expense of heavier notation.
\label{paper 3 sde_assumption}
\end{assumption}

\kong

Our goal is to find the Nash equilibrium. Although diffusion terms have been added to the inventory dynamics, they are not controlled and in fact have no influence on any agent's Hamiltonian. Hence, the Nash equilibrium can be similarly characterized by an FBSDE system.

\kong

\begin{proposition}
\label{sec 5 sto max prin}
    An admissible strategy profile $(\boldsymbol{v}^\e, (\boldsymbol{\delta}^j)_{j}) \in \big( \Pi_{i = 1}^N \mathbb{A}^{i, \e} \big) \times (\mathbb{A}^\m \times \mathbb{A}^\m)^{\tilde{N}}$ forms a Nash equilibrium if and only if
    \begin{equation*}
        v^{j, \e}_t = \varphi^j(\boldsymbol{Y}_t^\e, \boldsymbol{Y}_t^\m), \quad \delta^{i,a}_t = \psi^{i, a}(\boldsymbol{Y}_t^\m), \text{\quad and \quad} \delta^{i, b}_t = \psi^{i, b}(\boldsymbol{Y}_t^\m)
    \end{equation*}
    for all $i$ and $j$, where $\psi^a, \psi^b: \mathbb{R}^{\tilde{N}} \to \mathbb{R}^{\tilde{N}}$ and $\varphi: \mathbb{R}^{N} \times \mathbb{R}^{\tilde{N}} \to \mathbb{R}^N$ are some Lipschitz functions. Here, adjoint processes $\boldsymbol{Y}^\m$ and $\boldsymbol{Y}^\e$ solve the FBSDE system
    \begin{equation}
    \left\{
    \begin{aligned}
     dQ_t^{i, \m} &= - \hat{a}_t \, \zeta^i \big(\psi^{a}(\boldsymbol{Y}_t^\m) \big) \, dt + \hat{b}_t \, \zeta^i \big( \psi^{b}(\boldsymbol{Y}_t^\m) \big) \, dt + \epsilon \, dW_t^{i, \m},\\
    \, dY_t^{i, \m} &= 2\phi_t^{i, \m} \, Q_t^{i, \m} \, dt + dM_t^{i, \m},\\
    Q_0^{i, \m} &= q_0^{i, \m}, \quad Y_T^{i, \m} = -2A^{i, \m} \, Q_T^{i, \m};
    \end{aligned}
    \right.
    \label{paper 3 big FBSDE market maker approx}
    \end{equation}
    \begin{equation}
        \left\{
    \begin{aligned}
     dQ_t^{j, \e} &= \varphi^j(\boldsymbol{Y}_t^\e, \boldsymbol{Y}_t^\m) \, dt + \epsilon \, dW_t^{j, \e},\\
    \, dY_t^{j, \e} &= 2\phi_t^{j, \e} \, Q_t^{j, \e} \, dt + dM_t^{j, \e},\\
    Q_0^{j, \e} &= q_0^{j, \e}, \quad Y_T^{j, \e} = -2A^{j,\e} \, Q_T^{j, \e},
    \end{aligned}
    \right.
    \label{paper 3 big FBSDE optimal execution approx}
\end{equation}
for all $i$ and $j$, where order flows $\hat{a}, \hat{b}$ are defined as
\begin{equation*}
\begin{aligned}
    \hat{a}_t &= \kappa^a(L_t) + \sum_{i = 1}^N \varphi^i(\boldsymbol{Y}_t^\e, \boldsymbol{Y}_t^\m) \, \mathbb{I}(q_0^{i,e} < 0),\\
    \hat{b}_t &= \kappa^b(L_t) - \sum_{i = 1}^N \varphi^i(\boldsymbol{Y}_t^\e, \boldsymbol{Y}_t^\m) \, \mathbb{I}(q_0^{i,e} > 0).
\end{aligned}
\end{equation*}
\end{proposition}

\kong

Although the FBSDE system \eqref{paper 3 big FBSDE market maker approx}-\eqref{paper 3 big FBSDE optimal execution approx} is non-Markovian due to the noise $L$, the extended system \eqref{paper 3 noise sde}-\eqref{paper 3 big FBSDE optimal execution approx} is Markovian if we regard \eqref{paper 3 noise sde} as a $\hat{N}$-dimensional trivial FBSDE. Moreover, the extended system is non-degenerate because of the additional diffusion terms. The combination of non-degeneracy and Markovian nature enables us to apply tools from a more extensive literature.

\kong

\begin{theorem}
\label{sec 5 well pose}
    System \eqref{paper 3 noise sde}-\eqref{paper 3 big FBSDE optimal execution approx} has a unique solution.
\end{theorem}

\kong

\begin{example}
    Numerically, we study the influences of noise order flows on the equilibrium strategies of a single trader and two market makers. For parameters, we set $T = 1$ and $\epsilon = 10^{-3}$. We consider a single trader with initial inventory $q_0^{1, \e} = -5$ and $\beta_t = 0.1$ as the coefficient of the linear temporary impact. On the other hand, the initial inventories of two market makers are $q_0^{1, \m} = 0$ and $q_0^{2, \m} = -0.2$. Regarding the competition, we consider a simple case of \eqref{exponential comp}:
    \begin{equation*}
        \zeta^1(\delta^1, \delta^2) = \frac{\exp(-\delta^1)}{\exp(-\delta^1) + \exp(-\delta^2)}.
    \end{equation*}
    The running and terminal penalty parameters of three players are $\phi^{1, \e}_t = 0.04$, $A^{1, \e}=4$ and $\phi^{1, \m}_t = \phi_t^{2, \m}=0.01$, $A^{1, \m} = A^{2, \m} = 1$.

    In terms of order flows, we model $L$ as a one-dimensional market sentiment process, which satisfies a bounded mean-reverting SDE:
    \begin{equation*}
        dL_t = -2 \, \text{tanh}(L_t) \, dt + 0.2 \, dW_t^0; \quad L_0 = 0.
    \end{equation*}
    Since $L_0 = 0$, the market sentiment is zero on average. For functions $\kappa^a$ and $\kappa^b$, we adopt logistic-styles functions:
    \begin{equation*}
        \kappa^a(\ell) = 1 + \frac{8}{1 + \exp(-\ell)} \quad \text{and} \quad \kappa^b(\ell) = 1 + \frac{8}{1 + \exp(\ell)}.
    \end{equation*}
    The first constant term denotes the base trading rate, while the second term captures the additional rate induced by sentiment; higher sentiment increases the ask order rate and decreases the bid rate.
    
    To solve the multidimensional FBSDE, we implement a Deep BSDE solver, which effectively circumvents the curse of dimensionality associated with traditional grid-based methods. Following the architecture proposed by \cite{han2018solving}, the solver reformulates the FBSDE as a stochastic control problem, where the initial value of the adjoint process $Y_0$ and the control process $(Z_t)_{t \in [0, T]}$ are approximated by multi-layer perceptrons. The forward inventory process $Q$ and the backward adjoint process $Y$ are evolved using an Euler-Maruyama discretization over the interval $[0, T]$. The network parameters $\theta$ are calibrated by minimizing a loss function defined by the mean squared error between the simulated terminal value $Y_T^{(\theta)}$ and terminal condition prescribed in the FBSDE.  We simulate $10^4$ equilibria. Figure \ref{sec 5 simulations} shows a few simulated trading rates and spreads of market maker $(1, \m)$, who posts a lower ask price, together with their mean paths across all simulations. Note that the mean strategies match the results in the previous section; see Figure \ref{fig equil strat}. Our goal is to understand the deviation of each path from its mean.

    \begin{figure}
        \centering
        \includegraphics[width=0.495\linewidth]{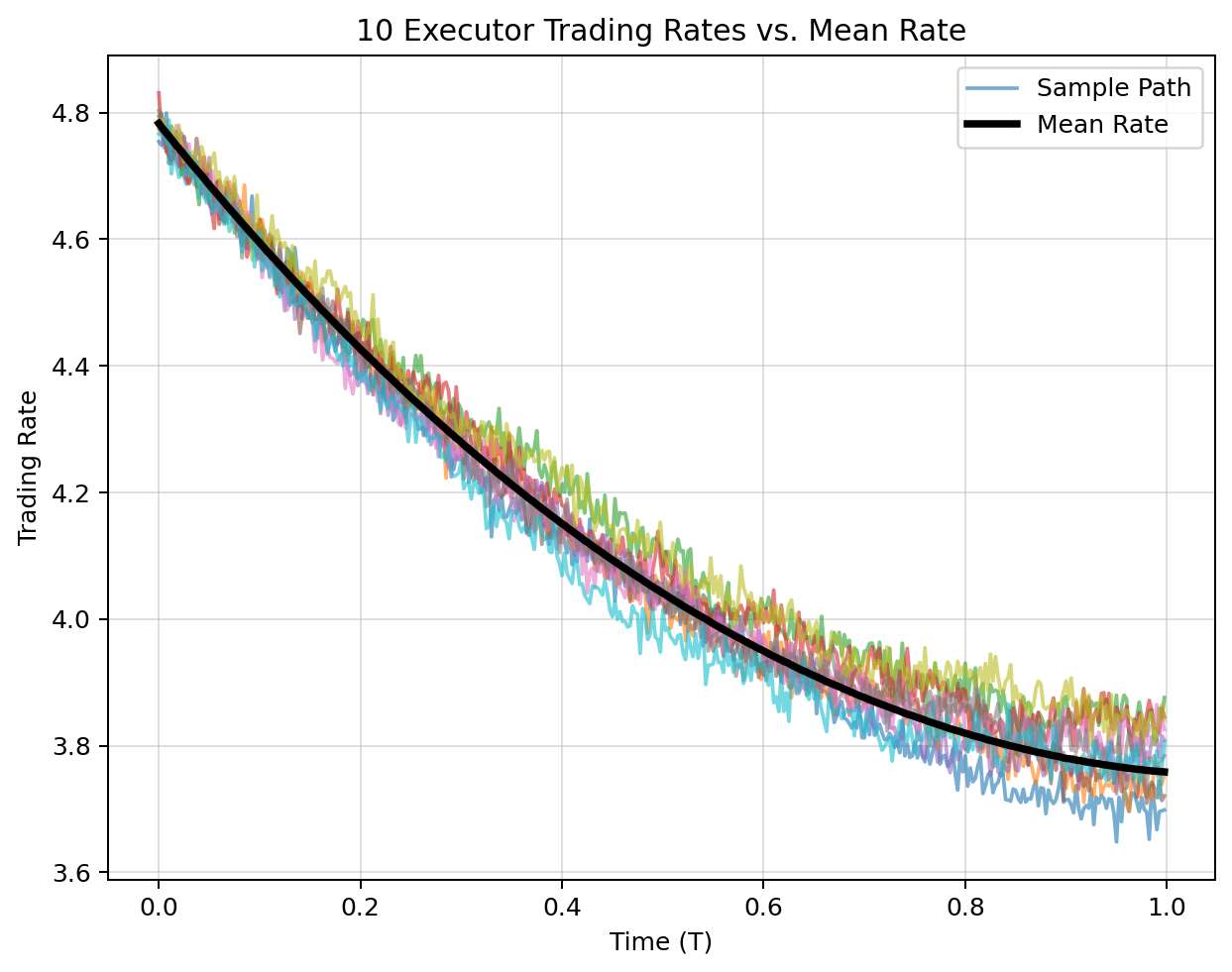}
        \includegraphics[width=0.495\linewidth]{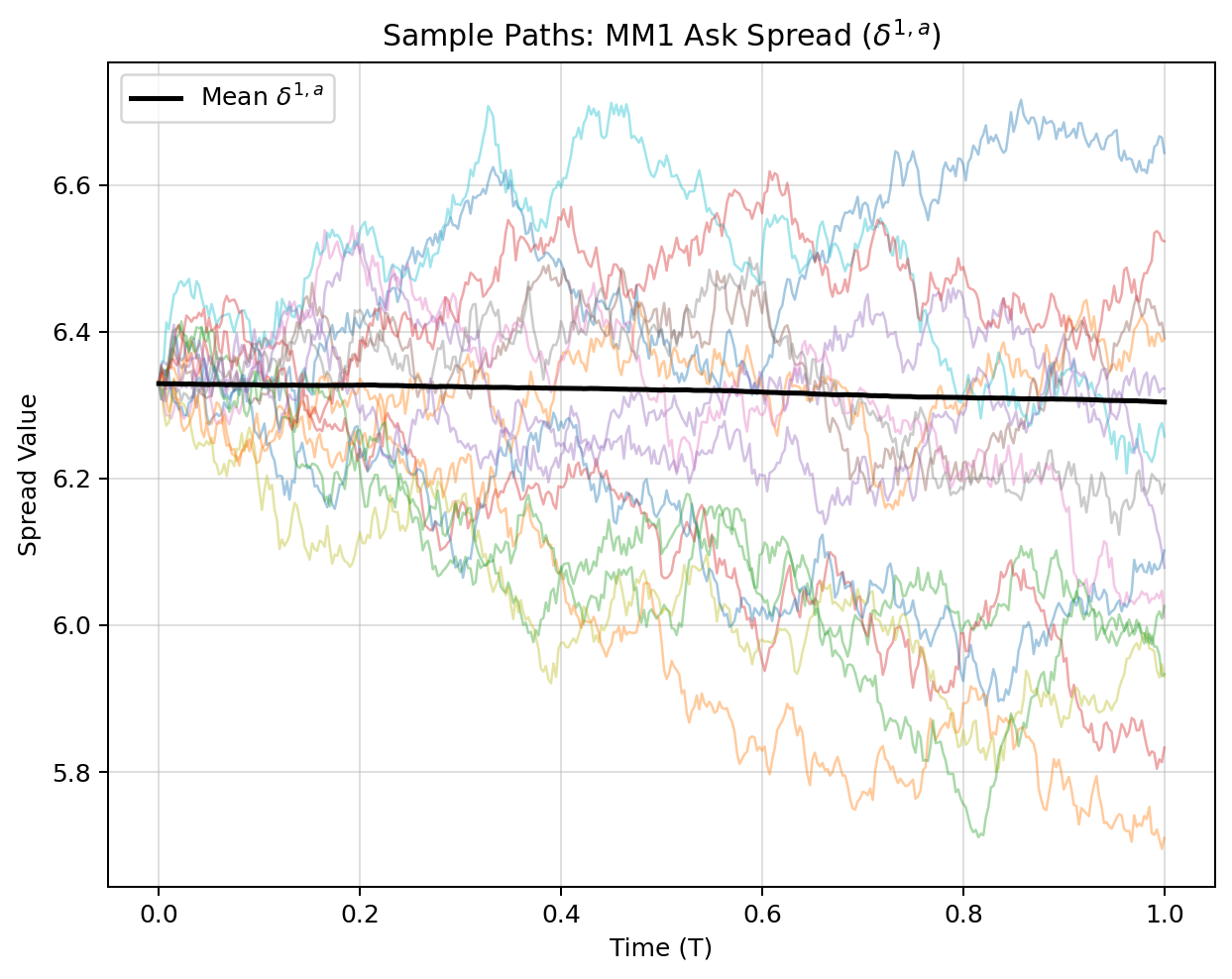}
        \caption{Simulated equilibrium trading rate and spread with mean values}
        \label{sec 5 simulations}
    \end{figure}
    
    Figure \ref{sec 5 correlations} reveals the correlations among trading rate $v^{1, \e}$, ask spread $\delta^{1, a}$, and noise ask rate $(\kappa^a(L_t))_t$ at different times. We first note the following observations.

    \begin{itemize}
        \item The spread $\delta^{1, a}$ is highly positively correlated with the noise ask rate, although this correlation decays slightly over time. This behaviour is consistent with control-theoretic perspective of market making in Example 2.12 of \cite{guo2023macroscopic}. The spread is shown to depend on the conditional expectation of future weighted order imbalance, adjusted by the current inventory level. Our simulations indicate that, when $L$ is small, the expected imbalance varies approximately linearly with the current noise ask rate.\\

        \vspace{-0.3cm}

        \item The correlation between trading rate $v^{1, \e}$ and the noise ask rate is negative: its magnitude increases rapidly at early times and then decays slightly near the end. This pattern agrees with \cite{cartea2016incorporating}, where optimal execution is formulated as a control problem. In that work, the execution strategy is adjusted by a weighted average of the future expected order imbalance. The initially small absolute correlation reflects the influence of the running penalty, while the late-time decay is driven by the terminal penalty.\\

        \vspace{-0.3cm}

        \item The correlation between the trading rate $v^{1, \e}$ and the spread $\delta^{1, a}$ is negative, with magnitude increasing over time. Notably, even though both $v^{1, \e}$ and the noise rate are both market ask orders, their respective correlations with the spread are remarkably distinct.
    \end{itemize}

    Recall the optimal control form \eqref{paper 3 big game feedback}. We provide a unified interpretation. Initially, the trader's execution is largely dictated by the running penalty, necessitating a steady volume regardless of market fluctuations. However, as her urgency level $Y^{1, \e}$ (see Example \ref{sec 4 example}) decays, she becomes more sensitive to prevailing market conditions. When the noise ask rate is high, the trader anticipates a subsequent rise in both the market makers' inventories and the price. Mathematically, a rising price reduces the trading profit margin as per \eqref{paper 3 big game feedback}. Economically, she anticipates that the eventual mean-reversion of the noise flow---combined with the market makers' risk aversion---will lead to a contraction in their positions and a subsequent compression of their spreads. Consequently, she strategically defers execution to capture these more favourable future prices.

    \begin{figure}
        \centering
        \includegraphics[width=0.65\linewidth]{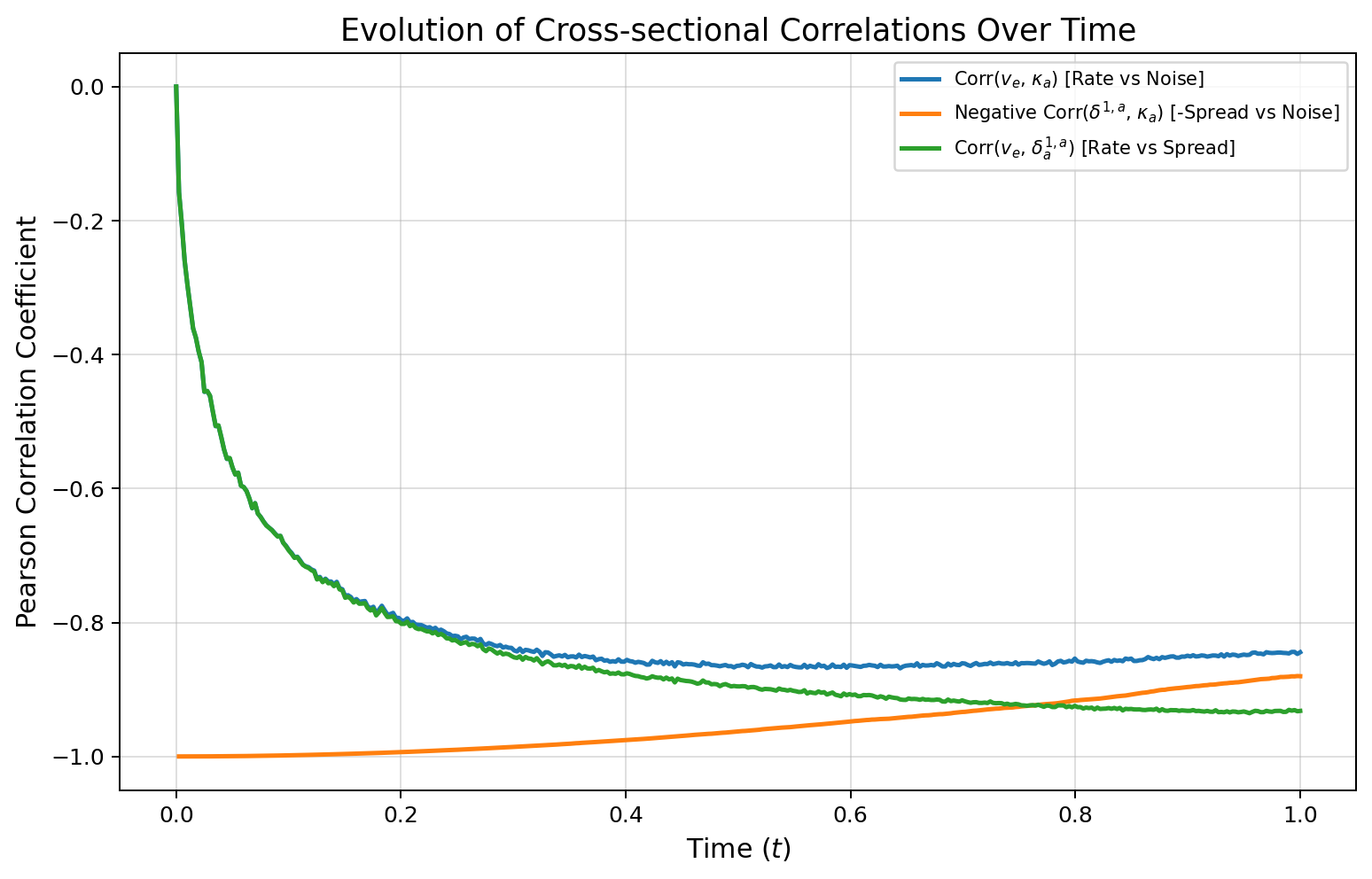}
        \caption{Correlations among trading rate, ask spread, and noise ask rate}
        \label{sec 5 correlations}
    \end{figure}

    The left panel of Figure \ref{high low cost} compares the volume-weighted average price (VWAP) for both strategic and noise traders, ignoring temporary impact. The two average costs are highly correlated and nearly identical. Although the trader follows markedly different execution strategies, she effectively tracks the VWAP of the market. Indeed, the strategic trader slightly outperforms the market. With lower mean and variance, she attains a lower price than noise traders in over $67\%$ of simulations, and the $95\%$ confidence interval for the regression slope is $[0.968,0.969]$. To examine cost variation, we consider cases in which both average prices exceed $6.6$ and cases in which both are below $6.1$. The right panel displays sample paths and the mean trajectories for these two regimes. The results indicate that cost variation is driven by market sentiment: higher sentiment, together with elevated noise ask rates, produces higher average prices.

    \begin{figure}
        \centering
        \includegraphics[width=0.495\linewidth]{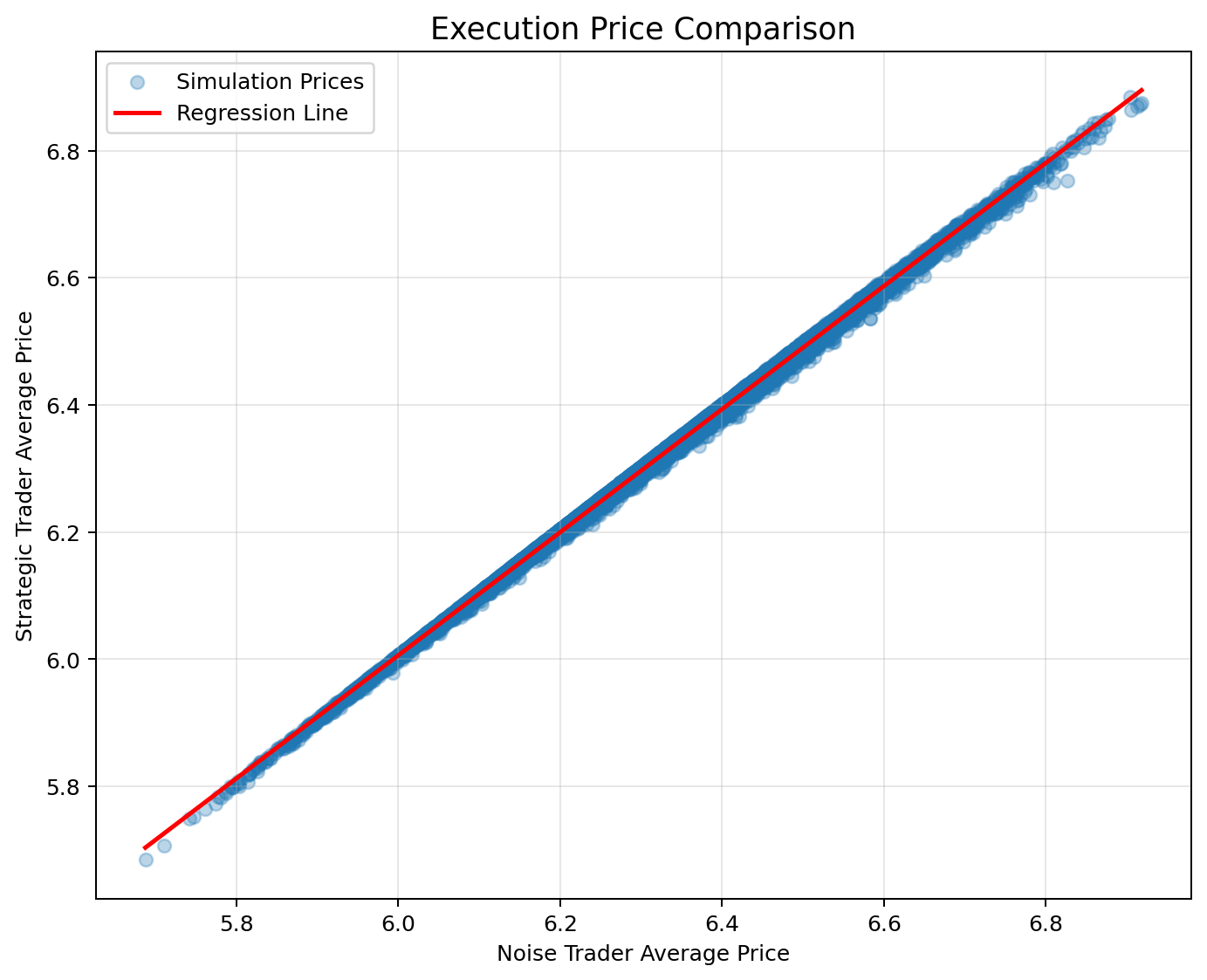}
        \includegraphics[width=0.495\linewidth]{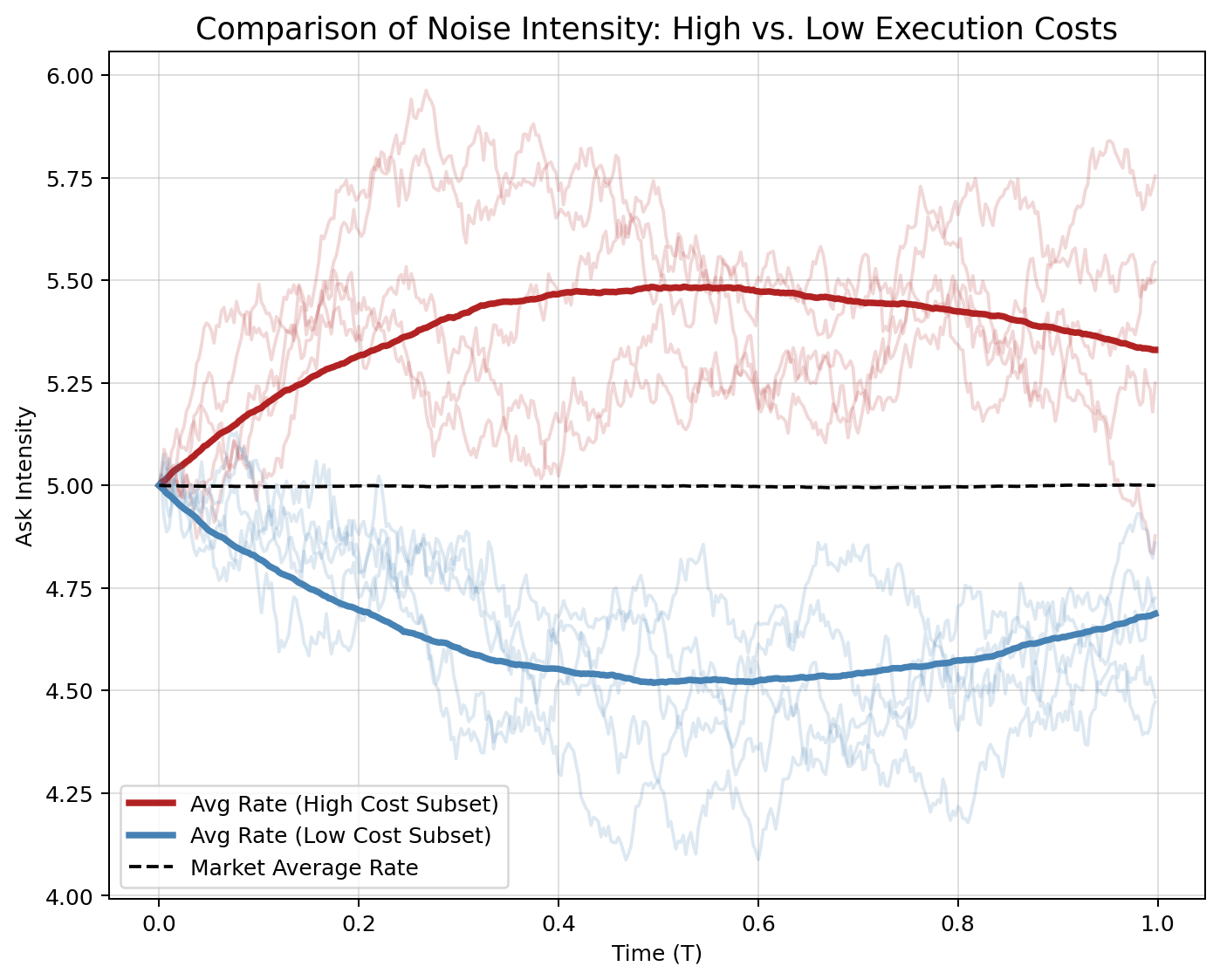}
        \caption{Average execution cost per unit for strategic and noise traders}
        \label{high low cost}
    \end{figure}
\end{example}

\vspace{0.2cm}

\section{Appendix: Proofs in Section \ref{paper 3 section 2}}

\subsection{Proof of Proposition \ref{FBSDE char in Sec2}}

\begin{proof}
    This result is a direct consequence of the stochastic maximum principle. Since the forward equation and terminal condition are similar, we just sketch the proof for the backward equation. The partial derivative of the Hamiltonian $H^{i, \e}$ in $q^{i, \e}$ reads
    \begin{equation*}
        \frac{\partial H^{i,\e}}{\partial q^{i,\e}} = \frac{\alpha_t}{N} \, \sum_{j=1}^N v^{j, \e} -  2\phi_t^{i, \e} \, q^{i, \e}.
    \end{equation*}
Because $v^{i, \e}$ is given by \eqref{paper 3 solve issac}, we then have
\begin{equation*}
\begin{aligned}
    \frac{\partial H^{i,\e}}{\partial q^{i,\e}} &= \frac{\alpha_t}{N} \sum_{j=1}^N \bigg( \frac{N}{(N + 1)\, \beta_t} \,
    \Big( N \, y^{j, \e} - \sum_{k \neq j} y^{k, \e} \Big)\\
    & \hspace{2cm} + \frac{\alpha_t}{(N + 1) \, \beta_t} \, \Big( N \, q^{j, \e} - \sum_{k \neq j} q^{k, \e} \Big) - \frac{N}{(N + 1) \, \beta_t} \, \varpi_t \bigg) -  2\phi_t^{i,\e} \, q^{i,\e}\\
    &= \frac{\alpha_t}{(N + 1) \, \beta_t} \, \sum_{j=1}^N y^{j,\e} + \frac{(\alpha_t)^2}{N (N + 1) \, \beta_t} \, \sum_{j=1}^N q^{j,\e} - \frac{N \, \alpha_t}{(N + 1) \, \beta_t}\, \varpi_t - 2\phi_t^{i, \e} \, q^{i, \e}.
\end{aligned}
\end{equation*}
The backward part of \eqref{paper 3 execut game FBSDE} is just the vector representation for the above equation.
\end{proof}

\kong

\subsection{Proof of Theorem \ref{paper 3 stochastic riccati approach}}

The following standard properties of the matrix exponential turn out to be essential in studying the well-posedness. We summarize them in the following proposition without the proof.

\kong

\begin{proposition}
    The matrix exponential has the following properties:
    \kong
    
    \begin{itemize}
        \item If the square matrix $X$ is symmetric, then the matrix exponential $e^X$ is positive definite;\\
        \vspace{-0.2cm}

        \item If $X Y =Y X$, then $e^X \, e^Y = e^{X+Y}$. Consequently, the inverse of $e^X$ is $e^{-X}$;\\
        \vspace{-0.2cm}

        \item Given $X(t)$ as a square matrix of differentiable functions, then
        \begin{equation}
            \frac{d}{dt} e^{X(t)} = \frac{d}{dt} X(t) \, e^{X(t)} = e^{X(t)} \, \frac{d}{dt} X(t)
            \label{paper 3 matrix exp deriv}
        \end{equation}
        if and only if $X(t)$ and $\frac{d}{dt} X(t)$ commute. In particular, suppose $X(t) = \int_0^t x(s) \, ds$, then \eqref{paper 3 matrix exp deriv} holds if $x(t_1) \, x(t_2) = x(t_2) \, x(t_1) $ for any $t_1$ and $t_2$.
    \end{itemize}
\end{proposition}

\kong

\begin{proof}
    We start with the second case. To solve \eqref{paper 3 execut game FBSDE}, the linear structure suggests the affine ansatz
    \begin{equation}
        \boldsymbol{Y}^\e_t = R_t \, \boldsymbol{Q}^\e_t + H_t,
        \label{paper 3 affine ansatz}
    \end{equation}
    where $R$ is a matrix-valued process. By matching the coefficients, process $R$ and $H$ solve the following coupled BSDE system:
    \begin{align}
        dR_t &= \big( G_t + F_t \, R_t - R_t \, E_t - R_t \, D_t \, R_t \big) \, dt + dM_t^1, \label{BSRE} \\
        dH_t &= \bigg( -R_t\, D_t \, H_t + R_t \, \frac{N}{(N + 1)\, \beta_t}\, \w_t + F_t \, H_t + \frac{N \, \alpha_t}{(N + 1)\, \beta_t}\, \w_t \bigg) \, dt + dM_t^2, \nonumber
    \end{align}
    such that $R_T = - L$ and $H_T = \w_T$. The BSRE \eqref{BSRE} for $R$ is non-symmetric in the sense of Riccati equation. Provided with a bounded process $R$, process $H$ is the unique solution of a Lipschitz BSDE. Therefore, it suffices to prove the existence and uniqueness of a bounded solution to the BSRE \eqref{BSRE}. 
    
    By definition, we can see for any $s,t \in [0,T]$ that $F_s \, F_t = F_t \, F_s$ and $E_s \, E_t = E_t \, E_s$. The property of matrix exponential yields
    \begin{equation*}
    \begin{aligned}
        d \big( e^{-\int_0^t F_u \,du} \, R_t \, e^{\int_0^t E_u \,du} \big) = \big( - e^{-\int_0^t F_u \,du} \,F_t \big) \, R_t \, e^{\int_0^t E_u \,du} &+ e^{-\int_0^t F_u \,du} \, (dR_t) \, e^{\int_0^t E_u \,du} \\
        &+ e^{-\int_0^t F_u \,du} \, R_t \, \big( E_t \, e^{\int_0^t E_u \,du} \big).
    \end{aligned}
    \end{equation*}
    Defining $\tilde{R}_t := e^{-\int_0^t F_u \,du} \, R_t \, e^{\int_0^t E_u \,du}$, the non-singularity property of matrix exponentials allows us to equivalently study the Riccati equation for $\tilde{R}$. Direct calculations yield:
    \begin{equation}
        d\tilde{R}_t = \big( e^{-\int_0^t F_u \,du} \, G_t \, e^{\int_0^t E_u \,du} - \tilde{R}_t \, e^{-\int_0^t E_u \,du} \, D_t \, e^{\int_0^t F_u \,du} \, \tilde{R}_t \big) \, dt + d\tilde{M}_t,
        \label{paper 3 execut game aux}
    \end{equation}
    such that $\tilde{R}_T = - e^{-\int_0^T F_u \,du} \, L \, e^{\int_0^T E_u \,du}$.
    
    To study \eqref{paper 3 execut game aux}, we proceed to show the symmetry and positive definiteness of the coefficient matrices. First, let us comment that $G_t, D_t$, and $A$ are positive semi-definite given the conditions of the theorem, due to symmetry, positive diagonal entries, and diagonal dominance. Note that $D_t, E_t,$ and $F_t$ are matrices of type ($\mathcal{S})$: (1) all the diagonal entries are the same; (2) all the off-diagonal entries are also the same. Because such type of matrices is closed under multiplications, we can see that $e^{-\int_0^t E_u \,du} \, D_t \, e^{\int_0^t F_u \,du}$ is also symmetric for all $t$. Further, consider the transformation 
    \begin{equation*}
    \begin{aligned}
        e^{-\int_0^t E_u \,du} \, D_t \, e^{\int_0^t F_u \,du} &= e^{-\int_0^t E_u \,du} \, D_t \, e^{\int_0^t E_u \,du} \, e^{-\int_0^t E_u \,du} \, e^{\int_0^t F_u \,du}\\
        & = e^{-\int_0^t E_u \,du} \, D_t \, e^{\int_0^t E_u \,du} \, e^{-\int_0^t (E_u - F_u) \,du}.
    \end{aligned}
    \end{equation*}
    The last equality is true because $\int_0^t E_u \,du$ and $\int_0^t F_u \,du$ commute:
    \begin{equation*}
    \begin{aligned}
        \int_0^t E_u \,du \, \int_0^t F_u \,du &= \bigg( \int_0^t \frac{\alpha_u}{(N + 1) \, \beta_u} \,du \bigg) \, \bigg( \int_0^t \frac{-\alpha_u }{(N + 1) \, \beta_u} \,du \bigg) \, B \,O\\
        &= \bigg( \int_0^t \frac{\alpha_u}{(N + 1) \, \beta_u} \,du \bigg) \, \bigg( \int_0^t \frac{-\alpha_u }{(N + 1) \, \beta_u} \,du \bigg) \, O \, B = \int_0^t F_u \,du \, \int_0^t E_u \,du.
    \end{aligned}
    \end{equation*}
    Because $e^{-\int_0^t (E_u - F_u) \, du}$ is symmetric and positive definite, there exists a unique symmetric matrix $K$ such that $ e^{-\int_0^t (E_u - F_u) \, du} = K \, K$. We also observe that the eigenvalues of 
    \begin{equation*}
        e^{-\int_0^t E_u \,du} \, D_t \, e^{\int_0^t F_u \,du} = e^{-\int_0^t E_u \,du} \, D_t \, e^{\int_0^t E_u \,du} \, K \, K
    \end{equation*}
    are the same with
    \begin{equation}
        K \, e^{-\int_0^t E_u \,du} \, D_t \, e^{\int_0^t E_u \,du} \, K.
        \label{paper 3 eigen prop}
    \end{equation}
    The fact that matrix \eqref{paper 3 eigen prop} is positive semi-definite follows from: for any $x\in\mathbb{R}^N$, it holds
    \begin{equation*}
        x^* \, \Big( K \, e^{-\int_0^t E_u \,du} \, D_t \, e^{\int_0^t E_u \,du} \, K \Big) \, x = (K \, x)^* \, \Big( e^{-\int_0^t E_u \,du} \, D_t \, e^{\int_0^t E_u \,du} \Big) \, (K\, x) \geq 0. 
    \end{equation*}
    Here, the last inequality is true because $e^{-\int_0^t E_u \,du} \, D_t \, e^{\int_0^t E_u \,du}$ is not only symmetric but also has non-negative eigenvalues that are the same with $D_t$. Now, we can conclude that 
    \begin{equation*}
        e^{-\int_0^t E_u \,du} \, D_t \, e^{\int_0^t F_u \,du}
    \end{equation*}
    is symmetric and positive semi-definite. The same is true for both
    \begin{equation*}
        e^{-\int_0^t F_u \,du} \, G_t \, e^{\int_0^t E_u \,du} \text{\; and \;} e^{-\int_0^T F_u \,du} \, L \, e^{\int_0^T E_u \, du}
    \end{equation*}
    through similar arguments.

    Since the BSRE \eqref{paper 3 execut game aux} is of symmetric and positive semi-definite type, the existence of a unique bounded, symmetric, and negative semi-definite $\tilde{R}$ is studied in \cite{peng1992stochastic}. We proceed to show that the solution $\tilde{R}$ is of type $(\mathcal{S})$. For convenience let us define 
    \begin{gather*}
        \tilde{D}_t = e^{-\int_0^t E_u \,du} \, D_t \, e^{\int_0^t F_u \,du}, \quad \tilde{G}_t = e^{-\int_0^t F_u \,du} \, G_t \, e^{\int_0^t E_u \,du}, \text{\; and \;} \tilde{L} = e^{-\int_0^T F_u \,du} \, L \, e^{\int_0^T E_u \,du},
    \end{gather*}
    while noticing that they are all of type $(\mathcal{S})$. For any $i, j \in \{1, \dots, N\}$, denote by $E_{i j}$ the permutation matrix of the $i$-th row and $j$-th row:
    \begin{equation*}
        E_{ij} = 
        \begin{bmatrix}
            & \vdots & & \vdots\\
            \cdots & 0 & \cdots & 1 & \cdots\\
            & \vdots & & \vdots\\
            \cdots & 1 & \cdots & 0 & \cdots\\
            & \vdots & & \vdots
        \end{bmatrix}.
    \end{equation*}
    Consider the transform $\hat{R}_t = E_{ij} \, \tilde{R}_t \, E_{ij}$. Being aware of
    \begin{equation*}
        E_{ij} \, \tilde{R}_t \, \tilde{D}_t \, \tilde{R}_t \, E_{ij} = \big( E_{ij} \, \tilde{R}_t \, E_{ij} \big)  \big( E_{ij} \, \tilde{D}_t \, E_{ij} \big) \big( E_{ij} \, \tilde{R}_t \, E_{ij} \big) = \hat{R}_t \, \tilde{D}_t \, \hat{R}_t, 
    \end{equation*}
    we find that $\hat{R}$ solves
    \begin{equation*}
        d\hat{R}_t = \big( \tilde{G}_t - \hat{R}_t \, \tilde{D}_t \, \hat{R}_t \big) \, dt + d\hat{M}_t
    \end{equation*}
    such that $\hat{R}_T = - \tilde{L}$, which is the same as \eqref{paper 3 execut game aux}. The uniqueness of the solution implies $\tilde{R} = \hat{R} = E_{ij} \, \tilde{R} \, E_{ij}$. The arbitrariness of $i$ and $j$ guarantees that $\tilde{R}_t$ is of type ($\mathcal{S}$). Consequently, by the definition $R_t = e^{\int_0^t F_u \, du} \tilde{R}_t \, e^{-\int_0^t E_u \, du}$, we obtain that $R_t$ is a symmetric matrix with non-positive eigenvalues through a similar argument as above. 
    
    Given the well-posedness of $R$ and $H$, a solution $(\boldsymbol{Q}^\e, \boldsymbol{Y}^\e, \boldsymbol{M}^\e)$ can be constructed by plugging \eqref{paper 3 affine ansatz} back to system \eqref{paper 3 execut game FBSDE}. With respect to the uniqueness, we apply the idea of decouple field; see \cite{ankirchner2020optimal} for a short review. Due to the boundedness of $R$, the linear form \eqref{paper 3 affine ansatz} defines a regular decoupling field (see Definition 2.4 in \cite{ankirchner2020optimal}). Subsequently, the existence of a regular decoupling field guarantees the uniqueness of the solution, by Theorem 2.6 in \cite{ankirchner2020optimal}. This completes the proof of the second case. The proof for the first case follows by similar arguments as above.
\end{proof}

\kong

\subsection{Proof of Theorem \ref{method of continuation}}

The following result of matrix algebra turns out to be useful.

\kong

\begin{theorem}[\cite{varah1975lower}]
\label{paper 3 inverse_norm_matrix}
Assume $\boldsymbol{A}\in\mathbb{R}^{n\times n}$ is strictly diagonally dominant (by rows) matrix and set the `gap' $\alpha=\min_{1\leq k\leq n}\{|\boldsymbol{a}_{kk}-\sum_{j\neq k}|\boldsymbol{a}_{kj}|\}$. Then, $\|\boldsymbol{A}^{-1}\|_{\infty}\leq 1\,/\,\alpha$, where $\|\cdot\|_\infty$ is the matrix norm induced by the vector $\infty$-norm.
\end{theorem}

\kong

\begin{proof}
    Let us verify the monotone condition in the method of continuation ((H2.3) in \cite{peng1999fully}). It suffices to show that there exists a constant $r > 0$ such that
    \begin{equation}
    \begin{aligned}
        r \, |\Delta q|^2 &\leq \big( -F_t \, \Delta y + G_t \, \Delta q \big)^* \, \Delta q + \big( D_t \, \Delta y - E_t \, \Delta q \big)^* \, \Delta y\\
        &= \Delta q^* \, G_t \, \Delta q + \Delta y^* \, D_t \, \Delta y - \Delta q^* \, (E_t + F_t) \, \Delta y,\\
            \label{paper 3 monotone condit}
        r \, |\Delta q|^2 & \leq \Delta q^*\, L \, \Delta q
    \end{aligned}
    \end{equation}
    holds for any $\Delta q, \Delta y \in \mathbb{R}^N$. The second inequality is a direct consequence of the positive lower bound for each $A^{i, \e}$. Because matrix $D_t$ is symmetric positive definite, the right hand side of \eqref{paper 3 monotone condit} is convex with respect to $\Delta y$. The first-order condition yields
    \begin{equation}
    \begin{aligned}
         \Delta q^* \, G_t \, \Delta q + \Delta y^* \, D_t \, \Delta y &- \Delta q^* \, (E_t + F_t) \, \Delta y \\
         & \geq \Delta q^* \, G_t \, \Delta q + (V_t \, \Delta q)^* \, D_t \, (V_t \, \Delta q) - \Delta q^* \, (E_t + F_t) \, (V_t \, \Delta q)\\
         & = \Delta q^* \, \big( G_t + V_t^* \, D_t \, V_t - (E_t + F_t) \, V_t \big) \, \Delta q,
    \end{aligned}
    \label{paper 3 simple monotone condit}
    \end{equation}
    where $V_t = D_t^{-1} \, (E_t + F_t)/2$. We know $V_t$ is symmetric because $D_t, E_t,$ and $F_t$ are all of type $(\mathcal{S})$. It then follows
    \begin{equation*}
    \begin{aligned}
         V_t^* \, D_t \, V_t - (E_t +F_t) \, V_t &= \frac{1}{4} \, (E_t +F_t) \, D_t^{-1} \, (E_t +F_t) - \frac{1}{2} \, (E_t +F_t) \, D_t^{-1} \, (E_t +F_t)\\
         &=-\frac{1}{4} \, \frac{(\alpha_t)^2}{N(N + 1)\, \beta_t} \, \big( B - 2\,O + O\, B^{-1}\,O \big)\\
         &=-\frac{1}{4} \, \frac{(\alpha_t)^2}{N(N + 1)\, \beta_t} \, \big(B + (N - 2) \, O \big).
    \end{aligned}
    \end{equation*}
    Here, we have applied the fact that the column sum of $B^{-1}$ is $1$ and thus $O B^{-1} = O$. Consequently, we have
    \begin{equation*}
        G_t + V_t^* \, D_t \, V_t - (E_t + F_t) \, V_t =
        \begin{bmatrix}
        2\phi_t^{1,\e} & \cdots & 0\\
        \vdots & \ddots & \vdots\\
        0  & \cdots & 2\phi_t^{N,\e}
    \end{bmatrix}
    - \frac{(\alpha_t)^2}{4N \, \beta_t} \,
    \begin{bmatrix}
        2& \cdots & 1\\
        \vdots & \ddots & \vdots\\
        1  & \cdots & 2
    \end{bmatrix}.
    \end{equation*}
    Based on the conditions of the theorem, we can see that the symmetric matrix above has positive diagonal entries and is diagonally dominant by some constant $C > 0$ for all $t$, establishing the positive definiteness. In view of Theorem \ref{paper 3 inverse_norm_matrix} and the equivalence of matrix norms, the smallest eigenvalue is bounded away from $0$ with respect to $t$. Hence, a suitable $r > 0$ can be found. Therefore, the monotone condition guarantees the well-posedness of the equation.
\end{proof}

\vspace{0.2cm}

\section{Appendix: Proofs in Section \ref{paper 3 section 3}}

\subsection{Proof of Theorem \ref{paper 3 big maximum principle}}

The following result of non-smooth implicit function theorem will be used in the proof.

\begin{proposition}[\cite{guo2024macroscopicmarketmakinggames}]
\label{general_implicit}
Assume that $F:\mathbb{R}^n\times\mathbb{R}^m\to \mathbb{R}^n$ is a locally Lipschitz mapping such that:

\kong

\begin{itemize}
    \item [(1)] For every $y\in\mathbb{R}^m$, the functional $\varphi_y:\mathbb{R}^n\to\mathbb{R}$ given by the formula
    \begin{equation*}
        \varphi_y(x):=\frac{1}{2}\,|F(x,y)|^2
    \end{equation*}
    is coercive, i.e., $\lim_{|x|\to\infty}\varphi_y(x)=\infty$;\\
    
    \item[(2)] Any matrix in the set $\partial_xF(x,y)$ is of maximal rank for all $(x,y)\in\mathbb{R}^n\times\mathbb{R}^m$;\\
    
    \item[(3)] Define the function $\tilde{F}:\mathbb{R}^{m+n}\to\mathbb{R}^{m+n}$ by
    \begin{equation*}
        \tilde{F}(x,y)=\big(y, \;F(x,y)\big)
    \end{equation*}
    and let $\tilde{S}$ denote the unit sphere in $\mathbb{R}^{m+n}$. There exists a constant $\upsilon>0$ such that the distance between $\partial \tilde{F}(x,y)\,\tilde{S}$ and $0$ is at least $\upsilon$, for all $(x,y)\in\mathbb{R}^n\times\mathbb{R}^m$. Here, set $\partial \tilde{F}(x,y)\,\tilde{S}$ is defined by
    \begin{equation*}
        \partial \tilde{F}(x,y) \, \tilde{S} := \big\{ U \, v \, : \, U \in \tilde{F}(x,y) \text{\, and \,} v \in \tilde{S} \big\}
    \end{equation*}
    and the distance refers to $\inf\{|U\,v| \, : \, U \in \tilde{F}(x,y) \text{\, and \,} v \in \tilde{S} \}$.
\end{itemize}
\kong

\noindent Then, there exists a unique Lipschitz function $f:\mathbb{R}^m\to\mathbb{R}^n$ such that $F(x,y)=0$ and $x=f(y)$ are equivalent in the set $\mathbb{R}^n\times\mathbb{R}^m$. 
\end{proposition}

\noindent We refer the reader to the appendix of \cite{guo2024macroscopicmarketmakinggames} for the definition and key properties of the Clarke generalized derivatives.

\begin{proof}
    We start with the necessary condition. Suppose that $( \boldsymbol{v}^\e, (\boldsymbol{\delta}^j)_{j})$ forms a Nash equilibrium. Equipped with trading strategies $\boldsymbol{v}^\e$, the $\tilde{N}$ market makers solve a market making game studied in \cite{guo2024macroscopicmarketmakinggames}. Based on Theorem 3.9 in \cite{guo2024macroscopicmarketmakinggames}, there exist Lipschitz mappings $\psi^a, \psi^b : \mathbb{R}^{\tilde{N}} \to \mathbb{R}^{\tilde{N}}$ such that the equilibrium strategy satisfies
    \begin{equation*}
        \delta_t^{i, a} = \psi^{i, a}(\boldsymbol{Y}_t^\m) \text{ \; and \; } \delta_t^{i, b} = \psi^{i, b}(\boldsymbol{Y}_t^\m), 
    \end{equation*}
    where the adjoint process $\boldsymbol{Y}^\m$ solves the FBSDE
    \begin{equation*}
        \left\{
        \begin{aligned}
         dQ_t^{i, \m} &= - \tilde{a}_t \, \Lambda \big(\psi^{i,a}(\boldsymbol{Y}_t^\m) - \bar{\psi}^{i,a}(\boldsymbol{Y}_t^\m) \big) \, dt + \tilde{b}_t \, \Lambda \big(\psi^{i,b}(\boldsymbol{Y}_t^\m) - \bar{\psi}^{i,b}(\boldsymbol{Y}_t^\m) \big ) \, dt,\\
        \, dY_t^{i, \m} &= 2\phi_t^{i, \m} \, Q_t^{i, \m} \, dt + dM_t^{i, \m},\\
        Q_0^{i, \m} &= q_0^{i, \m}, \quad Y_T^{i, \m} = -2A^{i,\m} \, Q_T^{i,\m},
        \end{aligned}
        \right.
    \end{equation*}
    for all $i \in \{1, \dots, \tilde{N}\}$, where $\bar{\psi}^{i, a}(\boldsymbol{y}) = \min_{k \neq i} \psi^{k, a}(\boldsymbol{y})$. 
    
    Given quoting strategies $(\boldsymbol{\delta}^j)_{j} = ( \psi^{a}(\boldsymbol{Y}_t^\m), \psi^{b}(\boldsymbol{Y}_t^\m) )_{t \in [0, T]}$, all traders participate in the optimal execution game as in Section \ref{paper 3 section 2}. Different from the previous section, the Hamiltonian of trader $(i,\e)$ now reads
    \begin{equation*}
    \begin{aligned}
         H^{i, \e}(t, q^{i, \e}, y^{i, \e}, v^{i, \e}; v^{-i, \e}) &= v^{i,\e} \, y^{i,\e} \\
         &- \Big[ \mathbb{I}(q_0^{i,\e} < 0) \,
    \min_{j} \psi^{j, a}(\boldsymbol{Y}_t^\m) - \mathbb{I}(q_0^{i,\e} > 0) \, \min_{j} \psi^{j, b}(\boldsymbol{Y}_t^\m) \Big] \, v^{i, \e} \\
    &- \lambda^i \big( \boldsymbol{v}^\e \big) - \phi_t^{i, \e} \, \big( q^{i,\e} \big)^2.
    \end{aligned}
    \end{equation*}
    Observing that (1) the variables $q^{i,\e}$ and $v^{i,\e}$ enter the Hamiltonian separately, and (2) the function $\lambda^{i}$ is convex in its $i$-th argument, it follows that the Hamiltonian is concave in $(q^{i, \e}, v^{i, \e})$ unconditionally. We can thus apply the stochastic maximum principle directly. Let us define
    \begin{equation*}
    \tilde{y}^{i, \e} = y^{i,\e} - \mathbb{I}(q_0^{i, \e} < 0) \, 
    \min_{j} \psi^{j, a}(\boldsymbol{Y}_t^\m) + \mathbb{I}(q_0^{i,\e} > 0) \, \min_{j} \psi^{j, b}(\boldsymbol{Y}_t^\m)
    \end{equation*}
    and denote by $\iota^i(\cdot \, ; v^{-i, \e})$ the inverse function of $( \partial \lambda^i / \partial u^{i})(\dots v^{i - 1, \e}, \, \cdot \, , v^{i+1, \e} \dots)$. To maximise all the Hamiltonians simultaneously, the first order condition yields
    \begin{equation*}
    \begin{aligned}
        \mathfrak{F}^{i}(\boldsymbol{v}^\e, \tilde{y}^{i, \e}) := v^{i, \e} - \Big[  \mathbb{I}(q_0^{i,\e} &< 0) \, \iota^i(\tilde{y}^{i,\e}; v^{-i, \e} ) \vee 0 \wedge \tilde{\xi}\\
        &+ \mathbb{I}(q_0^{i, \e} > 0) \, \iota^i( \tilde{y}^{i,\e}; v^{-i,\e} ) \vee (-\tilde{\xi}) \wedge 0 \Big] = 0
        \end{aligned}
    \end{equation*}
    for all $i$. Since $t$ is a fixed time, we omit the dependence of $\mathfrak{F}^i$ and $\tilde{y}^{i,\e}$ on $t$. Setting $\mathfrak{F}=(\mathfrak{F}^1, \dots, \mathfrak{F}^N): \mathbb{R}^N \times \mathbb{R}^N \to \mathbb{R}^N$ and $\boldsymbol{\tilde{y}}^{\e}=(\tilde{y}^{1,\e}, \dots, \tilde{y}^{N,\e})$, we obtain the condition
    \begin{equation*}
        \mathfrak{F}(\boldsymbol{v}^\e, \boldsymbol{\tilde{y}}^{\e}) = 0.
    \end{equation*}
    Due to the truncation by $\tilde{\xi}$ and $0$, the Brouwer fixed-point theorem ensures that, for every $\boldsymbol{\tilde{y}}^\e \in \mathbb{R}^N$, there exists $\boldsymbol{v}^\e \in \mathbb{R}^N$ such that the above equation holds. 
    
    To show the Lipschitz dependence of $\boldsymbol{v}^\e$ on $\boldsymbol{\tilde{y}}^\e$, we utilize the implicit function theorem \ref{general_implicit}. The first coercive condition is a direct consequence of the truncation. By techniques for computing derivatives of inverse functions, we can have
    \begin{align*}
        \frac{\partial \iota^i}{\partial u^{i}}(u^{i}; u^{-i}) &= \bigg[ \frac{ \partial^2 \lambda^i}{\partial (u^{i})^2} \big( \dots, u^{i - 1}, \iota^i(u^{i}; u^{-i}), u^{i + 1}, \dots \big)\bigg]^{-1}\\
        \frac{\partial \iota^i}{\partial u^{k}}(u^{i}; u^{-i}) &= (-1) \, \frac{ \partial^2 \lambda^i}{\partial u^{i} \, \partial u^{k}} \, \big/ \, \frac{ \partial^2 \lambda^i}{\partial (u^{i})^2} \big( \dots, u^{i - 1}, \iota^i(u^{i}; u^{-i}), u^{i + 1}, \dots \big)
    \end{align*}
    for any $k\neq i$. Whenever differentiable, the partial derivative of $\mathfrak{F}^i$ with respect to $v^{k,\e}$ is either
    \begin{equation*}
    \begin{aligned}
        \frac{\partial \mathfrak{F}^i}{\partial v^{k, \e}} &= -\mathbb{I}(q_0^{i, \e} < 0) \, \frac{\partial \iota^i}{\partial u^{k}}(\tilde{y}^{i, \e}; v^{-i,\e}) - \mathbb{I}(q_0^{i, \e} > 0) \, \frac{\partial \iota^i}{\partial u^{k}}(\tilde{y}^{i,\e}; v^{-i,\e})\\
        &= (-1) \, \big[ \mathbb{I}(q_0^{i, \e} < 0) + \mathbb{I}(q_0^{i, \e} > 0) \big] \, \frac{ \partial^2 \lambda^i}{\partial u^{i} \, \partial u^{k}} \, \big/ \, \frac{ \partial^2 \lambda^i}{\partial (u^{i})^2}
    \end{aligned}
    \end{equation*}
    or $0$, for any $k \neq i$. For all $(\boldsymbol{v}^\e, \boldsymbol{\tilde{y}}^\e)$, it follows that
    \begin{equation*}
         \frac{\partial \mathfrak{F}^i}{\partial v^{i, \e}} - \sum_{k \neq i} \Big| \frac{\partial \mathfrak{F}^i}{\partial v^{k, \e}} \Big| \geq 1 - \sum_{k \neq i} \Big| \frac{ \partial^2 \lambda^i}{\partial u^{i}\, \partial u^{k}} \, \big/ \, \frac{ \partial^2 \lambda^i}{\partial (u^{i})^2} \Big| > C
    \end{equation*}
    for some $C > 0$, based on the properties in Definition \ref{paper 3 temp price impact}. Since the above inequality still holds after any convex combination, we can conclude the non-singularity of the generalized Jacobian of $\mathfrak{F}$ with respect to $\boldsymbol{v}^\e$, due to the strict diagonal dominance.

    To check the last condition, let us define $\mathfrak{G}: \mathbb{R}^{2N} \to \mathbb{R}^{2N}$ by 
    \begin{equation*}
    \mathfrak{G}(\boldsymbol{v}^\e, \boldsymbol{\tilde{y}}^\e) = ( \boldsymbol{\tilde{y}}^\e, \, \mathfrak{F}(\boldsymbol{v}^\e, \boldsymbol{\tilde{y}}^\e) )
    \end{equation*}
    Whenever $\mathfrak{G}$ is differentiable, the Jacobian matrix $\nabla \mathfrak{G}$ has the form
\begin{equation}
    \nabla \mathfrak{G}(\boldsymbol{v}^\e, \boldsymbol{\tilde{y}}^\e) =
\begin{bmatrix}
I & 0\\
\nabla_{\boldsymbol{\tilde{y}}^\e} \mathfrak{F}(\boldsymbol{v}^\e, \boldsymbol{\tilde{y}}^\e) & \nabla_{\boldsymbol{v}^\e} \mathfrak{F}(\boldsymbol{v}^\e, \boldsymbol{\tilde{y}}^\e)
\end{bmatrix}
,
\label{paper 3 block_matrix}
\end{equation}
where $I\in\mathbb{R}^{N\times N}$ is the identity matrix. Recalling the derivative of inverse functions, the Jacobian $\nabla_{\boldsymbol{\tilde{y}}^\e} \mathfrak{F}(\boldsymbol{v}^\e, \boldsymbol{ \tilde{y}}^\e )$ is a diagonal matrix with 
\begin{equation*}
\begin{aligned}
    \big[ \nabla_{\boldsymbol{\tilde{y}}^\e} \mathfrak{F}(\boldsymbol{v}^\e, \boldsymbol{\tilde{y}}^\e) \big]_{ii} &= -\mathbb{I}(q_0^{i,\e} < 0) \, \frac{\partial \iota^i}{\partial \tilde{y}^{i, \e}}(\tilde{y}^{i,\e}; v^{-i,\e}) - \mathbb{I}(q_0^{i,\e} > 0) \, \frac{\partial \iota^i}{\partial \tilde{y}^{i, \e}}(\tilde{y}^{i,\e}; v^{-i,\e})\\
        &= (-1) \, \big[ \mathbb{I}(q_0^{i,\e} > 0) + \mathbb{I}(q_0^{i, \e} < 0) \big] \, \big/ \, \frac{\partial^2 \lambda^i}{\partial (u^{i})^2}
\end{aligned}        
\end{equation*}
or $0$, for all $i$. Definition \ref{paper 3 temp price impact} then implies that the absolute value of $\big[ \nabla_{\boldsymbol{\tilde{y}}^\e} \mathfrak{F}(\boldsymbol{v}^\e, \boldsymbol{\tilde{y}}^\e) \big]_{ii}$ is less than some constant $C$ for all $(\boldsymbol{v}^\e, \boldsymbol{\tilde{y}}^\e)$. In view of the singular value decomposition (SVD), we now show the smallest singular value $\underline{\sigma}(\boldsymbol{v}^\e, \boldsymbol{\tilde{y}}^\e)$ of $\nabla\mathfrak{G} (\boldsymbol{v}^\e, \boldsymbol{\tilde{y}}^\e)$ is uniformly bounded away from $0$. First, note that the smallest singular value of a matrix is the reciprocal of the $2$-norm of its inverse:
\begin{equation*}
    \underline{\sigma}(\boldsymbol{v}^\e, \boldsymbol{\tilde{y}}^\e) =\frac{1}{\|\nabla\mathfrak{G}(\boldsymbol{v}^\e, \boldsymbol{\tilde{y}}^\e)^{-1}\|_2},
\end{equation*}
where we use $\|\cdot\|_p$ to indicate the matrix norm induced by the vector $p$-norm. Compute that
\begin{equation*}
    \nabla\mathfrak{G}(\boldsymbol{v}^\e, \boldsymbol{\tilde{y}}^\e)^{-1} =
    \begin{bmatrix}
I & 0\\
-\nabla_{\boldsymbol{v}^\e} \mathfrak{F}(\boldsymbol{v}^\e, \boldsymbol{\tilde{y}}^\e)^{-1} \, \nabla_{\boldsymbol{\tilde{y}}^\e} \mathfrak{F}(\boldsymbol{v}^\e, \boldsymbol{\tilde{y}}^\e) & \nabla_{\boldsymbol{v}^\e} \mathfrak{F}(\boldsymbol{v}^\e, \boldsymbol{\tilde{y}}^\e)^{-1}
\end{bmatrix}.
\end{equation*}
the triangle inequality and the matrix norm inequality $\|\cdot\|_2\leq \sqrt{N}\,\|\cdot\|_\infty$ yield
\begin{equation*}
\begin{aligned}
    \| \nabla\mathfrak{G}(\boldsymbol{v}^\e, \boldsymbol{\tilde{y}}^\e)^{-1} \|_2 &\leq \|I\|_2+ \| \nabla_{\boldsymbol{v}^\e} \mathfrak{F}(\boldsymbol{v}^\e, \boldsymbol{\tilde{y}}^\e)^{-1} \, \nabla_{\boldsymbol{\tilde{y}}^\e} \mathfrak{F}(\boldsymbol{v}^\e, \boldsymbol{\tilde{y}}^\e) \|_2 + \| \nabla_{\boldsymbol{v}^\e} \mathfrak{F}(\boldsymbol{v}^\e, \boldsymbol{\tilde{y}}^\e)^{-1} \|_2\\
    &\leq 1 + \| \nabla_{\boldsymbol{v}^\e} \mathfrak{F}(\boldsymbol{v}^\e, \boldsymbol{\tilde{y}}^\e)^{-1} \|_2 \, \| \nabla_{\boldsymbol{\tilde{y}}^\e} \mathfrak{F}(\boldsymbol{v}^\e, \boldsymbol{\tilde{y}}^\e) \|_2 + \| \nabla_{\boldsymbol{v}^\e} \mathfrak{F}(\boldsymbol{v}^\e, \boldsymbol{\tilde{y}}^\e)^{-1} \|_2\\
    &\leq 1 + C \,\| \nabla_{\boldsymbol{v}^\e} \mathfrak{F}(\boldsymbol{v}^\e, \boldsymbol{\tilde{y}}^\e)^{-1} \|_\infty.
\end{aligned}
\end{equation*}
In view of Theorem \ref{paper 3 inverse_norm_matrix}, one can obtain a uniform lower bound of the singular value from
\begin{equation}
\begin{aligned}
    \underline{\sigma} (\boldsymbol{v}^\e, \boldsymbol{\tilde{y}}^\e) &\geq \big( 1 + C \, \| \nabla_{\boldsymbol{v}^\e} \mathfrak{F}(\boldsymbol{v}^\e, \boldsymbol{\tilde{y}}^\e)^{-1} \|_\infty \big)^{-1}\\
    &\geq \bigg( 1 + C \, \min_i \Big\{ \, 1 - \sum_{k \neq i} \Big| \frac{ \partial^2 \lambda^i}{\partial u^{i} \, \partial u^{k}} \, \Big/ \, \frac{ \partial^2 \lambda^i}{\partial (u^{i})^2} \Big| \, \Big\}^{-1} \bigg)^{-1}\\
    &\geq (1+C)^{-1},
\end{aligned}
\end{equation}
where we have applied the fact that the diagonal dominance `gap' of $\nabla_{\boldsymbol{v}^\e}\mathfrak{F}(\boldsymbol{v}^\e, \boldsymbol{\tilde{y}}^\e)$ is bounded away from $0$. The discussion on the generalized Jacobian follows a similar routine. The key is to observe that the diagonal dominance `gap' of the generalized Jacobian $\partial_{\boldsymbol{v}^\e}\mathfrak{F}$ remains bounded away from $0$ under any convex combination.

By the implicit function theorem, there exists Lipschitz mapping $\varphi: \mathbb{R}^N \to \mathbb{R}^N$ such that $\mathfrak{F}(\boldsymbol{v}^\e, \boldsymbol{\tilde{y}}^\e) = 0$ is equivalent to $\mathfrak{F}(\varphi(\boldsymbol{\tilde{y}}^\e), \boldsymbol{\tilde{y}}^\e) = 0$. In accord with the definition of $\tilde{\boldsymbol{y}}^\e$, we can equivalently write $\varphi(\tilde{\boldsymbol{y}}^\e)$ as $\varphi(\boldsymbol{y}^\e, \boldsymbol{Y}_t^\m)$, where $\varphi$ now maps $\mathbb{R}^N \times \mathbb{R}^{\tilde{N}}$ to $\mathbb{R}^N$ by a slight abuse of notation. It is still Lipschitz in both $\boldsymbol{y}^\e$ and $\boldsymbol{Y}_t^\m$, after recognizing that the $\min$ function is Lipschitz. The stochastic maximum principle then suggests the FBSDE
\begin{equation*}
        \left\{
    \begin{aligned}
     dQ_t^{i, \e} &= \varphi^i(\boldsymbol{Y}_t^\e, \boldsymbol{Y}_t^\m) \, dt,\\
    \, dY_t^{i, \e} &= 2\phi_t^{i, \e} \, Q_t^{i, \e} \, dt + dM_t^{i, \e},\\
    Q_0^{i, \e} &= q_0^{i, \e}, \quad Y_T^{i, \e} = -2A^{i,\e}\,Q_T^{i, \e},
    \end{aligned}
    \right.
\end{equation*}
and that the equilibrium strategies satisfy
\begin{equation*}
    v^{i, \e}_t = \varphi^i(\boldsymbol{Y}_t^\e, \boldsymbol{Y}_t^\m)
\end{equation*}
for any $i$. A fixed point argument finally gives the system \eqref{paper 3 big FBSDE market maker}-\eqref{paper 3 big FBSDE optimal execution} and we complete the proof for the necessary condition.

With respect to the sufficient condition, consider a strategy profile $( \boldsymbol{v}^\e, (\boldsymbol{\delta}^j)_{j})$ specified by \eqref{general game feedback}. Based on the stochastic maximum principle, the strategies \((\boldsymbol{\delta}^j)_{j}\) constitute best responses for the market makers when \(\boldsymbol{v}^\e\) is fixed, and conversely. Hence the profile forms a Nash equilibrium, which completes the proof.
\end{proof}

\kong

\subsection{Proof of Proposition \ref{local wellpose}}

\begin{proof}
    It suffices to verify that forward equations of \eqref{paper 3 big FBSDE market maker} and \eqref{paper 3 big FBSDE optimal execution} are Lipschitz with respect to $\boldsymbol{y}^\e$ and $\boldsymbol{y}^\m$. Local well-posedness then follows from standard results for Lipschitz FBSDEs. For the optimal execution part, it is straightforward since $\varphi$ is a Lipschitz mapping. In regard to the market making side, if we consider $\hat{a}_t$ and $\hat{b}_t$ as functions of $\boldsymbol{y}^\e$ and $\boldsymbol{y}^\m$, they are both Lipschitz due to the property of $\varphi$. In view of the properties that: (1) function $\psi^a$ and $\min$ are both Lipschitz and (2) function $\psi^a$ takes value in $[-\xi, \xi]^{\tilde{N}}$, we can see that $\Lambda \big( \psi^{i, a} (\boldsymbol{y}^\m) - \bar{\psi}^{i, a} (\boldsymbol{y}^\m) \big)$ is Lipschitz in $\boldsymbol{y}^\m$. Observe that $|\hat{a}_t| \leq C + N \, \tilde{\xi}$ for some constant $C$ and $| \, \Lambda \big(\psi^{i,a}(\boldsymbol{y}^\m) - \bar{\psi}^{i,a}(\boldsymbol{y}^\m) \big) \, | \leq \Lambda(-2\xi)$. The property that the product of bounded Lipschitz functions is still Lipschitz completes the proof.
\end{proof}

\vspace{0.2cm}

\section{Appendix: Proofs in Section \ref{paper 3 section 4}}

\subsection{Proof of Proposition \ref{paper 3 first character}}

\begin{proof}
Suppose $(v^{1, \e}, (\boldsymbol{\delta}^j)_{j = 1}^2)$ is a Nash equilibrium. Provided with the trading strategy $v^{1, \e}$, the two market makers solve the market making game. The Hamiltonian of market maker $(i, \m)$ reads
\begin{equation*}
\begin{aligned}
    &\Big[ -\tilde{a}_t \, \exp \big( - \gamma( \delta^{i, a} - \bar{\delta}^{i, a} ) \big)  + b_t \, \exp \big( - \gamma (\delta^{i, b} - \bar{\delta}^{i, b})  \big) \Big] \, y^{i, \m}\\
    &\hspace{2cm} + \delta^{i, a} \, \tilde{a}_t \, \exp \big( - \gamma( \delta^{i, a} - \bar{\delta}^{i, a}) \big)  + \delta^{i, b} \, b_t \, \exp \big( - \gamma(\delta^{i, b} - \bar{\delta}^{i, b} ) \big)  - \phi_t^{i, \m} \, \big( q^{i, \m} \big)^2.
\end{aligned}
\end{equation*}
When maximizing the Hamiltonian, the tractability of exponential functions allows $(\delta^{i, a}, \delta^{i, b})$ to depend only on $y^{i, \m}$:
\begin{equation*}
    \hat{\delta}^{i, a} = \bigg( \frac{1}{\gamma} + y^{i, \m} \bigg) \vee (-\xi) \wedge \xi \text{\; and \;} \delta^{i, b} = \bigg( \frac{1}{\gamma} - y^{i, \m} \bigg) \vee (-\xi) \wedge \xi.
\end{equation*}
The stochastic maximum principle implies that $(Y^{j,\m})_{j=1}^2$ satisfies an FBSDE system; the forward equation of this system for agent $(1, \m)$ reads
\begin{equation}
\begin{aligned}
     dQ_t^{1, \m} = - \tilde{a}_t \, &\exp \bigg\{-\gamma \, \Big[ \Big(\frac{1}{\gamma} + Y_t^{1, \m} \Big) \vee (-\xi) \wedge \xi - \Big(\frac{1}{\gamma} + Y_t^{2, \m} \Big) \vee (-\xi) \wedge \xi \Big] \bigg\} \, dt,\\
     &+ b_t \, \exp \bigg\{-\gamma \, \Big[ \Big(\frac{1}{\gamma} - Y_t^{1, \m} \Big) \vee (-\xi) \wedge \xi - \Big(\frac{1}{\gamma} - Y_t^{2, \m} \Big) \vee (-\xi) \wedge \xi \Big] \bigg\} \, dt.
    \label{paper 3 big FBSDE _1}
\end{aligned}
\end{equation}
The forward equation of agent $(2, \m)$ is symmetric. Here, the backward equation is omitted for conciseness. Due to the homogeneity of penalty parameters, the ordering property in the market making games---Lemma 5.2 in \cite{guo2024macroscopicmarketmakinggames}---yields 
\begin{equation*}
Y_t^{1, \m} \leq Y_t^{2, \m}
\end{equation*}
for all $t$. 

Denote by $(\bar a, \bar b, \bar \phi, \bar A)$ the upper bounds of the order-flow and penalty parameters $(a, b, \phi, A)$, respectively. Based on the estimation:
\begin{equation*}
\begin{aligned}
    Q_t^{1, \m} &\leq q_0^{1,m} + \int_0^t  b_u \, \exp \bigg\{-\gamma \, \Big[ \Big(\frac{1}{\gamma} - Y_u^{1, \m} \Big) \vee (-\xi) \wedge \xi - \Big(\frac{1}{\gamma} - Y_u^{2, \m} \Big) \vee (-\xi) \wedge \xi \Big] \bigg\} \, du,\\
    & \leq q_0^{1, \m} + \bar{b} \, T, 
\end{aligned}
\end{equation*}
we can obtain that
\begin{align*}
    Y_t^{1, \m} &= -\mathbb{E}_t \bigg[ 2A \,Q_T^{1, \m} + 2 \int_t^T \phi_u \, Q_u^{1, \m} \, du \bigg]\\
    &\geq -(2\bar{A} + 2 \bar{\phi} \, T) \, \big(|q_0^{1, \m}| + \bar{b} \, T \big).
\end{align*}
We see that $Y^{1, \m}$ is bounded below by some constant independent of $\xi$ and $\tilde{\xi}$. Therefore, we first set $\xi \geq \gamma^{-1} + (2\bar{A} + 2 \bar{\phi} \, T) \,(|q_0^{1, \m}| + \bar{b} \, T )$. Such a large choice of $\xi$ yields
\begin{equation*}
    -\xi \leq \frac{1}{\gamma} + Y_t^{1, \m} \leq \frac{1}{\gamma} + Y_t^{2, \m} \text{\; and \;} \xi \geq \frac{1}{\gamma} - Y_t^{1, \m} \geq \frac{1}{\gamma} - Y_t^{2, \m}
\end{equation*}
for all $t$. In other words, half of the truncations no longer have any effect. 

On the other hand, equipped with $(Y^{k, \m})_{k=1}^2$, the trader is dealing with an optimal execution problem. Examining her Hamiltonian 
\begin{equation*}
\begin{aligned}
    v^{1, \e} \, y^{1, \e} - v^{1, \e} \, \min_{1 \leq j \leq 2} \delta^{j, a} &- \beta_t \, \big( v^{1, \e} \big)^2 - \phi_t^{i, \e} \, \big( q^{i,\e} \big)^2\\
    &= v^{1,\e} \, y^{1,\e} - v^{1, \e} \, \Big[ \Big( \frac{1}{\gamma} + Y_t^{1, \m} \Big) \wedge \xi \Big] - \beta_t \, \big( v^{1, \e} \big)^2 - \phi_t^{i, \e} \, \big( q^{i,\e} \big)^2,
\end{aligned}
\end{equation*}
the optimal feedback control reads
\begin{equation*}
    v^{1, \e} = \frac{1}{2 \, \beta_t} \, \Big[ y^{1, \e} - \Big(\frac{1}{\gamma} + Y_t^{1, \m} \Big) \wedge \xi \Big] \vee 0 \wedge \tilde{\xi}.
\end{equation*}
The stochastic maximum principle implies that $Y^{1, \e}$ solves an FBSDE with the forward equation
\begin{equation}
     dQ_t^{1,\e} = \frac{1}{2 \, \beta_t} \, \Big[ Y_t^{1, \e} - \Big(\frac{1}{\gamma} + Y_t^{1, \m} \Big) \wedge \xi \Big] \vee 0 \wedge \tilde{\xi} \, dt.
    \label{paper 3 big FBSDE _2}
\end{equation}
Since $Q^{1, \e}$ is increasing, a similar argument yields that $Y^{1, \e}$ is upper bounded by some constant independent of $\xi$ and $\tilde{\xi}$:
\begin{equation*}
    Y_t^{1, \e} \leq (2\bar{A} + 2 \bar{\phi} \, T) \, |q_0^{1, \e}|.
\end{equation*}
Combined with the lower bound of $Y^{1,\m}$, we are able to let the truncation $\tilde{\xi}$ have no effect by setting
\begin{equation*}
    \tilde{\xi} \geq \frac{1}{2 \, \underline{\beta}}  \bigg[ \frac{1}{\gamma} + (2\bar{A} + 2\bar{\phi} \, T) \, \big( |q_0^{1,\e}| + |q_0^{1, \m}| + \bar{b} \, T \big) \bigg],
\end{equation*}
where $\underline{\beta}$ denotes the lower bound of the process $\beta$. Note that such choice of $\tilde{\xi}$ is independent of $\xi$. Consequently, the resulting order flow $v^{1,\e} + a$ is also bounded above by some constant independent of $\xi$. Therefore, Proposition 5.3 in \cite{guo2024macroscopicmarketmakinggames} enables us to further pick $\xi$ sufficiently large to additionally ensure that
\begin{equation*}
    \xi \geq \frac{1}{\gamma} + Y_t^{2, \m} \geq \frac{1}{\gamma} + Y_t^{1, \m} \text{ \; and \; } -\xi \leq \frac{1}{\gamma} - Y_t^{2, \m} \leq \frac{1}{\gamma} - Y_t^{1, \m}
\end{equation*}
for all $t$. System \eqref{paper 3 big FBSDE _1}-\eqref{paper 3 big FBSDE _2} finally leads to \eqref{paper 3 big FBSDE 1}-\eqref{paper 3 big FBSDE 2} since both $\xi$ and $\tilde{\xi}$ have been removed.
\end{proof}

\kong

\subsection{Proof of Proposition \ref{Riccati with M matrix}}
A mean value theorem from non-smooth analysis is stated below and will be used in the proof.

\kong

\begin{theorem}[\cite{clarke1990optimization}]
\label{mean value and calculus}
Let $F$ be Lipschitz on an open convex set $U$ in $\mathbb{R}^n$, and let $x$ and $y$ be two points in $U$. Then, it holds that
\begin{equation*}
F(y) - F(x) \in \textnormal{co} \Big\{ K \, (y - x) \, : \, K \in \partial F(z) \textnormal{ and } z \textnormal{ is in the line segment between $x$ and $y$} \Big\}.
\end{equation*}
The notation $\textnormal{co}(A)$ denotes the convex hull generated by the set $A$
\end{theorem}

\kong

\begin{proof}
    Given $q_0^{1, \m} \geq  q_0^{2, \m}$, it can be deduced from Lemma 5.2 in \cite{guo2024macroscopicmarketmakinggames} that $Q_t^{1, \m} \geq Q_t^{2, \m}$ and $Y_t^{1, \m} \leq Y_t^{2, \m}$ for any $t$ as long as the solution exists. Equivalently, if $\mathcal{Q}_0^{2, \m} \geq 0$, then $\mathcal{Q}_t^{2, \m} \geq 0$ and $\mathcal{Y}_t^{2, \m} \leq 0$ for any $t$. Therefore, it suffices to consider the case where the initial condition of $\mathcal{Q}^{2, \m}$ is non-negative. Moreover, we only need to study $(\mathcal{Q}^{2, \m}, \mathcal{Y}^{2, \m}, \mathcal{M}^{2, \m})$ and $(\mathcal{Q}^{1, \e}, \mathcal{Y}^{1, \e}, \mathcal{M}^{1, \e})$, provided with which $(\mathcal{Q}^{1, \m}, \mathcal{Y}^{1, \m}, \mathcal{M}^{1, \m})$ can be immediately obtained. Let us introduce the following notations:
    \begin{align*}
        \mathcal{U}^1(t, \mathcal{Y}^{2, \m}, \mathcal{Y}^{1, \e}) &= \bigg(a_t + b_t + \frac{1}{2 \, \beta_t} \, \Big( \mathcal{Y}^{1, \e} - \frac{1}{\gamma} \Big) \vee 0 \bigg) \, \Big( \exp(\gamma \, \mathcal{Y}^{2, \m}) - \exp(-\gamma \, \mathcal{Y}^{2, \m}) \Big),\\
        \mathcal{U}^2(t, \mathcal{Y}^{2, \m}, \mathcal{Y}^{1, \e}) &= \bigg( \frac{1}{2 \, \beta_t} \, \Big( \mathcal{Y}^{1, \e} - \frac{1}{\gamma} \Big) \vee 0 \bigg) \, \Big( 1 + \exp( - \gamma \, \mathcal{Y}^{2, \m}) \Big)\\
    &\hspace{4.5cm} + a_t \, \exp( - \gamma \, \mathcal{Y}^{2, \m}) - b_t \, \exp( \gamma \, \mathcal{Y}^{2, \m}),\\
     \mathcal{U}(t, \mathcal{Y}^{2, \m}, \mathcal{Y}^{1, \e}) &= \begin{bmatrix}
        \mathcal{U}^1\\
        \mathcal{U}^2
    \end{bmatrix}(t, \mathcal{Y}^{2, \m}, \mathcal{Y}^{1, \e}).
    \end{align*}
    
    Due to its Lipschitz property, for some $s > 0$ sufficiently close to $T$, we know the FBSDE \eqref{paper 3 equiv big FBSDE} is well-posed on $[s, T]$, that is
    \begin{align*}
    d\mathcal{Q}_t^{2, \m} &= \mathcal{U}^1(t, \mathcal{Y}_t^{2, \m}, \mathcal{Y}_t^{1, \e}) \, dt,\\
    d\mathcal{Q}_t^{1, \e} &= \mathcal{U}^2(t, \mathcal{Y}_t^{2, \m}, \mathcal{Y}_t^{1, \e}) \, dt,
\end{align*}
with $\mathcal{Q}_s^{2, \m} = \iota^{2, \m} \geq 0$ and $\mathcal{Q}_s^{1, \e} = \iota^{1, \e}$. Consequently, given any initial conditions $\mathcal{Q}_{\tau}^{2, \m} = \iota^{2, \m}$ and $\mathcal{Q}_{\tau}^{1, \e} = \iota^{1, \e}$ at time $\tau \in [s, T]$, there exist unique $\mathcal{Y}_\tau^{2, \m}$ and $\mathcal{Y}_\tau^{1, \e}$. We denote this mapping by $\theta: [s, T] \times \Omega \times (\mathbb{R}_+ \times \mathbb{R})$ defined as
\begin{equation*}
    \theta\big(\tau, (\iota^{2, \m}, \iota^{1, \e}) \big) := (\mathcal{Y}_\tau^{2, \m}, \mathcal{Y}_\tau^{1, \e}),
\end{equation*}
which is known as the decoupling field; see Definition 5.7 in \cite{guo2024macroscopicmarketmakinggames}. The focus is on the Lipschitz property of $\theta$.

Given two initial conditions $(\iota^{2, \m}, \iota^{1, \e})$ and $(\tilde{\iota}^{2, \m}, \tilde{\iota}^{1, \e})$ with $\iota^{2, \m}, \tilde{\iota}^{2, \m} \geq 0$, denote by $(\mathcal{Q}^{2, \m}, \mathcal{Y}^{2, \m})$, $(\mathcal{Q}^{1, \e}, \mathcal{Y}^{1, \e})$ and $(\tilde{\mathcal{Q}}^{2, \m}, \tilde{\mathcal{Y}}^{2, \m})$, $(\tilde{\mathcal{Q}}^{1, \e}, \tilde{\mathcal{Y}}^{1, \e})$ the corresponding local solutions. Let us write
\begin{gather*}
    \mathscr{Q} := \begin{bmatrix}
        \tilde{\mathcal{Q}}^{2, \m}\\
        \tilde{\mathcal{Q}}^{1, \e}
    \end{bmatrix} - \begin{bmatrix}
        \mathcal{Q}^{2, \m} \\
        \mathcal{Q}^{1, \e}
    \end{bmatrix} \text{\; and \;} \mathscr{Y} := \begin{bmatrix}
        \tilde{\mathcal{Y}}^{2, \m} \\
        \tilde{\mathcal{Y}}^{1, \e}
    \end{bmatrix} - \begin{bmatrix}
        \mathcal{Y}^{2, \m} \\
        \mathcal{Y}^{1, \e}
    \end{bmatrix}.
\end{gather*}
We aim to derive a linear FBSDE for $(\mathscr{Q}, \mathscr{Y})$. The Jacobian matrix of $\mathcal{U}$ with respect to $(y^{2, \m}, y^{1, \e})$ is either
\begin{equation}
    \begin{bmatrix}
    \gamma \, \big( \breve{a}_t + b_t \big) \, \big( \exp(\gamma \, y^{2, \m}) + \exp(-\gamma \, y^{2, \m}) \big) & - \frac{1}{2 \, \beta_t} \, \big( \exp(-\gamma \, y^{2, \m}) - \exp(\gamma \, y^{2, \m}) \big)
    \\
    -\gamma \, \breve{a}_t \, \exp(-\gamma \, y^{2, \m}) - \gamma \, b_t \, \exp(\gamma \, y^{2, \m}) & \frac{1}{2 \, \beta_t} \, \big( 1 + \exp(-\gamma \, y^{2, \m}) \big)
    \label{paper 3 jacob matrix 1}
    \end{bmatrix}
\end{equation}
when $\frac{1}{2 \, \beta_t} \, ( y^{1, \e} - \frac{1}{\gamma} ) \geq 0$, or otherwise
\begin{equation}
    \begin{bmatrix}
    \gamma \, \big(a_t + b_t\big) \, \big( \exp(\gamma \, y^{2, \m}) + \exp(-\gamma \, y^{2, \m}) \big) & 0
    \\
    -\gamma \, a_t\, \exp(- \gamma \, y^{2, \m}) - \gamma \, b_t \, \exp(\gamma \, y^{2, \m}) & 0
    \end{bmatrix},
    \label{paper 3 jacob matrix 2}
\end{equation}
where $\breve{a}_t = a_t + \frac{1}{2 \, \beta_t} \, ( y^{1, \e} - \frac{1}{\gamma} ) \vee 0$. Because $\iota^{2, \m}, \tilde{\iota}^{2, \m} \geq 0$, it follows for all $t$ that $\max(\mathcal{Y}_t^{2, \m}, \tilde{\mathcal{Y}}_t^{2, \m}) \leq 0$. For any $(y^{2, \m}, y^{1, \e})$ that belongs to the line segment between $(\mathcal{Y}_t^{2, \m}, \mathcal{Y}_t^{1, \e})$ and $(\tilde{\mathcal{Y}}_t^{2, \m}, \tilde{\mathcal{Y}}_t^{1, \e})$, it is straightforward to check that both \eqref{paper 3 jacob matrix 1} and \eqref{paper 3 jacob matrix 2} are $M_+$-matrices, due to the non-negative column sums. We apply the mean value theorem \ref{mean value and calculus} to obtain 
\begin{equation*}
    \mathcal{U}(t, \tilde{\mathcal{Y}}_t^{2, \m}, \tilde{\mathcal{Y}}_t^{1, \e}) - \mathcal{U}(t, \mathcal{Y}_t^{2, \m}, \mathcal{Y}_t^{1, \e}) = G_t \, \mathscr{Y}_t,
\end{equation*}
where $G_t$ is an $M_+$-matrix for any $t\in[s, T]$. Indeed, the non-negativity of column sums remains under any convex combinations. 

The linear system for $(\mathscr{Q}, \mathscr{Y})$ hence reads
\begin{equation}
\left\{
    \begin{aligned}
     d\mathscr{Q}_t &= G_t \, \mathscr{Y}_t \, dt,\\
    \, d\mathscr{Y}_t &= 2 \phi_t \, \mathscr{Q}_t \, dt + d\mathscr{M}_t,\\
    \mathscr{Q}_s &= \begin{pmatrix}
        \tilde{\iota}^{2, \m} - \iota^{2, \m}\\
        \tilde{\iota}^{1, \e} - \iota^{1, \e}
    \end{pmatrix}, \quad \mathscr{Y}_T = -2A \, \mathscr{Q}_T.
    \end{aligned}
    \right.
    \label{variational FBSDE}
\end{equation}
This is known as the \textit{variational FBSDE}, introduced by \cite{ma2015well} for the one-dimensional cases. We proceed to show the existence and uniqueness of the solution. Consider the affine ansatz
\begin{equation}
    \mathscr{Y}_t = R_t \, \mathscr{Q}_t.
    \label{affine ansatz for vari}
\end{equation}
Consequently, the BSRE \eqref{paper 3 M+ BSRE} for $R$ is obtained by matching the coefficients. Considering equation \eqref{paper 3 M+ BSRE}, if there exists a unique bounded solution $R$ on $[s, T]$, then a solution of the system \eqref{variational FBSDE} can be constructed by plugging \eqref{affine ansatz for vari} back to \eqref{variational FBSDE}. For uniqueness, due to the boundedness of $R$, the linear form \eqref{affine ansatz for vari} defines a regular decoupling field for \eqref{variational FBSDE}; see Definition 5.8 in \cite{guo2024macroscopicmarketmakinggames}. Subsequently, the existence of a regular decoupling field guarantees the uniqueness of the solution, by Theorem 5.9 in \cite{guo2024macroscopicmarketmakinggames}.

Given the well-posedness of \eqref{variational FBSDE}, the decoupling field $\theta$ is regular on $[s, T]$, fulfilling the integrability and uniform Lipschitz continuity properties. The integrability can be deduced by the boundedness of the solution. To derive the Lipschitz property, by the definitions of $\mathscr{Y}_t$ and $\mathscr{Q}_t$ we can have
\begin{equation}
    \big| \,\theta\big(s, (\tilde{\iota}^{2, \m}, \tilde{\iota}^{1, \e}) \big) - \theta\big(s, (\iota^{2, \m}, \iota^{1, \e}) \big) \, \big| = | \mathscr{Y}_s | \leq \|R_s\|_2 \, \big| (\tilde{\iota}^{2, \m}, \tilde{\iota}^{1, \e}) -(\iota^{2, \m}, \iota^{1, \e}) \big|,
    \label{regular decoupling}
\end{equation}
for any $(\iota^{2, \m}, \iota^{1, \e})$ and $(\tilde{\iota}^{2, \m}, \tilde{\iota}^{1, \e})$ with $\iota^{2, \m}, \tilde{\iota}^{2, \m} \geq 0$. Here, the notation $\| \cdot \|_2$ denotes the spectral norm of a matrix.

A regular decoupling field yields the well-posedness of \eqref{paper 3 equiv big FBSDE} on $[s, T]$. According to Theorem 5.10 in \cite{guo2024macroscopicmarketmakinggames}, the only case when the decoupling field can not be extended to the whole interval $[0, T]$ is that its Lipschitz coefficient may explode if $s$ approaches some $t_{\min} \geq 0$. However, this cannot occur if \eqref{paper 3 M+ BSRE} admits a unique bounded solution on the whole interval $[0, T]$, as shown in \eqref{regular decoupling}. Therefore, the well-posedness of \eqref{paper 3 M+ BSRE} ensures that a regular decoupling field can be extended to the entire interval $[0, T]$, which in turn implies the well-posedness of \eqref{paper 3 equiv big FBSDE}.
\end{proof}

\kong

\subsection{Proof of Theorem \ref{global well-posed M+}}
We first introduce the Radon's lemma. Consider the following linear matrix differential equations:
\begin{equation}
\begin{bmatrix}
    V'(t) \\
    U'(t)
\end{bmatrix}
=
\begin{bmatrix}
    0 & G_t\\
    2\phi_t\, I & 0
\end{bmatrix} 
\begin{bmatrix}
    V(t) \\
    U(t)
\end{bmatrix}
; \quad
\begin{bmatrix}
    V(T) \\
    U(T)
\end{bmatrix}
=
\begin{bmatrix}
    I \\
    -2A \, I
\end{bmatrix}.
    \label{radon system}
\end{equation}
The purpose of Radon's lemma is to establish a connection between the matrix Riccati equation \eqref{paper 3 M+ BSRE} and the linear equation \eqref{radon system}.

\kong

\begin{lemma}[\cite{freiling2002survey}]
\label{radon lemma}
(1) Let $R$ be a solution
of the deterministic Riccati equation \eqref{paper 3 M+ BSRE} on some interval $\mathcal{J} \subseteq [0, T]$ such that $T\in\mathcal{J}$. Denote by $V$ the unique solution of the linear equation
\begin{equation*}
    V'(t) = G_t \, R_t \, V(t); \quad V(T) = I,
\end{equation*}
for $t\in \mathcal{J}$ and set $U(t) = R_t\, V(t)$. Then, the matrix $\begin{bmatrix}
    V(t) \\ U(t)
\end{bmatrix}$
defines for $t \in \mathcal{J}$ the solution of the linear differential equation \eqref{radon system}.

(2) Suppose $\begin{bmatrix}
    V(t) \\ U(t)
\end{bmatrix}$
is on some interval $\mathcal{J} \subseteq [0, T]$ a solution of the linear differential equation \eqref{radon system} such that $\det V(t) \neq 0$ for all $t \in \mathcal{J}$, then
\begin{equation}
    R: \mathcal{J}\to \mathbb{R}^{2 \times 2}, \quad t \mapsto U(t) \, V(t)^{-1} =: R_t
    \label{radon bridge}
\end{equation}
is a solution to the deterministic Riccati equation \eqref{paper 3 M+ BSRE} on $\mathcal{J}$.
\end{lemma}

\kong

Radon’s lemma reveals that the matrix Riccati equation \eqref{paper 3 M+ BSRE} is locally equivalent to the linear differential equation \eqref{radon system}. This equivalence persists until a potentially finite blow-up time, as indicated by \eqref{radon bridge}. Furthermore, it is evident from \eqref{radon bridge} that the solution $R$ of the Riccati equation experiences blow-ups at moments when the determinant of $V$ vanishes. 

\kong

\begin{proof}
\noindent Because we are considering the deterministic case, then the coefficient $G$ is also non-random. Let us introduce the following notations:
\begin{equation*}
    R(t) = \begin{bmatrix}
        r_1(t) & r_2(t)\\
        r_3(t) & r_4(t)
    \end{bmatrix} \text{\quad and \quad}
    G(t) = \begin{bmatrix}
        g_1(t) & -g_2(t)\\
        -g_3(t) & g_4(t)
    \end{bmatrix}. 
\end{equation*}
By the properties of $M_+$-matrix, for any $i$ and $t$,  we know that $g_i(t) \geq 0$, $g_1(t) \geq g_3(t)$, and $g_4(t) \geq g_2(t)$. For notational convenience, we will omit the explicit dependence on $t$ whenever feasible. Since the local well-posedness of \eqref{paper 3 M+ BSRE} is well-known, we begin with a priori estimate. We will first prove the following claim for the solution as long as it exists:
\begin{itemize}
    \item[(1)] The matrix $R$ is element-wise non-positive;\\
    
    \vspace{-0.4cm}

    \item[(2)] It is also diagonally dominant in row, i.e., $r_1 \leq r_2$ and $r_4 \leq r_3$.  
\end{itemize}

Since we are looking backward, we explicitly compute the backward drift:
\begin{align*}
    &R \, G \, R - 2\phi_t \, I\\
    &= \begin{bmatrix}
        r_1  (g_1 r_1 - g_3 r_2) + r_3 (g_4 r_2 - g_2 r_1) - 2\phi_t & r_2  (g_1 r_1 - g_3 r_2) + r_4 (g_4 r_2 - g_2 r_1)\\
        r_1  (g_1 r_3 - g_3 r_4) + r_3 (g_4 r_4 - g_2 r_3) & r_2  (g_1 r_3 - g_3 r_4) + r_4 (g_4 r_4 - g_2 r_3) - 2\phi_t
    \end{bmatrix}.
\end{align*}
Let us consider the scenario when $r_2 = 0$. The corresponding drift for $r_2$ (the upper-right element of above matrix) reads $-g_2 \, (r_1 \, r_4)$. Because $r_2(T) = 0$ and $r_1(T) = r_4(T) = -2A <0$, it implies that $r_2$ will never go above $0$ as long as both $r_1$ and $r_4$ remain non-positive. A similar implication holds for $r_3$.

The only concern for above implications is that $r_1$ or $r_4$ may become positive. We take $r_1$ as the case for discussion. Due to the continuity of paths and $r_1(T) < r_2(T)$, it must hold that $r_1(t) = r_2(t) \leq 0$ in order for $r_1$ to go above $0$, for some $t < T$. Suppose that $r_1 = r_2$ no later than when $r_3$ meets $r_4$. In that case of $r_1 = r_2$, denote by $\mu_1$ and $\mu_2$ the backward drifts of $r_1$ and $r_2$ accordingly. Note that
\begin{align*}
    \mu_1 - \mu_2 &= (r_1)^2 (g_1 - g_3) + r_1 r_3 (g_4 -g_2) - 2\phi_t - (r_1)^2 (g_1 - g_3) - r_1 r_4 (g_4 - g_2)\\
    &= r_1 (r_3 - r_4) (g_4 - g_2) - 2\phi_t\\
    &\leq 0.
\end{align*}
The last inequality holds because $r_1 \leq 0$ and $r_3 - r_4 \geq 0$. Since the case of $r_3 = r_4$ can be discussed similarly, we conclude that $r_1 \leq r_2 \leq 0$ and $r_4 \leq r_3 \leq 0$ for all $t$. Hence, we have proved the claim proposed. 

We then apply Lemma \ref{radon lemma}. It suffices to study the determinant of $V$. Define $v(t) := \text{det\,}V(t)$. By the Jacobi's formula, we obtain
\begin{align*}
    v'(t) &= v(t) \, \text{tr} \big(V(t)^{-1} V'(t) \big)\\
    &=v(t) \, \text{tr} \big(G(t) \, U(t) \, V(t)^{-1} \big)\\
    &=v(t) \, \text{tr} \big(G(t) \, R(t)\big).
\end{align*}
Using the notations defined for matrices $G$ and $R$, the trace term reads
\begin{align*}
    \text{tr} \big(G(t) \, R(t)\big) &= g_4 r_4 - g_2 r_3 + g_1 r_1 - g_3 r_2\\
    &\leq g_2 (r_4 - r_3) + g_3 (r_1 - r_2)\\
    &\leq 0,
\end{align*}
where we have applied the properties of $G$ and our claims for $R$. Combined with the terminal condition $v(T) = 1$, we can see that $v(t) \geq 1$ for any $t < T$. Hence, the non-singular property of $V$ yields the well-posedness of \eqref{paper 3 M+ BSRE}.
\end{proof}

\vspace{0.2cm}

\section{Appendix: Proofs in Section \ref{paper 3 section 5}}

\subsection{Proof of Proposition \ref{sec 5 sto max prin}}

\begin{proof}
In view of objective functionals, the addition of diffusion terms invites an extra term that is indifferent to the strategies. Combined with the fact that diffusion terms are not controlled, the Hamiltonian of each agent remains the same. The proof basically follows the line of Theorem \ref{paper 3 big maximum principle}, while the only difference lies in the new function $(\zeta^i)_{i}$.

We start with the necessary condition. Let the strategy profile $(\boldsymbol{v}^\e, (\boldsymbol{\delta}^j)_{j} )$ be a Nash equilibrium. Given the trading strategies $\boldsymbol{v}^\e$, market makers solve a market-making game. The Hamiltonian of market maker $(i, \m)$ reads
\begin{equation}
    \Big( - \tilde{a}_t \, \zeta^i \big( \boldsymbol{\delta}^{a} \big) + \tilde{b}_t \, \zeta^i \big(\boldsymbol{\delta}^{b} \big) \Big) \, y^{i, \m} + \delta^{i, a} \, \tilde{a}_t \, \zeta^i \big( \boldsymbol{\delta}^{a} \big) + \delta^{i, b} \, \tilde{b}_t \, \zeta^i \big( \boldsymbol{\delta}^{b} \big) - \phi_t^{i, \m} \, \big( q^{i, \m} \big)^2.
    \label{approx hamilton}
\end{equation}
Take the ask side for example. Equivalently, we intend to find $\delta^{i, a}$ that maximizes 
\begin{equation*}
    - \zeta^i \big( \boldsymbol{\delta}^{a} \big) \, y^{i, \m} + \delta^{i,a} \, \zeta^i \big( \boldsymbol{\delta}^{a} \big),
\end{equation*}
of which the derivative with respect to $\delta^{i,a}$ reads
\begin{equation}
    \frac{\partial \zeta^i}{\partial \delta^{i}} \big( \boldsymbol{\delta}^{a} \big) \, \Big[ \delta^{i, a} + \frac{ \zeta^i \big( \boldsymbol{\delta}^{a} \big)}{\frac{\partial \zeta^i}{\partial \delta^{i} } \big( \boldsymbol{\delta}^{a} \big)} - y^{i, \m} \Big].
    \label{paper 3 maker deriva}
\end{equation}
Properties of $\zeta^i$ infers that the square bracket part of \eqref{paper 3 maker deriva} is increasing in $\delta^{i, a}$. Hence, the first-order condition and the truncation by $\xi$ yield 
\begin{equation}
    \delta^{i,a} - \Big[ y^{i,\m} - \frac{ \zeta^i \big( \boldsymbol{\delta}^{a}\big) }{ \frac{ \partial \zeta^i }{\partial \delta^{i} } \big(\boldsymbol{\delta}^{a} \big)} \Big] \vee (-\xi) \wedge \xi = 0.
    \label{paper 3 maker hamilt}
\end{equation}

Because \eqref{paper 3 maker hamilt} should apply for every $i$, we have
\begin{equation}
    \mathfrak{F}^\m(\boldsymbol{\delta}^a, \boldsymbol{y}^\m) = 0,
    \label{paper 3 maker implicit function}
\end{equation}
where the $i$-th entry of function $\mathfrak{F}^\m: \mathbb{R}^{\tilde{N}} \times \mathbb{R}^{\tilde{N}} \to \mathbb{R}^{\tilde{N}}$ is given by the left hand side of \eqref{paper 3 maker hamilt}. In view of the proof of Theorem \ref{paper 3 big maximum principle}, to apply the implicit function theorem \ref{general_implicit}, it suffices to prove that $\nabla_{\boldsymbol{\delta}^a}\mathfrak{F}^\m (\boldsymbol{\delta}^a, \boldsymbol{y}^\m)$
has positive diagonals and the diagonal dominance gap is bounded away from $0$,  whenever differentiable. Cases when truncation $\pm \xi$ takes into effect are straightforward. Let us then pay attention to the untruncated case. Observe that
\begin{equation*}
    \frac{\partial \mathfrak{F}^{i, \m}}{\partial \delta^{i}} = 2 - \frac{\zeta^i \big( \boldsymbol{\delta}^{a} \big) \, \frac{\partial^2 \zeta^i}{\partial (\delta^{i})^2} \big( \boldsymbol{\delta}^{a} \big) }{ \big[ \frac{\partial \zeta^i}{\partial \delta^{i}} \big( \boldsymbol{\delta}^{a} \big) \big]^2} > 0
\end{equation*}
and
\begin{equation*}
     \frac{\partial \mathfrak{F}^{i, \m}}{\partial \delta^{k}} = \frac{\frac{\partial \zeta^i}{\partial \delta^{k}} \big( \boldsymbol{\delta}^{a} \big)}{\frac{\partial \zeta^i}{\partial \delta^{i}} \big( \boldsymbol{\delta}^{a} \big) } - \frac{\zeta^i \big( \boldsymbol{\delta}^{a} \big) \, \frac{\partial^2 \zeta^i}{\partial \delta^{i} \partial \delta^{k}} \big( \boldsymbol{\delta}^{a}\big)}{ \big[ \frac{\partial \zeta^i}{\partial \delta^{i}} \big( \boldsymbol{\delta}^{a} \big) \big]^2}
\end{equation*}
for $k \neq i$. Let us then compute the following gap:
\begin{equation*}
\begin{aligned}
    \frac{\partial \mathfrak{F}^{i, \m}}{\partial \delta^{i}} &- \sum_{k \neq i}  \Big| \frac{\partial \mathfrak{F}^{i, \m}}{\partial \delta^{k}} \Big| \\
    &= \Big( \frac{\partial \zeta^i}{\partial \delta^{i}} \Big)^{-2} \, \bigg\{ 2 \, \Big(\frac{\partial \zeta^i}{\partial \delta^{i}} \Big)^2 - \zeta^i \, \frac{ \partial^2 \zeta^i }{\partial (\delta^{i})^2} - \sum_{k \neq i} \Big| \frac{\partial \zeta^i}{\partial \delta^{i}} \, \frac{\partial \zeta^i}{\partial \delta^{k}} - \zeta^i \, \frac{\partial^2 \zeta^i}{\partial \delta^{i} \partial \delta^{k}} \Big| \bigg\}.
\end{aligned}
\end{equation*}
The gap obtained above is bounded away from $0$ according to the definition of $\zeta^i$. Therefore, by the implicit function theorem, we know there exists a unique Lipschitz mapping $\psi^{a}: \mathbb{R}^{\tilde{N}} \to \mathbb{R}^{\tilde{N}}$ such that \eqref{paper 3 maker implicit function} is equivalent to
\begin{equation*}
   \boldsymbol{\delta}^a = \psi^a (\boldsymbol{y}^\m).
\end{equation*}
The bid side can be discussed in a similar way and we obtain the (unique) Lipschitz mapping $\psi^{b}: \mathbb{R}^{\tilde{N}} \to \mathbb{R}^{\tilde{N}}$ such that $\boldsymbol{\delta}^b = \psi^b( \boldsymbol{y}^\m )$. 

While the concavity in $q^{i, \m}$ is clear, the concavity in $(\delta^{i, a}, \delta^{i, b})$ is not ensured by the definition of $\zeta^i$. However, the separation of the state $q^{i, \m}$ from the controls $(\delta^{i, a}, \delta^{i, b})$ in the Hamiltonian \eqref{approx hamilton} still enables us to apply the maximum principle; see the proof in \cite{guo2023macroscopic} or \cite{guo2024macroscopicmarketmakinggames}. By the stochastic maximum principle, the adjoint process $\boldsymbol{Y}^\m$ solves the FBSDE
\begin{equation*}
    \left\{
    \begin{aligned}
     dQ_t^{i, \m} &= - \tilde{a}_t \, \zeta^i \big( \psi^{a}(\boldsymbol{Y}_t^\m) \big) \, dt + \tilde{b}_t \, \zeta^i \big( \psi^{b}(\boldsymbol{Y}_t^\m) \big) \, dt + \epsilon \, dW_t^{i, \m},\\
    \, dY_t^{i, \m} &= 2\phi_t^{i, \m} \, Q_t^{i, \m} \, dt + dM_t^{i, \m},\\
    Q_0^{i, \m} &= q_0^{i, \m}, \quad Y_T^{i, \m} = -2A^{i, \m} \, Q_T^{i, \m},
    \end{aligned}
    \right.
\end{equation*}
which leads to the equilibrium profile $(\boldsymbol{\delta}^a_t, \boldsymbol{\delta}^b_t)_{t} = (\psi^a (\boldsymbol{Y}_t^\m), \psi^b (\boldsymbol{Y}_t^\m))_{t}$. Here, we assume the global solution exists, which will be proved soon.

The rest follows the proof of Theorem \ref{paper 3 big maximum principle}. Provided with the quoting strategy profile $(\psi^a (\boldsymbol{Y}_t^\m), \psi^b (\boldsymbol{Y}_t^\m))_{t}$, the traders engage in an optimal execution game. There exists a Lipschitz mapping $\varphi: \mathbb{R}^N \times \mathbb{R}^{\tilde{N}} \to \mathbb{R}^N$ such that the equilibrium trading profile satisfies
\begin{equation*}
    v_t^{i,\e} = \varphi^i(\boldsymbol{Y}_t^\e, \boldsymbol{Y}_t^\m),
\end{equation*}
where the adjoint process $\boldsymbol{Y}^{\e}$ solves the FBSDE
\begin{equation*}
        \left\{
    \begin{aligned}
     dQ_t^{i, \e} &= \varphi^i(\boldsymbol{Y}_t^\e, \boldsymbol{Y}_t^\m) \, dt + \epsilon \, dW_t^{i, \e},\\
    \, dY_t^{i, \e} &= 2\phi_t^{i, \e} \, Q_t^{i, \e} \, dt + dM_t^{i, \e},\\
    Q_0^{i, \e} &= q_0^{i, \e}, \quad Y_T^{i, \e} = -2A^{i, \e} \, Q_T^{i, \e},
    \end{aligned}
    \right.
\end{equation*}
for all $i$. Hence, the necessary part is obtained via a fixed-point argument. The sufficient part proceeds along the lines of the proof of Theorem \ref{paper 3 big maximum principle}.
\end{proof}

\kong

\subsection{Proof of Theorem \ref{sec 5 well pose}}

\begin{proof}
Let us consider \eqref{paper 3 noise sde} as a $\hat{N}$-dimensional trivial FBSDE, where the backward process is simply zero. Then, the extended system \eqref{paper 3 noise sde}-\eqref{paper 3 big FBSDE optimal execution approx} is Markovian. The complexity of our system lies in the forward equation. For the market making system \eqref{paper 3 big FBSDE market maker approx}, the drift caused by the ask side is the product of
\begin{equation*}
    -\hat{a}_t \text{ \; with \; } \zeta^i \big(\psi^{a}(\boldsymbol{Y_t}^\m)\big),
\end{equation*}
where $\hat{a}_t = \kappa^a(L_t) + \sum_{i = 1}^N \varphi^{i} \big( \boldsymbol{Y}^\e_t, \boldsymbol{Y}^\m_t \big) \, \mathbb{I}(q_0^{i,e} < 0)$. Precisely, process $\hat{a}$ can be regarded as a function of $L, \boldsymbol{y}^\m$, and $\boldsymbol{y}^\e$. Based on Lipschitz property of $\kappa^a$ and $\varphi$, it follows that $\hat{a}$ is Lipschitz with respect to $L$, $\boldsymbol{y}^\m$, and $\boldsymbol{y}^\e$. On the other hand, in spite of unbounded derivatives, the truncation in the definition of $\psi^a$ allows us to treat $\zeta^i$ as a Lipschitz function. Consequently, function $\zeta^i (\psi^{a}(\boldsymbol{y}^\m))$ is Lipschitz in $\boldsymbol{y}^\m$. Moreover, the truncation $\tilde{\xi}$ and boundedness of $\kappa^a$ infers that $\hat{a}$ is also bounded; similarly the truncation $\xi$ guarantees the boundedness of $\zeta^i (\psi^{a}(\boldsymbol{y}^\m))$. Being aware that the product of bounded Lipschitz functions is still Lipschitz, we obtain the Lipschitz property of the ask-side drift with respect to $L, \boldsymbol{y}^\m$, and $\boldsymbol{y}^\e$. After a similar discussion on the bid side, the forward equation of \eqref{paper 3 big FBSDE market maker approx} is hence Lipschitz. Part of previous discussion already implies the forward equation of \eqref{paper 3 big FBSDE optimal execution approx} is Lipschitz regarding both $\boldsymbol{y}^\e$ and $\boldsymbol{y}^\m$. 

Introduce the extended forward process $\bdX_t := (L_t, \boldsymbol{Q}_t^\e, \boldsymbol{Q}_t^\m)$ and backward process $\bdY_t := ( \vec{0}, \bdY^\e_t, \bdY^\m_t)$. Define the following vector-valued functions and matrix-valued functions:
\begin{equation*}
\begin{aligned}
    g^\m (L, \boldsymbol{y}^\e, \boldsymbol{y}^\m) &= 
    \begin{pmatrix}
        -\hat{a} \, \zeta^1 \big( \psi^{a}(\boldsymbol{y}^\m) \big) + \hat{b} \, \zeta^1 \big( \psi^{b}(\boldsymbol{y}^\m) \big)\\
        \vdots \\
        -\hat{a} \, \zeta^{\tilde{N}} \big( \psi^{a}(\boldsymbol{y}^\m) \big) + \hat{b} \, \zeta^{\tilde{N}} \big( \psi^{b}(\boldsymbol{y}^\m) \big)\\
    \end{pmatrix}, \quad
    g (t, L, \boldsymbol{y}^\e, \boldsymbol{y}^\m) =
    \begin{pmatrix}
        \Gamma(t, L) \\
        \varphi(\boldsymbol{y}^\e, \boldsymbol{y}^\m) \\
        g^\m (L, \boldsymbol{y}^\e, \boldsymbol{y}^\m)
    \end{pmatrix},\\
    \sigma(t, L) &= \begin{bmatrix}
        \Sigma(t, L) & \boldsymbol{0}\\
        \boldsymbol{0} & \epsilon \, I_{N + \tilde{N}}
    \end{bmatrix},\\
    k^\m(t) & = \begin{bmatrix}
    \phi_t^{1,\m} & &\\
     & \ddots & \\
    & & \phi_t^{\tilde{N},\m}
  \end{bmatrix}, \;
  k^\e(t) = \begin{bmatrix}
    \phi_t^{1,\e} & &\\
     & \ddots & \\
    & & \phi_t^{N,\e}
  \end{bmatrix}, \; k(t) = \begin{bmatrix}
    \boldsymbol{0} & & \\
    & k^\e(t) & \\
    & & k^\m(t)
  \end{bmatrix},\\
  h^\m & = \begin{bmatrix}
    A^{1,\m} & &\\
     & \ddots & \\
    & & A^{\tilde{N},\m}
  \end{bmatrix}, \; 
  h^\e = \begin{bmatrix}
    A^{1,\e} & &\\
     & \ddots & \\
    & & A^{N,\e}
  \end{bmatrix}, \; 
  h = \begin{bmatrix}
    \boldsymbol{0} & & \\
    & h^\e & \\
    & & h^\m
  \end{bmatrix}, 
\end{aligned}
\end{equation*}
where $\boldsymbol{0}$ and $I_{N + \tilde{N}}$ are respectively zero and identity matrices of proper dimensions. Our FBSDE system can be concisely represented by
\begin{equation}
    \left\{
    \begin{aligned}
     d\bdX_t &= \sigma(t, \bdX_t) \, \sigma^{-1}(t, \bdX_t) \, g(t,\bdX_t, \bdY_t) + \sigma(t, \bdX_t) \, dW_t, \\
    \, d\bdY_t &= 2k(t) \, \bdX_t \, dt + d\boldsymbol{M}_t,\\
    \bdX_0 &= \bdx_0, \quad \bdY_T = -2 h\, \bdX_T.
    \end{aligned}
    \right.
    \label{paper 3 non-deg concise FBSDE}
\end{equation}

We intend to apply Theorem 2.3 in \cite{nam2022coupled} for the well-posedness of \eqref{paper 3 non-deg concise FBSDE}. For the existence of solutions, we check the condition (F1) and (B1). The condition (F1) is clear based on the Lipschitz property of function $\sigma$. The following observations:
\kong
\begin{itemize}
    \item matrices $h$ and $k(t)$ are bounded linear coefficients;\\
    \vspace{-0.2cm}
    \item function $g$ is bounded and continuous in $\bdy$,
\end{itemize}
\kong
ensure condition (B1). Regarding uniqueness, we look at Lemma 4.1 in \cite{nam2022coupled} and its proof. In short, the lemma states that, under assumptions (H1)-(H4)---which can be implied by the condition (F1) and (B1) in Theorem 2.3---the main FBSDE admits a strong solution. Moreover, if equation (4.14) in \cite{nam2022coupled} has a unique strong solution for every It\^{o} process $I$, then the solution of the main FBSDE is unique. In the proof, the It\^{o} process $I$ is shown to solve (4.10). In our setting the process $I$ solves a standard Lipschitz SDE. Therefore, the integrability of $I$ together with the Lipschitz property of the coefficients in (4.14) ensure the existence of a unique solution to (4.14), and hence the uniqueness of the FBSDE solution.
\end{proof}

\vspace{0.2cm}
\sloppy
\printbibliography

@article{almgren2001optimal,
  title={Optimal execution of portfolio transactions},
  author={Almgren, Robert and Chriss, Neil},
  journal={Journal of Risk},
  volume={3},
  pages={5--40},
  year={2001}
}

@article{bertsimas1998optimal,
  title={Optimal control of execution costs},
  author={Bertsimas, Dimitris and Lo, Andrew W},
  journal={Journal of financial markets},
  volume={1},
  number={1},
  pages={1--50},
  year={1998},
  publisher={Elsevier}
}

@article{cartea2017algorithmic,
  title={Algorithmic trading with model uncertainty},
  author={Cartea, {\'A}lvaro and Donnelly, Ryan and Jaimungal, Sebastian},
  journal={SIAM Journal on Financial Mathematics},
  volume={8},
  number={1},
  pages={635--671},
  year={2017},
  publisher={SIAM}
}

@article{delarue2002existence,
  title={On the existence and uniqueness of solutions to FBSDEs in a non-degenerate case},
  author={Delarue, Fran{\c{c}}ois},
  journal={Stochastic processes and their applications},
  volume={99},
  number={2},
  pages={209--286},
  year={2002},
  publisher={Elsevier}
}

@article{fromm2013existence,
  title={Existence, uniqueness and regularity of decoupling fields to multidimensional fully coupled FBSDEs},
  author={Fromm, Alexander and Imkeller, Peter},
  journal={arXiv preprint arXiv:1310.0499},
  year={2013}
}

@article{campi2020optimal,
  title={Optimal market making under partial information with general intensities},
  author={Campi, Luciano and Zabaljauregui, Diego},
  journal={Applied Mathematical Finance},
  volume={27},
  number={1-2},
  pages={1--45},
  year={2020},
  publisher={Taylor \& Francis}
}

@article{ho1981optimal,
  title={Optimal dealer pricing under transactions and return uncertainty},
  author={Ho, Thomas and Stoll, Hans R},
  journal={Journal of Financial economics},
  volume={9},
  number={1},
  pages={47--73},
  year={1981},
  publisher={Elsevier}
}

@article{carmona2023optimal,
  title={Optimal execution with quadratic variation inventories},
  author={Carmona, Rene and Leal, Laura},
  journal={SIAM Journal on Financial Mathematics},
  volume={14},
  number={3},
  pages={751--776},
  year={2023},
  publisher={SIAM}
}

@article{said2017market,
  title={Market impact: A systematic study of limit orders},
  author={Said, Emilio and Ayed, Ahmed Bel Hadj and Husson, Alexandre and Abergel, Fr{\'e}d{\'e}ric},
  journal={Market microstructure and liquidity},
  volume={3},
  number={03n04},
  pages={1850008},
  year={2017},
  publisher={World Scientific}
}

@article{han2018solving,
  title={Solving high-dimensional partial differential equations using deep learning},
  author={Han, Jiequn and Jentzen, Arnulf and E, Weinan},
  journal={Proceedings of the National Academy of Sciences},
  volume={115},
  number={34},
  pages={8505--8510},
  year={2018},
  publisher={National Academy of Sciences}
}

@article{guo2024macroscopicmarketmakinggames,
    title={Macroscopic market making games},
      author={Guo, Ivan and Jin, Shijia},
      journal={Mathematical Finance},
      volume={36},
      number={2},
      pages={352--373},
      year={2026},
      publisher={Wiley Online Library}
}

@article{almgren2003optimal,
  title={Optimal execution with nonlinear impact functions and trading-enhanced risk},
  author={Almgren, Robert F},
  journal={Applied mathematical finance},
  volume={10},
  number={1},
  pages={1--18},
  year={2003},
  publisher={Taylor \& Francis}
}

@article{cardaliaguet2018mean,
  title={Mean field game of controls and an application to trade crowding},
  author={Cardaliaguet, Pierre and Lehalle, Charles-Albert},
  journal={Mathematics and Financial Economics},
  volume={12},
  number={3},
  pages={335--363},
  year={2018},
  publisher={Springer}
}

@article{huang2019mean,
  title={Mean-field game strategies for optimal execution},
  author={Huang, Xuancheng and Jaimungal, Sebastian and Nourian, Mojtaba},
  journal={Applied Mathematical Finance},
  volume={26},
  number={2},
  pages={153--185},
  year={2019},
  publisher={Taylor \& Francis}
}

@article{avellaneda2008high,
  title={High-frequency trading in a limit order book},
  author={Avellaneda, Marco and Stoikov, Sasha},
  journal={Quantitative Finance},
  volume={8},
  number={3},
  pages={217--224},
  year={2008},
  publisher={Taylor \& Francis}
}

@article{lehalle2019incorporating,
  title={Incorporating signals into optimal trading},
  author={Lehalle, Charles-Albert and Neuman, Eyal},
  journal={Finance and Stochastics},
  volume={23},
  pages={275--311},
  year={2019},
  publisher={Springer}
}

@article{luo2021dynamic,
  title={Dynamic Equilibrium of Market Making with Price Competition},
  author={Luo, Jialiang and Zheng, Harry},
  journal={Dynamic Games and Applications},
  volume={11},
  number={3},
  pages={556--579},
  year={2021},
  publisher={Springer}
}

@article{jusselin2021optimal,
  title={Optimal market making with persistent order flow},
  author={Jusselin, Paul},
  journal={SIAM Journal on Financial Mathematics},
  volume={12},
  number={3},
  pages={1150--1200},
  year={2021},
  publisher={SIAM}
}

@article{gueant2017optimal,
  title={Optimal market making},
  author={Gu{\'e}ant, Olivier},
  journal={Applied Mathematical Finance},
  volume={24},
  number={2},
  pages={112--154},
  year={2017},
  publisher={Taylor \& Francis}
}

@book{horn1994topics,
  title={Topics in matrix analysis},
  author={Horn, Roger A and Johnson, Charles R},
  year={1994},
  publisher={Cambridge university press}
}

@article{cont2022dynamics,
  title={Dynamics of market making algorithms in dealer markets: Learning and tacit collusion},
  author={Cont, Rama and Xiong, Wei},
  journal={Mathematical Finance},
  year={2022},
  publisher={Wiley Online Library}
}

@article{casgrain2018mean,
  title={Mean field games with partial information for algorithmic trading},
  author={Casgrain, Philippe and Jaimungal, Sebastian},
  journal={arXiv preprint arXiv:1803.04094},
  year={2018}
}

@article{freiling2002survey,
  title={A survey of nonsymmetric Riccati equations},
  author={Freiling, Gerhard},
  journal={Linear algebra and its applications},
  volume={351},
  pages={243--270},
  year={2002},
  publisher={Elsevier}
}

@article{guo2023macroscopic,
  author  = {Guo, Ivan and Jin, Shijia and Nam, Kihun},
  title   = {Macroscopic Market Making},
  journal = {Finance and Stochastics},
  year    = {2026},
  note    = {Forthcoming},
}

@book{clarke1990optimization,
  title={Optimization and nonsmooth analysis},
  author={Clarke, Frank H},
  year={1990},
  publisher={SIAM}
}

@article{varah1975lower,
  title={A lower bound for the smallest singular value of a matrix},
  author={Varah, James M},
  journal={Linear Algebra and its applications},
  volume={11},
  number={1},
  pages={3--5},
  year={1975},
  publisher={Elsevier}
}

@article{carmona2018probabilistic,
  title        = {Probabilistic Theory of Mean Field Games with Applications, Vols. I--II},
  author       = {Carmona, Ren{\'e} and Delarue, Fran{\c{c}}ois},
  journal      = {Springer},
  year         = {2018},
  volume       = {I--II},
}

@article{ma2015well,
  title={On well-posedness of forward--backward SDEs---A unified approach},
  author={Ma, Jin and Wu, Zhen and Zhang, Detao and Zhang, Jianfeng},
  journal={The Annals of Applied Probability},
  volume={25},
  number={4},
  pages={2168--2214},
  year={2015},
  publisher={Institute of Mathematical Statistics}
}

@article{ankirchner2020optimal,
  title={Optimal position targeting via decoupling fields},
  author={Ankirchner, Stefan and Fromm, Alexander and Kruse, Thomas and Popier, Alexandre},
  journal={The Annals of Applied Probability},
  volume={30},
  number={2},
  pages={644--672},
  year={2020},
  publisher={Institute of Mathematical Statistics}
}

@article{peng1999fully,
  title={Fully coupled forward-backward stochastic differential equations and applications to optimal control},
  author={Peng, Shige and Wu, Zhen},
  journal={SIAM Journal on Control and Optimization},
  volume={37},
  number={3},
  pages={825--843},
  year={1999},
  publisher={SIAM}
}

@article{evangelista2020finite,
  title={On finite population games of optimal trading},
  author={Evangelista, David and Thamsten, Yuri},
  journal={arXiv preprint arXiv:2004.00790},
  year={2020}
}

@article{cartea2016incorporating,
  title={Incorporating order-flow into optimal execution},
  author={Cartea, {\'A}lvaro and Jaimungal, Sebastian},
  journal={Mathematics and Financial Economics},
  volume={10},
  pages={339--364},
  year={2016},
  publisher={Springer}
}

@article{voss2022two,
  title={A two-player portfolio tracking game},
  author={Vo{\ss}, Moritz},
  journal={Mathematics and Financial Economics},
  volume={16},
  number={4},
  pages={779--809},
  year={2022},
  publisher={Springer}
}

@article{neuman2023trading,
  title        = {Trading with the Crowd},
  author       = {Neuman, Eyal and Vo{\ss}, Moritz},
  journal      = {Mathematical Finance},
  year         = {2023},
  publisher    = {Wiley}
}

@article{schied2017state,
  title={A State-Constrained Differential Game Arising in Optimal Portfolio Liquidation},
  author={Schied, Alexander and Zhang, Tao},
  journal={Mathematical Finance},
  volume={27},
  number={3},
  pages={779--802},
  year={2017},
  publisher={Wiley Online Library}
}

@article{nam2022coupled,
  title={Coupled FBSDEs with measurable coefficients and its application to parabolic PDEs},
  author={Nam, Kihun and Xu, Yunxi},
  journal={Journal of Mathematical Analysis and Applications},
  volume={515},
  number={1},
  pages={126403},
  year={2022},
  publisher={Elsevier}
}

@article{peng1992stochastic,
  title={Stochastic Hamilton-Jacobi-Bellman equations},
  author={Peng, Shige},
  journal={SIAM Journal on Control and Optimization},
  volume={30},
  number={2},
  pages={284--304},
  year={1992},
  publisher={SIAM}
}

@article{drapeau2019fbsde,
  title={An FBSDE approach to market impact games with stochastic parameters},
  author={Drapeau, Samuel and Luo, Peng and Schied, Alexander and Xiong, Dewen},
  journal={Probability, Uncertainty and Quantitative Risk},
  volume={6},
  number={3},
  pages={237--260},
  year={2021},
  publisher={Probability, Uncertainty and Quantitative Risk}
}

@article{carmona2016lectures,
  title        = {Lectures on BSDEs, Stochastic Control, and Stochastic Differential Games with Financial Applications},
  author       = {Carmona, Ren{\'e}},
  journal      = {SIAM},
  year         = {2016}
  }

\end{document}